\documentclass[letterpaper,11pt]{article}

\usepackage{palatino}
\usepackage[paperwidth=8.5in,paperheight=11.0in,
  left=1.0in,right=1.0in,top=1.0in,bottom=1.0in,
  includefoot,heightrounded]{geometry}
  \usepackage{amsmath,amssymb,amsthm}

\title{A New Approach to Arguments of Quantum Knowledge}
\author{James Bartusek \thanks{Columbia University.}
\and
Ruta Jawale\thanks{UIUC.}
\and
Justin Raizes\thanks{NTT Research.}
\and
Kabir Tomer\thanks{UIUC. Part of this work was done while the authors were visiting the Simons Institute for the Theory of
Computing}
}
\date{}

\newcommand{\ruta}[1]{ }
\newcommand{\james}[1]{ }
\newcommand{\justin}[1]{ }

\newcommand{\kabir}[1]{ }

\usepackage[dvipsnames]{xcolor}
\usepackage{amsthm}
\usepackage{mathtools}
\usepackage{natbib}
\usepackage{graphicx}

\usepackage{parskip}

\usepackage{todonotes}
\usepackage{graphicx}
\usepackage{amsfonts}
\usepackage{amsthm}
\usepackage{framed}
\usepackage{amsmath}
\usepackage{amssymb}
\usepackage{hyperref}

\newcommand{\negpar}[1][-1em]{%
  \ifvmode\else\par\fi
  {\parindent=#1\leavevmode}\ignorespaces
}

\usepackage{tabularx}
\usepackage{caption}
\usepackage{subcaption}
\hypersetup{
    colorlinks=true,
    citecolor=Rhodamine,
    linkcolor=Rhodamine,
    filecolor=Rhodamine,      
    urlcolor=blue,
}

\usepackage[]{algorithm2e}
\usepackage{mathtools}
\usepackage{natbib}

\usepackage{varioref}
\usepackage{hyperref}
\usepackage{cleveref}

\crefname{figure}{Figure}{Figures}
\crefname{equation}{Equation}{Equations}

\newtheorem{theorem}{Theorem}[section]
\newtheorem{construction}[theorem]{Construction}
\crefname{theorem}{Theorem}{Theorems}

\newtheorem{informaltheorem}[theorem]{Informal Theorem}
\crefname{informaltheorem}{Informal Theorem}{Informal Theorems}

\theoremstyle{definition}
\newtheorem{definition}[theorem]{Definition}
\crefname{definition}{Definition}{Definitions}

\newtheorem{corollary}[theorem]{Corollary}
\crefname{corollary}{Corollary}{Corollaries}

\crefname{conjecture}{Conjecture}{Conjectures}

\crefname{proposition}{Proposition}{Propositions}

\newtheorem{lemma}[theorem]{Lemma}
\crefname{lemma}{Lemma}{Lemmas}

\newtheorem{claim}[theorem]{Claim}
\crefname{claim}{Claim}{Claims}

\crefname{fact}{Fact}{Facts}

\crefname{section}{Section}{Sections}
\crefname{construction}{Construction}{Constructions}
\crefname{definition}{Definition}{Definitions}

\newtheoremstyle{myremark}     {10pt}{10pt}{}{}{\scshape}{.}{.5em}{}
\theoremstyle{myremark}
\newtheorem{remark}{Remark}[section]
\Crefname{remark}{Remark}{Remarks}

\newcommand{\poly}{\ensuremath{\mathsf{poly}}\xspace}

\newcommand{\negl}{\ensuremath{\mathsf{negl}}\xspace}

\newcommand{\Eval}{\ensuremath{\mathsf{Eval}}\xspace}

\newcommand{\crs}{\ensuremath{\mathsf{crs}}}

\newcommand{\Gen}{\ensuremath{\mathsf{Gen}}\xspace}

\newcommand{\PRO}{\ensuremath{\mathsf{PrO}}\xspace}
\newcommand{\QPRO}{\ensuremath{\mathsf{QPrO}}\xspace}
\newcommand{\Samp}{\ensuremath{\mathsf{Samp}}\xspace}

\newcommand{\ObfM}{\ensuremath{\widetilde{\cM}}\xspace}
\newcommand{\ObfC}{\ensuremath{\widetilde{C}\xspace}}
\newcommand{\ObfV}{\ensuremath{\widetilde{\cV}}\xspace}

\newcommand{\CSA}{\ensuremath{\mathsf{CSA}}\xspace}
\newcommand{\Obf}{\ensuremath{\mathsf{Obf}}\xspace}
\newcommand{\SimGen}{\ensuremath{\mathsf{SimGen}}\xspace}
\newcommand{\SimObf}{\ensuremath{\mathsf{SimObf}}\xspace}

\newcommand{\JLLWObf}{\ensuremath{\mathsf{JLLWObf}}\xspace}
\newcommand{\Expand}{\ensuremath{\mathsf{Expand}}\xspace}
\newcommand{\flag}{\ensuremath{\mathsf{flag}}\xspace}
\newcommand{\info}{\ensuremath{\mathsf{info}}\xspace}
\newcommand{\normal}{\ensuremath{\mathsf{normal}}\xspace}
\newcommand{\otp}{\ensuremath{\mathsf{otp}}\xspace}
\newcommand{\hybsf}{\ensuremath{\mathsf{hyb}}\xspace}
\newcommand{\simsf}{\ensuremath{\mathsf{sim}}\xspace}

\newcommand{\ZXVer}{\mathsf{ZXVer}}

\newcommand{\Had}{\ensuremath{\mathsf{Had}}\xspace}
\newcommand{\List}{\ensuremath{\mathsf{List}}\xspace}
\newcommand{\Ext}{\ensuremath{\mathsf{Ext}}}

\newcommand{\xor}{\oplus}
\newcommand{\tensor}{\otimes}

\allowdisplaybreaks[1]

\DeclareMathOperator*{\expectation}{\mathbb{E}}
\newcommand{\E}{\expectation}

\newcommand{\chal}{\mathsf{chal}}
\newcommand{\Open}{\mathsf{Open}}
\newcommand{\Opened}{\Open(\chal)}
\newcommand{\Obfr}{\widetilde{r}}
\newcommand{\NP}{\mathsf{NP}}
\newcommand{\QMA}{\mathsf{QMA}}

\newcommand{\com}{\mathsf{com}}

  \newcommand{\cA}{\ensuremath{{\mathcal A}}\xspace}
  \newcommand{\cB}{\ensuremath{{\mathcal B}}\xspace}
  \newcommand{\cC}{\ensuremath{{\mathcal C}}\xspace}
  \newcommand{\cD}{\ensuremath{{\mathcal D}}\xspace}

  \newcommand{\cH}{\ensuremath{{\mathcal H}}\xspace}
  \newcommand{\cI}{\ensuremath{{\mathcal I}}\xspace}
  \newcommand{\cJ}{\ensuremath{{\mathcal J}}\xspace}
  \newcommand{\cK}{\ensuremath{{\mathcal K}}\xspace}
  \newcommand{\cM}{\ensuremath{{\mathcal M}}\xspace}
  
  \newcommand{\cO}{\ensuremath{{\mathcal O}}\xspace}
  \newcommand{\cP}{\ensuremath{{\mathcal P}}\xspace}
  \newcommand{\cQ}{\ensuremath{{\mathcal Q}}\xspace}
  \newcommand{\cR}{\ensuremath{{\mathcal R}}\xspace}
  \newcommand{\cS}{\ensuremath{{\mathcal S}}\xspace}

  \newcommand{\cV}{\ensuremath{{\mathcal V}}\xspace}

\newcommand{\zo}{\ensuremath{{\{0,1\}}}\xspace}

\newcommand{\secpar}{\lambda}

\newcommand{\st}{\: : \:}
\newcommand{\cL}{\ensuremath{{\mathcal L}}\xspace}
\newcommand{\Dec}{\mathsf{Dec}}
\newcommand{\Enc}{\mathsf{Enc}}
\newcommand{\KeyGen}{\mathsf{KeyGen}}
\newtheorem{property}{Property}[section]
\usepackage{setspace}

\newcommand{\sP}{\mathsf{P}}
\newcommand{\sV}{\mathsf{V}}

\newcommand{\Sim}{\mathsf{Sim}}

\newcommand{\Com}{\mathsf{Com}}

\newcommand{\td}{\mathsf{td}}

\newcommand{\bbE}{\mathbb{E}}

\newcommand{\bbR}{\mathbb{R}}
\newcommand{\bbF}{\mathbb{F}}
\newcommand{\bbN}{\mathbb{N}}

\usepackage{hyperref}
\usepackage{amsmath}
\usepackage{array}
\usepackage{amssymb}
\usepackage{amsthm}
\usepackage{amsfonts}
\usepackage{bm}
\usepackage{xspace}
\usepackage{float}
\usepackage{graphicx}
\usepackage{hyperref}
\usepackage{enumitem}
\usepackage{color}
\usepackage{comment}
\usepackage{appendix}
\usepackage{mathtools}
\usepackage{bm}

\usepackage[utf8]{inputenc}
\usepackage{balance}
\usepackage{csquotes}
\usepackage{microtype}
\usepackage{breqn}
\usepackage{titlesec}

\usepackage{arydshln}
\usepackage{booktabs}
\usepackage{nicefrac}
\usepackage{enumitem}
\PassOptionsToPackage{hyphens}{url}
\usepackage{url}

\newcommand{\commentout}[1]{}

\newcommand{\Ver}{\ensuremath{\mathsf{Ver}}\xspace}

\newcommand{\NIZK}{\ensuremath{\mathsf{NIZK}}\xspace}

\newcommand{\Verify}{\ensuremath{\mathsf{Verify}}\xspace}
\newcommand{\Setup}{\ensuremath{\mathsf{Setup}}\xspace}
\renewcommand{\Open}{\ensuremath{\mathsf{Open}}\xspace}
\newcommand{\Prove}{\ensuremath{\mathsf{Prove}}\xspace}
\newcommand{\Exp}{\ensuremath{\mathsf{Exp}}\xspace}
\newcommand{\Sup}{\ensuremath{\mathsf{Sup}}\xspace}
\newcommand{\Hyb}{\ensuremath{\mathsf{Hyb}}\xspace}

\newcommand{\pp}{\ensuremath{\mathsf{pp}}\xspace}
\newcommand{\sk}{\ensuremath{\mathsf{sk}}\xspace}
\newcommand{\pk}{\ensuremath{\mathsf{pk}}\xspace}

\newcommand{\ct}{\ensuremath{\mathsf{ct}}\xspace}

\newcommand{\Combine}{\ensuremath{\mathsf{Combine}}\xspace}
\newcommand{\secp}{\ensuremath{\lambda}\xspace}

\newcommand{\ZK}{\ensuremath{\mathsf{ZK}}\xspace}
\newcommand{\OBF}{\ensuremath{\mathsf{OBF}}\xspace}

\usepackage[nointegrals]{wasysym}

\usepackage{array,multirow,booktabs,adjustbox,graphicx}
\newcommand{\ignore}[1]{}

\usepackage{footnote}
\makesavenoteenv{figure}

\usepackage{physics}
\usepackage{tikz}
\usetikzlibrary{quantikz2}

\newcommand{\identity}{I}
\newcommand*{\defeq}{\mathrel{\vcenter{\baselineskip0.5ex \lineskiplimit0pt \hbox{\scriptsize.}\hbox{\scriptsize.}}}%
=}

\usepackage{multirow}

\usepackage{scrextend}

\DeclareMathOperator{\expect}{\mathbb{E}}
\newcommand{\concat}{\|}

\newcommand{\rand}{\$\$}

\newcommand{\Keys}{\mathsf{Keys}}
\newcommand{\Handles}{\mathsf{Handles}}
\newcommand{\FE}{\mathsf{FE}}
\newcommand{\pred}{\varphi}

\newcommand{\PRP}{\mathsf{PRP}}
\newcommand{\Bloat}{\mathsf{Bloat}}

\usepackage{mdframed}

\newcommand{\difference}[1]{{\color{red} #1}}
\newcommand{\diff}[1]{\difference{#1}}

\usepackage[normalem]{ulem}

\newcommand{\ATI}{\mathsf{ATI}}
\newcommand{\Accept}{\mathsf{Accept}}
\newcommand{\Reject}{\mathsf{Reject}}
\newcommand{\TD}{\mathsf{TD}}

\begin{document}
\pagebreak

\maketitle

\begin{abstract}

    We construct a publicly-verifiable non-interactive zero-knowledge argument system for QMA with the following properties of interest.

\begin{itemize}
    \item \textbf{Transparent setup.} Our protocol only requires a uniformly random string (URS) setup. The only prior publicly-verifiable NIZK for QMA (Bartusek and Malavolta, ITCS 2022) requires an \emph{entire obfuscated program} as the common reference string.
    \item \textbf{Extractability.} Valid QMA witnesses can be extracted directly from our accepting proofs. That is, we obtain a publicly-verifiable non-interactive argument of \emph{quantum knowledge}, which was previously only known in a privately-verifiable setting (Coladangelo, Vidick, and Zhang, CRYPTO 2020).
\end{itemize}

Our construction introduces a novel type of ZX QMA verifier with "strong completeness" and builds upon the coset state authentication scheme from (Bartusek, Brakerski, and Vaikuntanathan, STOC 2024) within the context of QMA verification. Along the way, we establish new properties of the authentication scheme.

The security of our construction rests on the heuristic use of a post-quantum indistinguishability obfuscator. Rather than rely on the full-fledged classical oracle model (i.e.\ ideal obfuscation), we isolate a particular game-based property of the obfuscator that suffices for our proof, which we dub the \emph{evasive composability} heuristic.

As an additional contribution, we study a general method for replacing heuristic use of obfuscation with heuristic use of hash functions in the post-quantum setting. In particular, we establish security of the ideal obfuscation scheme of Jain, Lin, Luo, and Wichs (CRYPTO 2023) in the \emph{quantum} pseudorandom oracle model (QPrO), which can be heuristically instantiated with a hash function. This gives us NIZK arguments of quantum knowledge for QMA in the QPrO, and additionally allows us to translate several quantum-cryptographic results that were only known in the classical oracle model to results in the QPrO.
\end{abstract}

\newpage

\tableofcontents

\section{Introduction}\label{sec:intro}

Inspecting the proof of a mathematical statement generally reveals significantly more information than the fact that the statement is true. Remarkably, \cite{BFM88} (building on an earlier \emph{interactive} system \cite{GMR}) demonstrated that this need not always be the case, using cryptography to produce convincing proofs of NP statements that reveal nothing beyond the validity of the statement. This was the first example of a \emph{non-interactive zero-knowledge} (NIZK) argument, which is now considered one of the most basic and natural cryptographic primitives with numerous applications throughout cryptography.

While it was established by \cite{GO94} that any non-interactive argument for a language outside of BPP requires some form of setup, the widespread utility of NIZKs rests on the \emph{public verifiability} of proofs with respect to the setup. That is, the setup algorithm must output some \emph{public} information $\crs$ that can be used by any user in order to verify a purported proof $\pi$. Hence, a NIZK argument for some language $\cL$ has the following syntax.

\begin{itemize}
    \item $\Setup(1^\secp) \to \crs$. The Setup algorithm produces some public parameters $\crs$.
    \item $\Prove(\crs,x) \to \pi$. The Prove algorithm takes $\crs$ and an instance $x \in \cL$, and produces a proof $\pi$ of this fact.
    \item $\Verify(\crs,x,\pi)$. The Verify algorithm takes $\crs$, an instance $x$, and a proof $\pi$, and either accepts or rejects.
\end{itemize}

Following \cite{BFM88}, the community has developed an impressive array of approaches (e.g.\ \cite{GO94,FLS,fischlin,GOS,SW,C:PeiShi19,JJ21,Waters24,WWW25}) to obtaining NIZK arguments for NP with security based upon various cryptographic assumptions.

\paragraph{Quantumly-provable statements.} Unfortunately, the situation changes dramatically if we expand the goal to proving mathematical statements in zero-knowledge whose proofs involve quantum information. That is, does there exist NIZK arguments for all of QMA? Perhaps surprisingly, given how well-understood NIZKs are in the classical setting, the existence of NIZKs for QMA from standard assumptions has remained wide open. 

In fact, there is currently only one result in the literature \cite{BM22} that proposes a candidate NIZK for QMA. This protocol builds on techniques from the area of \emph{classical verification of quantum computation} (CVQC). In particular, the public parameters $\crs$ contain an \emph{obfuscated program} that runs the CVQC verifier. In order to argue security, \cite{BM22} model this obfuscation as a classical oracle, an approach that can also be used to establish the existence of publicly-verifiable CVQC. However, the existence of both NIZKs and QMA and publicly-verifiable CVQC from standard assumptions is still unknown, and it has remained unclear whether these applications are fundamentally linked.

Instead, the community has developed several protocols for a significantly weaker notion of NIZK for QMA that requires some form of \emph{private} setup. For example, there exist several protocols from standard assumptions in the ``secret parameters model'', where we assume a trusted third-party that samples (potentially quantum) private randomness $r_P$ for the prover and $r_V$ for the verifier \cite{TCC:ACGH20, C:ColVidZha20,C:Shmueli21,C:BCKM21a,AC:MorYam22} (we note that some only require private randomness for the verifier, but not the prover). The community has also considered the ``shared EPR pair model'', where we assume that the prover and verifer begin the protocol with several shared EPR pairs \cite{AC:MorYam22,C:BarKhuSri23}. However, this model has the same drawbacks as the secret parameters model, as security relies on the assumption that these EPR pairs are set up correctly by an honest dealer, and they can only be used by the parties who receive them from the dealer. 

\paragraph{Transparent setup.} Even among the standard notion of publicly-verifiable NIZKs, there is an important distinction to make between \emph{transparent} and \emph{private-coin} setup. While a private-coin setup requires the Setup algorithm to sample $\crs$ from some structured distribution, a transparent setup only requires that $\crs$ is a uniformly random string. Thus, transparent setups are by far the easiest to realize in practice.  

Examples of transparent setups include the \emph{uniform reference string}, or URS, model,\footnote{Sometime this is referred to as a \emph{common random string}, or CRS, but this is often confused with the notion of a common \emph{reference} string, which may be structured.} and the \emph{random oracle model} (ROM). By now, we have several approaches for realizing NIZKs for all of NP with transparent setup (e.g.\ \cite{FLS,fischlin,C:PeiShi19}). However, as mentioned above, the only previously proposed NIZK for QMA \cite{BM22} has a highly structured $\crs$ that includes an obfuscated program. Thus, the following question has remain unaddressed.

\begin{center}
    \emph{Do there exist non-interactive zero-knowledge arguments for QMA \emph{with transparent setup}?}
\end{center}

\paragraph{Knowledge soundness.} In addition to minimizing the complexity of Setup, another important goal in the design of argument systems is to strengthen the soundness property. Traditionally, soundness guarantees that the prover cannot convince the verifier to accept a proof relative to any ``no'' instance. However, we often want to capture the idea that in order for a prover to produce a convincing proof (even of a ``yes'' instance), then they must \emph{possess} a valid witness. This is formalized by requiring that an accepting witness can be \emph{extracted} from any prover that manages to convince the verifier to accept its proof. Again, this property gives meaningful guarantees even when the statement to be proven is true, and has been broadly useful in the classical setting, for example in the area of anonymous credentials (e.g.\ \cite{CvH91,CL01}). 

While knowledge-soundness is again quite well-understood in the classical setting, we have far less convincing results in the quantum setting. The only non-interactive protocol that has been shown to have knowledge soundness is that of \cite{C:ColVidZha20}, which is in the secret parameters model. Thus, the following question has also been left unresolved.

\begin{center}
    \emph{Do there exist publicly-verifiable non-interactive zero-knowledge arguments \emph{of knowledge} for QMA?}
\end{center}

\subsection{Results}

In this work, we address both questions simultaneously by presenting a NIZK argument of quantum knowledge for QMA with transparent setup.

\begin{informaltheorem}
    Assuming any (post-quantum) NIZK argument of knowledge for NP with transparent setup, there exists a NIZK argument of quantum knowledge for QMA with transparent setup making heuristic use of a post-quantum obfuscator for classical computation.
\end{informaltheorem}

We observe that there exist post-quantum NIZK arguments of knowledge for NP in the URS model from LWE, by building on \cite{C:PeiShi19}. Thus, we obtain NIZK arguments of quantum knowledge for QMA in the URS model from standard cryptographic assumptions plus the heuristic use of a post-quantum obfuscator.

In fact, our protocol satisfies a very natural \emph{straightline extraction} property, where if the extractor programs the URS, it is then able to directly recover a witness from the prover's proof. Moreover, our protocol has no knowledge error, and we are able to extract a QMA witness of quality $1-\gamma$ for some arbitrarily small inverse polynomial $\gamma$ (see \cref{subsec:tech-protocol} for more discussion).

\paragraph{Evasive composability.} Next, we provide more detail about the heuristic manner in which we use the classical obfuscator. 

As discussed below in the technical overview, for much of the proof (of zero-knowledge) we use the standard indistinguishability property of the obfuscator. However, in one key step we resort to what we call the \emph{evasive composability heuristic}:

\begin{definition}[Evasive Composability Heuristic, simplified and informal]
    Let $\Obf$ be an obfuscator, and $\mathcal{S}$ be any ``non-contrived'' sampler that outputs two classical circuits $C_0,C_1$ along with some (potentially quantum) side information $\ket{\psi}$.\\ 
    
    \textbf{IF} for any QPT adversary $\mathcal{A}$ and each $b \in \{0,1\}$, it holds that
    \[\bigg| \Pr_{\ket{\psi},C_0,C_1 \gets \mathcal{S}}\left[\mathcal{A}(\ket{\psi},\Obf(C_b)) = 1\right] - \Pr_{\ket{\psi},C_0,C_1 \gets \mathcal{S}}\left[\mathcal{A}(\ket{\psi},\Obf(\mathsf{NULL})) = 1\right]\bigg| = \negl(\secp),\] \textbf{THEN} it holds that \[\bigg| \Pr_{\ket{\psi},C_0,C_1 \gets \mathcal{S}}\left[\mathcal{A}(\ket{\psi},\Obf(C_0 \| C_1)) = 1\right] - \Pr_{\ket{\psi},C_0,C_1 \gets \mathcal{S}}\left[\mathcal{A}(\ket{\psi},\Obf(\mathsf{NULL})) = 1\right]\bigg| = \negl(\secp),\] where $\mathsf{NULL}$ is the always-rejecting circuit, and $C_0 \| C_1$ is the ``composed'' circuit that maps $(b,x) \to C_b(x)$.
\end{definition}

Informally, this heuristic asserts that obfuscating a circuit composed of sub-circuits whose obfuscations are indistinguishable from null, is itself indistinguishable from null. While this appears to be quite reasonable for natural choice of samplers $\mathcal{S}$, we remark that there do exist contrived samplers (involving ``self-eating'' circuits) that violate the statement \cite{Barak01}. However, in our case we only need evasive composability to hold for a specific class of samplers, which do not seem to possess the self-referential properties used in the \cite{Barak01} impossibility result. Moreover, it is easy to see that this heuristic holds for \emph{any} choice of $\mathcal{S}$ in the classical oracle model, where $\Obf$ is modeled as a black-box. 

\paragraph{Comparison with \cite{BM22}.} As mentioned earlier, there is only one other candidate publicly-verifiable NIZK for QMA \cite{BM22}, and it is worth comparing our results a little more closely. We consider our approach to have three key benefits over \cite{BM22}.
\begin{itemize}
    \item Our protocol only requires a \emph{uniformly random string} setup, while \cite{BM22} requires a highly structured obfuscated program as the shared string.
    \item Our protocol satisfies a natural argument of quantum knowledge property with straightline extraction, which is not satisfied by \cite{BM22}.  
    \item While both approaches make heuristic use of classical obfuscation, we isolate our heuristic use to the evasive composability heuristic, while it is unclear how to do so with \cite{BM22}. We thus hope that our approach will yield more progress towards the long-standing goal of obtaining NIZKs for QMA from standard assumptions. In particular, our protocol does not rely on techniques that also yield publicly-verifiable CVQC (to the best of our knowledge), giving hope that NIZKs for QMA might be ``easier'' to construct than publicly-verifiable CVQC.
\end{itemize}

On the other hand, we remark that the \cite{BM22} arguments are \emph{classical} and \emph{succinct}, while ours are \emph{quantum}, and grow with the size of the witness, which allows us to establish the straightline knowledge extraction property. Thus, the results are strictly incomparable. Indeed, as we will see in the technical overview, our technical approach is completely different from that of \cite{BM22}.

\paragraph{The quantum pseudorandom oracle model.} As it stands, both known approaches (ours and \cite{BM22}) to NIZKs for QMA, as well as several other results in quantum cryptography, rely on the heuristic use of a classical obfuscator. However, if we are to rely on heuristic arguments, it is certainly in our interest to reduce the complexity of the object that we treat heuristically. For instance, it is a much more common practice in cryptography to make heuristic use of a simple \emph{hash function}, which can instantiated with cryptographic hashes such as SHA2 or SHA3. 

As a contribution of independent interest, we show that heuristic use of obfuscation in the \emph{post-quantum} setting can often be replaced with heuristic use of a hash function. To do so, we adapt the recently-introduced \emph{pseudorandom oracle}, or PrO, model \cite{JLLW23} to the post-quantum setting.

The pseudorandom oracle model is defined with respect to some pseudorandom function $\{f_k\}_k$. It internally samples a uniformly random permutation $\pi$ and presents the following interfaces:

    \begin{itemize}
        \item $\PRO(\Gen,k) \to \pi(k)$
        \item $\PRO(\Eval,h,x) \to f_{\pi^{-1}(h)}(x)$
    \end{itemize}

That is, one can generate \emph{handles} $\pi(k)$ corresponding to PRF keys $k$, which can be used to evaluate the function but reveal nothing about the key itself. As argued in \cite{JLLW23}, one can plausibly instantiate the PrO using a cryptographic hash function such as $\mathsf{SHA3}$, where $\PRO(\Gen,k) \to k$ and $\PRO(\Eval,k,x) \to \mathsf{SHA3}(k,x)$. Although there is clearly a mismatch with the idealized model in that the permutation $\pi$ is instantiated with the \emph{identity}, \cite{JLLW23} argue that this is justified based on the heuristic understanding that $\mathsf{SHA3}$ behaves likes a ``self-obfuscated'' PRF. Thus, just like the \emph{random} oracle model, we consider the pseudorandom oracle model to be a transparent setup assumption, as it can be plausibly instantiated using the public description of a cryptographic hash function.

Now, while \cite{JLLW23}'s main result was to show how to construct ideal obfuscation for classical circuits in the PrO model from functional encryption, they did not address the \emph{post-quantum} setting, where the PrO may be accessed in \emph{quantum superposition}. In this work, we fill that gap, and prove the following result of independent interest.

\begin{informaltheorem}
    Assuming sub-exponentially secure functional encryption, there exists (post-quantum) ideal obfuscation in the quantum pseudorandom oracle (QPrO) model.
\end{informaltheorem}

As corollaries, we obtain several results in the QPrO that were previously only known in the full-fledged classical oracle model (e.g.\ witness encryption for QMA and publicly-verifiable CVQC \cite{BM22}, copy-protection for all unlearnable functionalities \cite{C:ALLZZ21}, obfuscation for various classes of quantum circuits \cite{BKNY23,BBV24,HT25}, and quantum fire \cite{DBLP:journals/iacr/CakanGS25}).

Due to a non-black-box use of obfuscation in our NIZK argument, the above theorem doesn't immediately imply NIZKs of knowledge for QMA in the QPrO. However as we explain further in the technical overview, we do manage to show this result, encapsulated in the following theorem. 

\begin{informaltheorem}
    Assuming (post-quantum) NIZK arguments of knowledge for NP with transparent setup and (post-quantum) sub-exponentially secure functional encryption, there exists NIZK arguments of knowledge for QMA in the QPrO model with transparent setup.
\end{informaltheorem}

\section{Technical Overview}

\subsection{Approach}
 
From a bird's eye view, our approach is fairly natural: Given a $\QMA$ instance $x$, witness $\ket{\psi}$, and a $\QMA$ verification measurement $\mathcal{M}$, the proof consists of an appropriate ``encoding'' of $\ket{\psi}$, an appropriate ``obfuscation'' of the measurement $\mathcal{M}$, and a NIZK (for NP) argument that the obfuscation and encoding have been prepared honestly. 

However, it is not clear how to instantiate this approach. First of all, we don't currently have candidates for obfuscating arbitrary quantum measurements (let alone doing so in a provably-correct manner). Moreover, it is also in general unclear how to use a proof for NP (or even QMA!) to prove facts about quantum states, e.g.\ that the witness was encoded honestly.

Despite these obstacles, we show that a careful choice of the $\QMA$ verifier enables us to leverage techniques from obfuscation of quantum computation \cite{BBV24} to achieve a publicly-verifiable NIZK (of knowledge) for QMA. In particular, we use the \cite{BBV24} ``coset-state authentication'' scheme in order to encode the witness $\ket{\psi}$, which encodes each qubit according to the map
\[\CSA.\Enc: \ \ \ket{0} \to X^xZ^z\ket{S}, \ \ \ \ket{1} \to X^xZ^z\ket{S+\Delta},\] where $S$ is a random subspace of $\bbF_2^\secp$, and $x,z,\Delta$ are random vectors. This encoding scheme admits a classical circuit $\CSA.\Ver$ that can be used to \emph{verify} membership in the codespace, as well as classical circuits $\CSA.\Dec_0$ and $\CSA.\Dec_1$ that can be used to \emph{measure} the encoded state in the standard and Hadamard basis respectively.\footnote{Strictly speaking, these classical circuits must be applied in \emph{quantum superposition} in order to realize these verification and measurement functionalities.}

Before going further into the details, we sketch the high-level overview of our NIZK for $\QMA$. Let $\cM$ be a quantum verifier for a $\QMA$ promise problem $(\cL_{yes}, \cL_{no})$. We construct a protocol $(\Setup, \sP, \sV)$ with the following form.
\begin{itemize}
    \item $\Setup$ outputs the commom random string $\crs$ for a NIZK for $\NP$.
    \item $\sP$ takes input the $\crs$, an instance $x \in \cL_{yes}$, and a corresponding quantum witness $\ket{\psi}$. It outputs an encoding of the witness $\ket*{\widetilde{\psi}} = \CSA.\Enc(\ket{\psi})$, a classical obfuscation $\ObfV$ of the codespace membership tester $\CSA.\Ver$, a classical obfuscation of the quantum verifier $\ObfM$ (which will be derived from the $\CSA.\Dec_0$ and $\CSA.\Dec_1$ circuits in a manner described below), and a NIZK for $\NP$ proof $\pi$ that the obfuscations were prepared honestly.
    \item $\sV$ takes input the $\crs$, the instance $x$, the encoded witness $\ket*{\widetilde{\psi}}$, the obfuscations $\ObfV, \ObfM$, and proof $\pi$. It accepts iff (1) the NIZK for $\NP$ verifier accepts $(\crs, \pi)$, (2) the tester $\ObfV$ accepts $\ket*{\widetilde{\psi}}$, and (3) the quantum verifier $\ObfM$ accepts $\ket*{\widetilde{\psi}}$.
\end{itemize}

Before going further, we discuss some necessary background.

\paragraph{Background: ZX verifiers.} We recall~\cite{BL08,CM16,MNS16} that any $\QMA$ language can be verified using just standard and Hadamard basis measurements, followed by some classical post-processing. It will be convenient for us to describe such a verifier's behavior as a collection of \emph{coherent} ZX measurements on $n$ qubits, where each ZX measurement is specified by a sequence of bases $\theta \in \{0,1\}^n$ and a function $f: \{0,1\}^n \to \{0,1\}$. That is, a ZX measurement $M[\theta,f]$ is defined by the projector
\[M[\theta,f] \coloneqq H^\theta\left(\sum_{x:f(x) = 1}\ketbra{x}{x}\right)H^\theta,\]

and the $\QMA$ verifier is specified by some collection $\{M[\theta_i,f_i]\}_{i \in [N]}$ of ZX measurements. To verify a proof $\ket{\psi}$, it:

\begin{itemize}
    \item Samples $i \gets [N]$.
    \item Applies $\{\Pi[\theta_i,f_i], I - \Pi[\theta_i,f_i]\}$ to $\ket{\psi}$, and accepts if the measurement accepts.
\end{itemize}

Specifying the verifier in this manner is quite promising for instantiating our template above, as the $\CSA$ scheme supports ZX measurements on encoded states. However, note that the verifier doesn't only apply a ZX measurement -- it first must \emph{sample} the choice of measurement $i \gets [N]$ that it will perform. This seemingly innocuous sampling step actually introduces a subtle issue in finalizing the instantiation of our template.

\paragraph{Who samples the randomness?} One could imagine two ways to handle the sampling of $i \gets [N]$ in our NIZK for QMA. One strategy would have the \emph{prover} sample $i$, and then only send over an obfuscation of the $\CSA$ decoder for the particular choice of measurement $M[\theta_i,f_i]$. Unfortunately, this completely breaks soundness, as it may be possible for the prover to find some state $\ket{\psi}$ that is accepted by some \emph{fixed} measurement $M[\theta_i,f_i]$ (even if there is no state $\ket{\psi}$ that is accepted with high probability over the random choice of measurement).

Another strategy would be to have the prover obfuscate $\CSA$ decoders for the \emph{entire set} of ZX measurements, and then have the verifier choose which one to apply. Unfortunately, for traditional $\QMA$ verifiers (say, parallel repetition of XX/ZZ Hamiltonians), this completely breaks zero-knowledge, as there may exist two accepting witnesses $\ket{\psi_0}, \ket{\psi_1}$ and some choice of measurement $M[\theta_i,f_i]$ such that $M[\theta_i,f_i]$ accepts $\ket{\psi_0}$ and $\ket{\psi_1}$ with very different probabilities.

\paragraph{QMA verification with strong completeness.} We resolve this tension by sticking with the second choice, but explicitly designing a ZX verifier that does not have this issue. In particular, we say that a ZX verifier has \emph{strong completeness} if for every yes instance, there exists a witness $\ket{\psi}$ such that for \emph{all} choices of $i \in [N]$, it holds that 
\[\|M[\theta_i,f_i]\ket{\psi} \|^2 = 1-2^{-n},\] where $n$ is the instance size. We also require negligible soundness, where for any no instance and state $\ket{\psi}$,  \[
            \E_{i \gets [N]} \left[\| M[\theta_{i},f_{i}]\ket{\psi} \|^2\right] 
            = \negl(n).
        \]

That is, we boost the completeness guarantee to $1-2^{-n}$ for every choice of measurement, rather than just on average over the choice of measurement. Formally, we show that every promise problem in $\QMA$ has a ZX verifier with strong completeness, which may be of independent interest. We accomplish this with what we call a ``permuting QMA verifier''.

\paragraph{Permuting Verifier for QMA.} 
The permuting verifier is a modification of the standard parallel repetition amplification for QMA. 
Imagine that the verifier was given a large register which (allegedly) contained many copies of the same witness. 
Instead of sampling a measurement independently for each copy, the verifier starts with a fixed list of measurements containing each $M[\theta_i, f_i]$ many times, then permutes it randomly. Then, it applies the permuted list of measurements to the witness register and accepts if the majority of them accept.\footnote{Technically speaking, there is some weighting of the measurements involved, which we ignore here for the sake of exposition.}

If the verifier really were given many copies of the same witness, then the permutation does not matter; the verifier is just applying each $M[\theta_i, f_i]$ many times to the same state.
With a slight change in perspective, the verifier is simply estimating the outcome distribution of each measurement on the witness by performing it many times. This has high accuracy, so the verifier is convinced with overwhelming probability.

The case of soundness is more complicated. If each measurement were sampled independently at random, we could argue that each index constitutes its own QMA verifier, and so should reject with a fixed probability. Unfortunately, the measurements are highly correlated because there are a fixed number of each $M[\theta_i, f_i]$.

The saving grace is that the independent distribution is \emph{heavily} concentrated around its mean. We can view the independent distribution as first sampling the number of times to apply each $M[\theta_i, f_i]$, then permuting the resulting list.
The number of times each measurement appears in the independent list versus the fixed list is very close with high probability. If we were to replace each difference by a measurement that always accepts in the independent list, the number of modifications is very small relative to the overall size of the list. Finally, we can show that the number of accepting measurements can only increase by the number of replacements, which is not enough to bridge the gap between a NO instance and a YES instance. More details can be found in \cref{sec:qma-ver}.

\subsection{The protocol and analysis}\label{subsec:tech-protocol}

We are now ready to present our protocol in some more detail, and then give high-level, informal, overviews of our proofs of knowledge-soundness and zero-knowledge. 

Recall that we not only want to obfuscate the $\CSA$ algorithms, we also want to \emph{prove} correctness of these obfuscations. For reasons that we expand on below in \cref{subsec:QPRO-overview}, we define an abstraction called \emph{provably-correct obfuscation} which has all of the properties that we require. 

Let $\Obf$ be a provably-correct obfuscation with respect to public parameters $\crs$ (concretely, think of $\crs$ as the public parameters for some NIZK of knowledge for NP). Let $x$ be an instance, $\{M[\theta_i,f_i]\}_{i \in [N]}$ be the corresponding ZX verifier with strong completeness, and let $\ket{\psi}$ be a witness for $x$. Then our protocol operates as follows.

\begin{itemize}

    \item The prover $\sP$ takes input the $\crs$ and the witness $\ket{\psi}$. It computes $\ket*{\widetilde{\psi}} = \CSA.\Enc(\ket{\psi})$, sets $\ObfV = \Obf(\CSA.\Ver)$, and sets 
    \[\ObfV,\{\ObfM_i\}_{i \in [N]} = \Obf\left(\crs,\CSA.\Ver \| \{\CSA.\Dec_i\}_{i \in [N]}\right),\]

    where we are obfuscating the \emph{concatenation} of all $N+1$ programs, and parsing the resulting obfuscation as the part $\ObfV$ that can be used to evaluate $\CSA.\Ver$ and the parts $\ObfM_i$ that can be used to evaluate each $\CSA.\Dec_i$. Here, $\CSA.\Dec_i$ refers to the $\CSA$ algorithm that measures the encoded state according to $M[\theta_i,f_i]$.

    \item $\sV$ takes input the $\crs$, the instance $x$, the encoded witness $\ket*{\widetilde{\psi}}$, and the obfuscation $(\ObfV,\{\ObfM_i\}_i)$. It first checks that the obfuscation is well-formed using $\crs$. Then, it checks that applying $\ObfV$ to $\ket*{\widetilde{\psi}}$ accepts. Finally, it accepts if $\ObfM_i$ accepts $\ket*{\widetilde{\psi}}$ in expectation over the choice of $i \gets [N]$.
\end{itemize}

\paragraph{Argument of quantum knowledge.} We first show that our protocol satisfies a straightline (i.e. non-rewinding) argument of quantum knowledge property. To do so, we require that the provably-correct obfuscation satisfies an extraction property, namely that there exists an extractor that, given an obfuscated program, can extract the description of the plaintext program. Given this ability, our QMA extractor is quite natural, and consists of two parts $(\Ext_0, \Ext_1)$.

\begin{itemize}
    \item $\Ext_0$ outputs the public parameters $\crs$ along with a trapdoor $\td$ for the provably-correct obfuscation extractor.
    \item $\Ext_1$ takes input $\left(\td, \pi = (\ket*{\widetilde{\psi}}, \ObfV, \{\ObfM_i\}_i)\right)$. It uses the provably-correct obfuscation extractor to extract the description of $\CSA.\Ver$, which contains a description of the $\CSA$ \emph{authentication key} $k$. Given $k$, it can undo the encoding on the state $\ket*{\widetilde{\psi}}$ in order to obtain an underlying state $\ket{\psi}$.

\end{itemize}

It remains to analyze the extracted state $\ket{\psi}$. In particular, in the case that the verifier would have accepted, must the extracted state be a QMA witness of high quality?

In the description of the protocol above, we left the final line of the verifier's action somewhat unclear. For example, suppose that the verifier simply samples a uniform $i \gets [N]$, applies $\widetilde{\cM}_i$ to $\ket*{\widetilde{\psi}}$, and accepts if the measurement accepts. By the negligible soundness of the ZX verifier, it holds that if the verifier accepts with non-negligible probability, then the instance must have been a yes instance. Thus, this argument would be good enough to establish standard soundness of the protocol.

However, accepting with non-negligible probability only guarantees that the witness encoded by the prover would be accepted by the QMA verifier with some non-negligible probability. That is, the encoded state $\ket{\psi}$ might be a witness of some arbitrarily small inverse polynomial ``quality'', meaning that the QMA verifier accepts $\ket{\psi}$ with some inverse polynomial probability $\gamma = 1/\poly$. Thus, a natural question is whether it is possible to guarantee extraction of \emph{high} quality witnesses, namely states $\ket{\psi}$ that cause the QMA verifier to accept with probability close to 1.

To do so, the verifier will need to implement a procedure that essentially distinguishes between witnesses of varying quality. Fortunately, \cite{TCC:Zhandry20} derived a method for approximating the probability that a given state is accepted by any \emph{mixture of projective measurements}, which we can take to be the set $\{\widetilde{\cM}_i\}_i$. Given any desired cutoff $1-\gamma$ for inverse polynomial $\gamma$, the verifier can run \cite{TCC:Zhandry20}'s ``approximate threshold implementation'' (ATI) procedure, and only accept if ATI accepts. In this case, the state has been (approximately) projected into a space of states that cause the QMA verifier to accept with probability at least $1-\gamma$. 

This allows us to establish a strong extraction property: Our extractor will succeed with probability negligibly close to the verifier (meaning we have no knowledge error), and in the case of success, can output a QMA witness with arbitrarily high quality. This improves upon the prior extraction result of \cite{C:ColVidZha20}, whose protocol had some non-negligible knowledge error and lower witness quality (in addition to being in the secret parameters model).

\paragraph{Zero-knowledge.} Our proof of zero-knowledge is somewhat involved, and motivates our use of the evasive composability heuristic on the obfuscator, discussed earlier in \cref{sec:intro}.

Consider an encoded witness $\ket*{\widetilde{\psi}}$ along with its obfuscated codespace membership tester $\ObfV$ and set of obfuscated ZX measurements $\{\ObfM_i\}_i$. Roughly, our goal will be to replace $\ket*{\widetilde{\psi}}$ with an encoded zero state $\ket*{\widetilde{0}}$, which clearly contains no information about the witness.

Encouragingly, \cite{BBV24} has established that we can do this as long as the \emph{only} side information is the obfuscated codespace membership tester $\ObfV$. While technically they show this when the obfuscation is modeled as a black-box, it is easy to see that the only property they use for this is ``subspace-hiding'' \cite{C:Zhandry19}, which is implied by indistinguishability obfuscation. 

We take this one step further. In \cref{sec:csa}, we show that for any state $\ket{\psi}$ and ZX measurement $\mathcal{M}$ such that $\mathcal{M}$ accepts $\ket{\psi}$ with probability $1-\negl(\secp)$, it holds that \[\left(\ket*{\widetilde{\psi}},\ObfM\right) \approx \left(\ket*{\widetilde{0}},\ObfV\right),\] where we still only make use of indistinguishability obfuscation (iO).

Unfortunately, this claim is still not enough to argue zero-knowledge due to the presence of \emph{many} obfuscated ZX measurements. In fact, we run into trouble when trying to argue about an ensemble of the form
\[\left(\ket*{\widetilde{\psi}},\ObfM_0,\ObfM_1\right),\]

where $\ObfM_0$ is an obfuscated ZX measurement with respect to bases $\theta$, $\ObfM_1$ is an obfuscated ZX measurement with respect to bases $\theta'$, and $\theta \neq \theta'$. The reason is that a crucial step in the iO-based proof \emph{decomposes} the state $\ket*{\widetilde{\psi}}$ in the $\theta$-basis, and argues separately about each component. However, while $\ket*{\widetilde{\psi}}$ itself may be accepted by $\ObfM_1$, it's components in bases $\theta$ may \emph{not} be, meaning that $\ket*{\widetilde{\psi}}$ and $\ket*{\widetilde{\psi}}$ measured in the $\theta$-basis are no longer indistinguishable in the presence of $\ObfM_1$! 

Now, if we model the obfuscations as \emph{oracles}, then it is possible to ``paste'' all these arguments together with respect to the separate $\ObfM_i$ and prove that they are all indistinguishable from $\ObfV$ in one fell swoop. Indeed, we consider the difficulty above to merely be a difficulty with the particular proof techniques that we (and \cite{BBV24}) utilize, and leave it as a fascinating open question to identify a new proof technique that relies on only indistinguishability obfuscation.

In this work, rather than resorting to the full-fledged oracle model, we extract out a simple game-based property that we need from the obfuscator, which we refer to as the evasive composability heuristic. We consider this a step towards eventually removing heuristics entirely, and relying on only indistinguishability obfuscation or other concrete assumptions. In particular, this highlights a ``core'' property that we need here (and in other contexts such as obfuscation \cite{BBV24}) from the obfuscated $\CSA$ circuits, which current techniques suffice to prove in the oracle model but not in the plain model. We hope that this will lead to a crisper understanding of the current gap in obtaining results such as NIZKs for QMA and obfuscation for quantum programs in the plain model.

\subsection{The quantum pseudorandom oracle model}\label{subsec:QPRO-overview}

Finally, we discuss our results on the quantum pseudorandom oracle model. We begin by taking a look at our specific NIZK for QMA protocol, attempting to argue security in some idealized model in which the evasive composability heuristic holds.

First consider an attempt to instantiate our protocol with an \emph{ideal obfuscator}. In isolation, it is easy to see that ideal obfuscation satisfies the evasive composability heuristic. However, recall that we also need to prove \emph{correctness} of the obfuscated program. If the prover's obfuscation is  modeled completely as a black-box, it is unclear how to do this. Even if the (otherwise plain model) obfuscator makes use of a \emph{random oracle}, our proposed protocol would require NIZKs for oracle-aided $\NP$ languages, which are not known.

Drawing inspiration from real-world heuristic use of hash functions as random oracles, \cite{JLLW23} recently defined the \emph{pseudorandom oracle} (PrO) model. In the PrO model, query access to the ideal functionality is indistinguishable from a truly random function, yet there exist ``handles'' which can be used to uniquely specify a key for the PrO. 
These ``handles'' can in turn be used to prove properties regarding the PrO. In their paper, \cite{JLLW23} show that, assuming functional encryption, ideal obfuscation exists in the PrO model, thus showing that heuristic use of obfuscation can often be replaced with the heuristic use of a hash function.

For our purposes, we need to extend the JLLW analysis in two ways: (1) we need a \emph{provably-correct} ideal obfuscator in the PrO, and (2) we need to argue \emph{post-quantum} security, meaning the adversary gets quantum superposition access to the PrO, which we call the quantum PrO, or QPrO. It turns out that to address the first challenge, we must make use of a ``cut-and-choose'' trick  wherein we obfuscate multiple programs that each make use of different PrO keys, and require that the prover reveal a random subset of these keys. This enables the verifier to check that the prover is being (mostly) honest about its key to handle mapping (more details can be found in Section \ref{subsec:qpro-construction}). Next, we discuss the second challenge below.

\paragraph{Post-Quantum Security of JLLW.} Our main contribution in this section is to lift the JLLW results to the post-quantum setting, showing that heuristic use of \emph{post-quantum} obfuscation can often be replaced with the heuristic use of post-quantum hash functions.

At a high level, we carefully follow their construction and analysis in the classical setting in order to show that it is indeed secure in the post-quantum setting, but with some small caveats. The main caveat is that the quantum setting seems to require \emph{subexponential} security from the underlying primitives, as required by the works~\cite{FOCS:BitVai15,C:AnaJai15} that inspired JLLW's construction. JLLW's insight to simulate the obfuscated program for an exponential number of potential inputs is to adaptively reprogram the PrO only on the (polynomial number of) inputs which the adversary queries. Unfortunately, adaptively programming a quantum-accessible random oracle is out of reach of current techniques, at least in the context of simulation security. Instead, we are forced to individually address each input individually like in prior works, which causes an exponential security loss. 

The second difference is more minor. Technically, we need to rely on security of the underlying primitives \emph{with access to the QPrO}. The QPrO is implemented using a PRF and a random permutation. To say that the underlying primitives are secure in the QPrO, we need to be able to implement the random permutation. Unfortunately, efficiently statistically simulating a random permutation oracle is an open problem. Instead, we can simulate the QPrO in the plain model using pseudorandom permutations (PRPs). This resolves the issue at the cost of additionally assuming PRPs. Fortunately, post-quantum PRPs are known to be implied by post-quantum PRFs~\cite{DBLP:journals/quantum/Zhandry25}.

\paragraph{Roadmap.}
In \cref{sec:qma-ver}, we define and construct a ZX verifier with strong completeness. We then prove new properties of coset state authentication in \cref{sec:csa}. In \cref{sec:nizk-np}, we construct post-quantum NIZK for $\NP$ with knowledge soundness, by putting together results from prior work. Next, we formalize our definition of provably-correct obfuscation, and prove its existence using an obfuscation scheme in \cref{sec:prov-obf}. In \cref{sec:PROM}, we establish ideal obfuscation in the quantum pseudorandom oracle model. Finally, we construct and prove our main NIZK for $\QMA$ result in \cref{sec:nizk-qma}.

\section{Preliminaries}

We say that two distributions are $\delta$-indistinguishable if no polynomial time adversary can distinguish them with probability better than $\delta$. Frequently, $\delta$ will be an arbitrary negligible function, in which case we simply say that the two distributions are computationally distinguishable. In the case where $\delta = 2^{-\lambda^c}$ for some constant $c$, we say that the distributions are \emph{subexponentially} indistinguishable. If a primitive's security is based on the indistinguishability of two distributions, then we say it is $\delta$-secure if those distributions are $\delta$-indistinguishable. Additionally, for notational purposes, we use use $A[i]$ to denote indexing into a list or string $A$ with $i$. 

\subsection{Statistics}

We denote the spectral norm, which is the largest singular value of a matrix, by $\|\cdot\|_{\mathsf{spec}}$.

\begin{theorem}[Rectangular Matrix Bernstein Inequality\cite{FTML:Tropp15}] \label{thm:matrix-bernstein}
    Consider a finite sequence $\{\mathbf{Z}_k\}$ of independent, random matrices with dimensions $d_1\times d_2$. Assume that each matrix satisfies 
    \[
        \expect[\mathbf{Z}_k] = \mathbf{0} \quad \text{and} \quad 
        \|\mathbf{Z}_k\|_{\mathsf{spec}} \leq R \quad 
    \]
    Define
    \[
        \sigma^2 \coloneqq \max\left\{ \left\| \sum_k \expect[\mathbf{Z}_k \mathbf{Z}^*_k]\right\|_{\mathsf{spec}}, \left\|  \sum_k \expect[\mathbf{Z}^*_k \mathbf{Z}_k] \right\|_{\mathsf{spec}}\right\}
    \]
    Then for all $t\geq 0$,
    \[
        \Pr\left[ \left\|\sum_{k}\mathbf{Z}_k \right\|_{\mathsf{spec}} \geq t \right]
        \leq
        (d_1 + d_2)\exp\left(\frac{-t^2/2}{\sigma^2 + Rt/3}\right)
    \]
\end{theorem}

We can use this inequality to get a concentration inequality on the sum of independent random vectors.

\begin{lemma}\label{coro:vector-bernstein}
    Let $\{\mathbf{V}_k\}_{k\in[n]}$ be a finite sequence of independent random vectors with dimension $d$. If $\quad \|\mathbf{V}_k\|_{2} \leq R$ for some $R\in \bbR$ for all $k$, then
    \[
        \Pr[\left\| \sum_{k}\mathbf{V}_k - \expect\left[ \sum_{k} \mathbf{V}_k \right]\right\|_{2} \geq t] \leq 
        d \exp(\frac{-t^2}{2R(2n + t/3)})
    \]
    This also holds for the L-1 norm, since $\|\mathbf{w} \|_2\leq \|\mathbf{w} \|_1$ for all vectors $\mathbf{w}$. For $t = cn$, the bound becomes $d\exp(-c^2 k/(4R + 2c/3))$.
\end{lemma}
\begin{proof}
    Define  $\mathbf{Z}_k \coloneqq \mathbf{V}_k - \expect[\mathbf{V}_k]$. This random variable has mean $\mathbf{0}$ and satisfies $\quad \|\mathbf{Z}_k\|_{2} \leq 2R$. Since the spectral norm is equivalent the L2 norm on vectors, we may apply the matrix Bernstein inequality to bound $\sum_k \mathbf{Z}_k = \sum_k \mathbf{V}_k - \expect[\sum_{k} \mathbf{V}_k]$ in terms of
    \[
        \sigma^2 = \max\left\{ \left\| \sum_k \expect[\mathbf{Z}_k \mathbf{Z}^*_k]\right\|_{\mathsf{spec}}, \left\|  \sum_k \expect[\mathbf{Z}^*_k \mathbf{Z}_k] \right\|_{\mathsf{spec}}\right\}
    \]
    We can bound this by 
    \begin{align*}\allowdisplaybreaks
        \left\| \sum_k \expect[\mathbf{Z}_k \mathbf{Z}^*_k]\right\|_{\mathsf{spec}}
        &\leq 
        \sum_k \left\|\expect[\mathbf{Z}_k \mathbf{Z}^*_k]\right\|_{\mathsf{spec}}
        \\
        &\leq 
        \sum_k \expect[\left\|\mathbf{Z}_k \mathbf{Z}^*_k]\right\|_{\mathsf{spec}}]
        \\
        &= \sum_k \expect[\left\|\mathbf{Z}_k\right\|_{2}]
        \\
        &\leq 2nR
    \end{align*}
    using a combination of the triangle inequality, Jensen's inequality, and the a-priori bound on $\|\mathbf{V}_k\|_2$.
    A similar bound applies to $\left\| \sum_k \expect[\mathbf{Z}^*_k \mathbf{Z}_k]\right\|_{\mathsf{spec}}$, giving us an overall bound on $\sigma^2$.
\end{proof}

\subsection{Quantum Information}

First, we define the notion of a (binary-outcome) ZX measurement.

\begin{definition}[Binary-outcome ZX measurement]\label{def:ZX-measurement}
    An $n$-qubit binary-outcome ZX measurement is parameterized by a string $\theta \in \{0,1\}^n$ of basis choices (where each 0 corresponds to standard basis and each 1 corresponds to Hadamard basis), and a function $f : \{0,1\}^n \to \{0,1\}$. It is defined as \[\left\{M[\theta,f], \cI - M[\theta,f]\right\}, \ \ \ \text{where} \ \ \ M[\theta,f] \coloneqq H^\theta\left(\sum_{x:f(x) = 1}\ketbra{x}{x}\right)H^\theta.\] We say that the ZX measurement is  efficient if $f$ is computable by a uniform circuit of size polynomial in $n$. We consider only efficient ZX measurements in this work.
\end{definition}

\begin{lemma}[Gentle measurement \cite{DBLP:journals/tit/Winter99}]\label{lemma:gentle-measurement}
Let $\rho$ be a quantum state and let $(\Pi,\identity-\Pi)$ be a projective measurement such that $\Tr(\Pi\rho) \geq 1-\delta$. Let \[\rho' = \frac{\Pi\rho\Pi}{\Tr(\Pi\rho)}\] be the state after applying $(\Pi,\identity-\Pi)$ to $\rho$ and post-selecting on obtaining the first outcome. Then, $\TD(\rho,\rho') \leq 2\sqrt{\delta}$.
\end{lemma}

The following two lemmas are taken mostly verbatim from \cite{BBV24}.

\begin{lemma}[Oracle indistinguishability]\label{lemma:oracle-ind}
    For each $\secp \in \bbN$, let $\cK_\secp$ be a set of keys, and $\{z_k,O_k^0,O_k^1,S_k\}_{k \in \cK_\secp}$ be a set of strings $z_k$, classical functions $O_k^0,O_k^1$, and sets $S_k$. Suppose that the following properties hold.
    \begin{enumerate}
        \item The oracles $O_k^0$ and $O_k^1$ are identical on inputs outside of $S_k$.
        \item For any oracle-aided unitary $U$ with $q = q(\secp)$ queries, there is some $\epsilon = \epsilon(\secp)$ such that 
        \[\E_{k \gets \cK}\left[\big\| \Pi[S_k]U^{O_k^0}(z_k)\big\|^2\right] \leq \epsilon.\]
    \end{enumerate}
    Then, for any oracle-aided unitary $U$ with $q(\secp)$ queries and distinguisher $D$, \[\bigg| \Pr_{k \gets \cK}\left[D\left(k,U^{O_k^0}(z_k)\right) = 0\right] - \Pr_{k \gets \cK}\left[D\left(k,U^{O_k^1}(z_k)\right) = 0\right]\bigg| \leq 4q\sqrt{\epsilon}.\] 
\end{lemma}

\begin{lemma}[State decomposition]\label{lemma:state-decomp}
    Let $\cK$ be a set of keys, $N$ an integer, and $\{\ket{\psi_k},\{\Pi_{k,i}\}_{i \in [N]}\}_{k \in \cK}$ be a set of states $\ket{\psi_k}$ and projective submeasurements $\{\Pi_{k,i}\}_{i \in [N]}$ such that $\ket{\psi_k} \in \mathsf{Im}(\sum_i \Pi_{k,i})$ for each $k$. Then for any binary-outcome projector $D$, it holds that 
    \[\E_{k \gets \cK}\left[\|D\ket{\psi_k}\|^2\right] - \sum_i\E_{k \gets \cK}\|D\Pi_{k,i}\ket{\psi_k}\|^2 \leq N \cdot \sqrt{\sum_{i \neq j}\E_{k \gets \cK}\|\Pi_{k,j}D\Pi_{k,i}\ket{\psi_k}\|^2}.\]
\end{lemma}

\subsection{Estimating Quantum Acceptance Probabilities}

\cite{TCC:Zhandry20} gave a method of approximating the probability that a state is accepted by a POVM $(\cP= \sum_{i} p_i P_i, \cQ= \sum_{i} p_i (I- P_i))$ which is a  mixture of binary-outcome projective measurements $\{P_i, I-P_i\}$. Crucially, the method is almost-projective. In other words, if run twice, it will almost certainly give the same result both times. Later, \cite{C:ALLZZ21} observed that the technique can be applied to test if a state's acceptance probability is greater than some threshold.

Although the technique is quite general, we only need a few very specific properties that arise from plugging in specific parameters to the general technique.
We refer the reader to \cite{TCC:Zhandry20} for a fully detailed description of the general technique.

\begin{lemma}
    \label{lem:ati}
    Let $(\cP = \sum_{i} p_i P_i, \cQ= \sum_{i} p_i (I- P_i))$ be a mixture of projective measurements such that each is implementable in time $\poly(\secp)$ and one can sample from the distribution defined by $\Pr[i] = p_i$ in time $\poly(\secp)$. For any $\gamma(\secp) \in [1/\poly(\secp),1-1/\poly(\secp)]$, there exists an algorithm $\ATI_\gamma$ outputting $\Accept$ or $\Reject$ such that the following hold.
    \begin{itemize}
        \item \textbf{Efficient.} The expected running time of $\ATI_\gamma$ is $\poly(\secpar)$.
        
        \item \textbf{Approximately Projective.} $\ATI_\gamma$ is approximately projective. In other words, for all states $\rho$,
        \[
            \Pr\left[b_1 = b_2: \begin{array}{c}
                 (b_1, \rho') \gets \ATI_\gamma(\rho)  \\
                 b_2 \gets \ATI_\gamma(\rho') 
            \end{array}\right] = 1-\negl(\secpar).
        \]
        
        \item \textbf{Soundness.} For every state $\rho$
        \[
            \Pr\left[\begin{array}{c}
                 b = \Accept\ \land \\
                 \Tr[\cP\rho'] \leq 1-\gamma(\secp)
                 \end{array}
                 :
                 \ (\rho', b)\gets \ATI_\gamma(\rho)
                 \right] = \negl(\secpar).
        \]

        \item For every state $\ket{\psi}$ such that $\Tr[\cP \ket{\psi}\!\bra{\psi}] \geq 1 - \negl(\secpar)$,
        \[
            \Pr[\Accept \gets \ATI_\gamma(\ket{\psi})] \geq 1 - \negl(\secpar).
        \]

        \item If $\Tr[\cP \rho] < 1-\gamma(\secp)$ for every state $\rho$, then for any state $\rho$,
        \[
            \Pr[\Accept \gets \ATI_\gamma(\rho)] = \negl(\secpar).
        \]
    \end{itemize}
\end{lemma}
\begin{proof}[Proof Sketch]
    This follows by plugging in explicit parameters to Corollary 1 in \cite{C:ALLZZ21} and Theorem 2 in \cite{TCC:Zhandry20}. Specifically, set the approximation precision $\epsilon = \gamma/2$, set the approximation accuracy $\delta = 2^\secpar$, and set the threshold to 1. 
    
    Their algorithm runs in time $\poly(2/\gamma(\secp), \log(2^\secpar)) = \poly(\secpar)$ and is $\delta$-approximately projective. It $\delta$-approximates the threshold projective implementation $(\Pi_{\geq 1 - \gamma/2}, I - \Pi_{\geq 1-\gamma/2})$ of $(\cP, \cQ)$. Specifically, if $\ket{\psi}$ is in the image of $\Pi_{\geq 1 - \gamma/2}$, then $\ATI_\gamma$ accepts $\ket{\psi}$ with probability $1-\delta$ and otherwise it rejects it with probability $1-\delta$. Here, $\Pi_{\geq 1 - \gamma/2}$ projects onto eigenstates of the projective implementation of $(\cP, \cQ)$ with eigenvalues $\geq 1-\gamma/2$. Any such eigenstate $\ket{\psi}$ with eigenvalue $\zeta$ has $\Tr[\cP \ket{\psi}] = \zeta$.

    For any state $\rho$ where $\Pr[\Accept \gets \ATI_\gamma(\rho)] =\negl(\secpar)$, it is clearly the case that 
     \[
        \Pr\left[\begin{array}{c}
             b = \Accept\ \land \\
             \Tr[\cP\rho'] \leq 1-\gamma
             \end{array}
             :
             \ (\rho', b)\gets \ATI_\gamma(\rho)
             \right] = \negl(\secpar).
    \]
    On the other hand, if $\ATI_\gamma$ accepts $\rho$ with noticeable probability, then by approximate projectivity the probability that $\ATI_\gamma$ accepts $\rho$ but then rejects the residual state $\rho'$ is negligible. Suppose we are in the case where $\Pr[\Accept \gets \ATI_\gamma(\rho')] = 1-\negl(\secpar)$.
    Since $\ATI_\gamma$ $(2^{-\secpar})$-approximates $(\Pi_{\geq 1 - \gamma/2}, I - \Pi_{\geq 1-\gamma/2})$, $\rho'$ must have negligible projection onto the eigenspace of the projective implementation of $\cP$ with eigenvalues $< 1-\gamma/2$. In other words,  $\Tr[\cP \rho'] \geq 1-\gamma/2 - \negl > 1-\gamma$.

    If $\Tr[\cP \ket{\psi}] \geq 1 - \negl(\secpar)$, then $\ket{\psi}$ must have negligible projection onto the eigenspace of the projective implementation of $(\cP, \cQ)$ with eigenvalues $\leq 1-1/p$ for any $p = \poly(\secpar)$. Thus, $\ket{\psi}$ is accepted by $\ATI_\gamma$ with probability $1-\delta = 1-2^{-\secpar}$. 

    On the other hand, if $\Tr[\cP \rho] < 1-\gamma$ for every $\rho$, then the maximum eigenvalue of projective implementation of $(\cP, \cQ)$ is $< \gamma$. Therefore every state $\rho$ is in the image of $1- \Pi_{\geq 1-\gamma/2}$ and thus $\ATI_\gamma$ rejects $\rho$ with probability $1-\delta = 1-2^{-\secpar}$.
\end{proof}

\subsection{Quantum Complexity}

Let $\cB$ denote the Hilbert space of a single qubit, and let $D(\cB^{\otimes p})$ be the set of $p$-qubit density matrices, i.e.\ $\rho \in D(\cB^{\otimes p})$ iff $\rho$ is PSD and has trace 1.

\begin{definition}[$\QMA$ Promise Problem]
    \label{def:qma}
    A promise problem $(\cL_{yes}, \cL_{no}) \in \QMA$ if there exists a quantum polynomial-size family of circuits $Q = \{Q_n\}_{n \in \bbN}$, a polynomial $p(\cdot)$, and an $\epsilon(\cdot)$ such that $2^{-\Omega(\cdot)} \le \epsilon(\cdot) \le \frac{1}{3}$, such that
    \begin{itemize}
        \item For all $x \in \cL_{yes}$, there exists $\ket{\psi} \in \cB^{\tensor p(|x|)}$ such that $\Pr[Q_{|x|}(x, \ket{\psi}) = 1] \ge 1 - \epsilon(|x|)$.
        \item For all $x \in \cL_{no}$ and all $\ket{\psi} \in \cB^{\tensor p(|x|)}$, $\Pr[Q_{|x|}(x, \ket{\psi}) = 1] \le \epsilon(|x|)$.
    \end{itemize}
\end{definition}

Given a quantum polynomial-size family of circuits $Q = \{Q_n\}_{n \in \bbN}$ where each $Q_n$ takes as input a string $x \in \{0,1\}^n$ and a state $\ket{\psi}$ on $p(n)$ qubits, and a function $\gamma: \bbN \to [0,1]$, we define the sets

\[R_{Q,\gamma} \coloneqq \bigcup_{n \in \mathbb{N}}\left\{(x, \rho) \in \{0,1\}^n \otimes D(\cB^{\otimes p(|x|)}) \ | \ \Pr[Q_n(x,\rho) = 1] \geq \gamma(n)\right\},\]

and 
\[N_{Q,\gamma} \coloneqq \bigcup_{n \in \mathbb{N}}\left\{x \in \{0,1\}^n \ | \ \forall \rho \in D(\cB^{\otimes p(|x|)}), \Pr[Q_n(x,\rho) = 1] < \gamma(n)\right\}.\]

Using this notation, we can define the notion of a QMA relation.

\begin{definition}[$\QMA$ Relation]
    \label{def:qma-rel}
    A QMA language $(\cL_{yes}, \cL_{no})$ is specified by a QMA Relation $(Q,\alpha,\beta)$, where $Q = \{Q_n\}_{n \in \bbN}$ is a quantum polynomial-size family of circuits, and $\alpha,\beta: \bbN \to [0,1]$ are functions, if 
    \[\cL_{yes} \subseteq \bigcup_{n \in \mathbb{N}}\left\{x \in \{0,1\}^n \ | \ \exists \rho \in D(\cB^{\otimes p(|x|)}) \ \text{s.t.} \ (x,\rho) \in R_{Q,\alpha}\right\},\] and $\cL_{no} \subseteq N_{Q,\beta}$.

\end{definition}

Next, we discuss the Local Hamiltonian problem.

\begin{definition}[$2$-local $ZX$-Hamiltonian problem~\cite{BL08,CM16,MNS16}]
    \label{def:2locHam}
    The $2$-local $ZX$-Hamiltonian promise problem $(\cL_{yes}, \cL_{no})$, with functions $a, b$ where $b(n) > a(n)$ and gap $b(n) - a(n) > \poly(n)^{-1}$ for all $n \in \bbN$ is defined as follows. An instance is a Hermitian operator on some number $n$ of qubits, taking the following form:
    \begin{equation*}
        H = \sum_{\substack{i < j \\ S \in \{Z, X\}}} p_{i,j}P_{i, j, S} 
    \end{equation*}
    where probability $p_{i,j} \in [0, 1]$ with $\sum_{i < j} 2p_{i,j} = 1$,  and projector $P_{i, j, S} = \frac{\identity + (-1)^{\beta_{i, j}}S_iS_j}{2}$ for $\beta_{i, j} \in \zo$. 
    \begin{itemize}
        \item $H \in \cL_{yes}$ if the smallest eigenvalue of $H$ is at most $a(n)$.
        \item $H \in \cL_{no}$ if the smallest eigenvalue of $H$ is at least $b(n)$.
    \end{itemize}
\end{definition}

\begin{theorem}[$2$-local $ZX$-Hamiltonian is QMA-complete~\cite{BL08}]
    The 2-local $ZX$-Hamiltonian problem with functions $a,b$ (\cref{def:2locHam}) is QMA-complete if $b(n) - a(n) > \poly(n)^{-1}$.
\end{theorem}

\begin{definition}[$2$-local ZX-Hamiltonian Verifier~\cite{MNS16}]
    \label{def:zxham}
    Let $(\cL_{yes}, \cL_{no})$ be a $2$-local $ZX$-Hamiltonian promise problem. There exists functions $p_{yes}, p_{no}$ where $p_{yes},p_{no}: \bbN \to [0,1]$ and $p_{yes}(n) - p_{no}(n) \ge \poly(n)^{-1}$ for all $n$ such that the following construction has the subsequent properties:

    \textbf{Construction.}
    \begin{itemize}
        \item $(i, j, S) \gets \Samp(H; r)$: The classical polynomial-size circuit $\Samp$ on input instance $H$ outputs indices $i, j$ and choice of basis $S \in \{Z, X\}$ with probability $p_{i,j}$ using uniform randomness $r$.
        
        \item $\ZXVer(H, \ket{\psi}) \in \zo$: The quantum polynomial-size circuit $\ZXVer$ on input instance $H$ and witness $\ket{\psi}$,
        \begin{enumerate}
            \item Sample projector indices $(i, j, S) \gets \Samp(H)$.
            \item Measure the $i$th and $j$th qubits of $\ket{\psi}$ with the projector $\{M_0 = \frac{\identity - S}{2}, M_1 = \frac{\identity + S}{2}\}$ to get $b_i$ and $b_j$.
            \item Output $b_i \xor b_j \xor \beta_{i,j}$. 
            
        \end{enumerate}
    \end{itemize}

    \textbf{Properties.} 
    \begin{itemize}
        \item \textbf{Correctness.}
        For every $H \in \cL_{yes}$ with witness $\ket{\psi}$,
        \begin{equation*}
            \Pr[\ZXVer(H, \ket{\psi}) = 1] \ge p_{yes}(n)
        \end{equation*}
        
        \item \textbf{Soundness.}
        For every $H \in \cL_{no}$, and for every $\rho$,
        \begin{equation*}
            \Pr[\ZXVer(H, \rho) = 1] \le p_{no}(n)
        \end{equation*}
    \end{itemize}
\end{definition}

\subsection{Encryption}

\paragraph{Public-key encryption.} We first define standard (post-quantum) public-key encryption.

\begin{definition}[Post-Quantum Public-Key Encryption]
\label{def:enc}
$(\Gen, \Enc, \Dec)$ is a post-quantum public-key encryption scheme if it has the following syntax and properties.

\noindent \textbf{Syntax.}
\begin{itemize}
    \item $(\pk, \sk) \gets \Gen(1^\lambda)$: The polynomial-time algorithm $\Gen$ on input security parameter $1^\lambda$ outputs a public key $\pk$ and a secret key $\sk$.
    \item $\ct \gets \Enc(\pk, m; r)$: The polynomial-time algorithm $\Enc$ on input a public key $\pk$, message $m$ and randomness $r \in \zo^{r(\lambda)}$ outputs a ciphertext $\ct$.
    \item $m \gets \Dec(\sk, \ct)$: The polynomial-time algorithm $\Dec$ on input a secret key $\sk$ and a ciphertext $\ct$ outputs a message $m$.
\end{itemize}

\noindent \textbf{Properties.}
\begin{itemize}
    \item \textbf{Perfect Correctness}: For every $\lambda \in \bbN^+$ and every $m, r$,
    \begin{equation*}
        \Pr_{\substack{(\pk, \sk) \gets \Gen(1^\lambda)}}[\Dec(\sk, \Enc(\pk, m; r)) = m] = 1.
    \end{equation*}

    \item \textbf{Indistinguishability under Chosen-Plaintext (IND-CPA) Secure}: There exists a negligible function $\negl(\cdot)$ such that for every polynomial-size quantum circuit $\cA = (\cA_0, \cA_1)$ and every sufficiently large $\lambda \in \bbN^+$
    \begin{equation*}
        \left\vert \Pr_{\substack{(\pk, \sk) \gets \Gen(1^\lambda) \\  (m_0, m_1, \zeta) \gets \cA_0(1^\lambda, \pk) \\ \ct \gets \Enc(\pk, m_0)}}[\cA_1(1^\lambda, \ct, \zeta) = 1] - \Pr_{\substack{(\pk, \sk) \gets \Gen(1^\lambda) \\  (m_0, m_1, \zeta) \gets \cA_0(1^\lambda, \pk) \\ \ct \gets \Enc(\pk, m_1)}}[\cA_1(1^\lambda, \ct, \zeta) = 1]\right\vert \le \negl(\lambda).
    \end{equation*}
\end{itemize}
\end{definition}

\begin{theorem}
    \label{thm:enc}
\end{theorem}

\paragraph{Functional encryption.} Next, we define a flavor of (1-key) functional encryption described in \cite{JLLW23}, with the additional requirement that it is post-quantum secure.

\begin{definition}[Post-quantum 1-key FE]\label{def:FE}
    A post-quantum (public-key) 1-key functional encryption scheme (for circuits) has the following syntax and properties.
    
    \noindent \textbf{Syntax.}
    \begin{itemize}
        \item $(\pk,\sk_f) \gets \Gen(1^\secp,f)$: The $\Gen$ algorithm takes a circuit $f: \{0,1\}^n \to \{0,1\}^*$ and outputs a public (encryption) key $\pk$ and a secret (decryption) key for $f$.
        \item $\ct \gets \Enc(\pk,z)$: The $\Enc$ algorithm takes a public key and a plaintext $z \in \{0,1\}^n$ and outputs a ciphertext $\ct$.
        \item $f(z) \gets \Dec(\sk_f,\ct)$: The $\Dec$ algorithm takes a secret key for $f$ and a ciphtertext and outputs a string $f(z)$.
    \end{itemize}
    \noindent \textbf{Properties.}
    \begin{itemize}
        \item \textbf{Perfect correctness.} For all $\secp \in \bbN$, circuit $f : \{0,1\}^n \to \{0,1\}^*$, and input $z \in \{0,1\}^n$, it holds that
        \[\Pr\left[\Dec(\sk_f,\ct) = f(z) : \begin{array}{r}(\pk,\sk_f) \gets \Gen(1^\secp,f) \\ \ct \gets \Enc(\pk,z)\end{array}\right] = 1.\]
        \item \textbf{Subquadratic-sublinear efficiency.} $\Enc$ runs in time $(n^{2-2\epsilon} + m^{1-\epsilon})\poly(\secp)$ for some constant $\epsilon > 0$, where $n = |z|$ is the input length of $f$ and $m = |f|$ is the circuit size of $f$.
        \item \textbf{Post-quantum adaptive security.} For any $b \in \{0,1\}$ and adversary $\cA$, let $\Exp^{\cA,b}_{1\text{-key}}$ be the following experiment.
        \begin{itemize}
            \item $\cA(1^\secp)$ outputs a circuit $f: \{0,1\}^n \to \{0,1\}^*$. Run $(\pk,\sk_f) \gets \Gen(1^\secp,f)$, and send $(\pk,\sk_f)$ to $\cA$.
            \item $\cA$ chooses two inputs $z_0,z_1 \in \{0,1\}^n$. Run $\ct \gets \Enc(\pk,z_b)$ and send $\ct$ to $\cA$.
            \item $\cA$ outputs a bit $b' \in \{0,1\}$. The outcome of the experiment is $b'$ if $f(z_0) = f(z_1)$, and is otherwise set to 0.
        \end{itemize}
        There exists an $\epsilon > 0$ such that for any QPT adversary $\cA$, it holds that  \[\bigg|\Pr\left[\Exp^{\cA,0}_{1\text{-key}} = 0\right] - \Pr\left[\Exp^{\cA,1}_{1\text{-key}} = 0\right]\bigg| \leq 2^{-\secp^\epsilon}.\]
    \end{itemize}
\end{definition}

\subsection{NIZK for NP}

\begin{definition}[Post-Quantum NIZK for $\NP$ in the CRS Model]
\label{def:nizk-np}
Let $\NP$ relation $\cR$ with corresponding language $\cL$ be given such that they can be indexed by a security parameter $\lambda \in \bbN$.

$\Pi = (\Setup, \sP, \sV)$ is a post-quantum non-interactive zero-knowledge argument for $\NP$ in the URS model if it has the following syntax and properties.

\noindent \textbf{Syntax.}
The input $1^\lambda$ is left out when it is clear from context.
\begin{itemize}
    \item $\crs \gets \Setup(1^\lambda)$: The probabilistic polynomial-size circuit $\Setup$ on input $1^\lambda$ outputs a common random string $\crs$.
    \item $\pi \gets \sP(1^\lambda, \crs, x, w)$: The probabilistic polynomial-size circuit $\sP$ on input a common random string $\crs$ and instance and witness pair $(x, w) \in \cR_\lambda$, outputs a proof $\pi$.
    \item $\sV(1^\lambda, \crs, x, \pi) \in \zo$: The probabilistic polynomial-size circuit $\sV$ on input a common random string $\crs$, an instance $x$, and a proof $\pi$ outputs $1$ iff $\pi$ is a valid proof for $x$.
\end{itemize}

\noindent \textbf{Properties.}
\begin{itemize}
    \item {\bf Uniform Random String.}
    $\Setup(1^\lambda)$ outputs a uniformly random string $\crs$.
    
    \item {\bf Perfect Completeness.}
    For every $\lambda \in \bbN$ and every $(x, w) \in \cR_\lambda$,
    \begin{equation*}
        \Pr_{\substack{\crs \gets \Setup(1^\lambda) \\ \pi \gets \sP(\crs, x, w)}}[\sV(\crs, x, \pi) = 1] = 1.
    \end{equation*}

    \item {\bf Adaptive Statistical (Computational) Soundness.}
    There exists a negligible function $\negl(\cdot)$ such that for every unbounded (polynomial-size) quantum circuit $\cA$ and every sufficiently large $\lambda \in \bbN$,
    \begin{equation*}
        \Pr_{\substack{\crs \gets \Setup(1^\lambda) \\ (x, \pi) \gets \cA(\crs)}}[\sV(\crs, x, \pi) = 1 \wedge x \not\in \cL_\lambda]  \le \negl(\lambda).
    \end{equation*}

    \item {\bf Non-Adaptive Computational $T$-Soundness.}
    There exists a negligible function $\negl(\cdot)$ such that for every $\poly(T)$-size quantum circuit $\cA$ and every sufficiently large $\lambda \in \bbN$ and $x \not\in \cL_\lambda$,
    \begin{equation*}
        \Pr_{\substack{\crs \gets \Setup(1^\lambda) \\ \pi \gets \cA(\crs)}}[\sV(\crs, x, \pi) = 1]  \le \negl(T(\lambda)).
    \end{equation*}

    \item {\bf Adaptive Computational Zero-Knowledge.}
    There exists a probabilistic polynomial-size circuit $\Sim = (\Sim_0, \Sim_1)$ and a negligible function $\negl(\cdot)$ such that for every polynomial-size quantum circuit $\cD = (\cD_0, \cD_1)$, and every sufficiently large $\lambda \in \bbN$,
    \begin{align*}
        &\left\vert \Pr_{\substack{\crs \gets \Setup(1^\lambda) \\ (x, w, \zeta) \gets \cD_0(\crs) \\ \pi \gets \sP(\crs, x, w)}}[\cD_1(\crs, x, \pi, \zeta) = 1 \wedge x \in \cL_\lambda] - \Pr_{\substack{\crs \gets \Sim_0(1^\lambda) \\ (x, w, \zeta) \gets \cD_0(\crs) \\ \pi \gets \Sim_1(\crs, x)}}[\cD_1(\crs, x, \pi, \zeta) = 1 \wedge x \in \cL_\lambda]\right\vert \\
        & \quad \le \negl(\lambda).
    \end{align*}

    \item {\bf Non-Adaptive Statistical Zero-Knowledge.}
    There exists a probabilistic polynomial-size circuit $\Sim$ and a negligible function $\negl(\cdot)$ such that for every unbounded quantum circuit $\cD$, and every sufficiently large $\lambda \in \bbN$ and every $(x, w) \in \cR_\lambda$,
    \begin{equation*}
        \left\vert \Pr_{\substack{\crs \gets \Setup(1^\lambda) \\ \pi \gets \sP(\crs, x, w)}}[\cD(\crs, x, \pi) = 1] - \Pr_{\substack{(\crs, \pi) \gets \Sim(x)}}[\cD(\crs, x, \pi) = 1]\right\vert \le \negl(\lambda).
    \end{equation*}
\end{itemize}
\end{definition}

\begin{theorem}[Post-Quantum NIZK proof for $\NP$ with CRS~\cite{C:PeiShi19}]
    \label{thm:nizk-np-crs}
    Assuming the polynomial quantum hardness of LWE, there exists an adaptively statistically sound, adaptively computationally zero-knowledge non-interactive protocol for $\NP$ having a common reference string (\cref{def:nizk-np}).
\end{theorem}

\begin{theorem}[Post-Quantum NISZK argument for $\NP$ with URS~\cite{C:PeiShi19}]
    \label{thm:niszk-np-urs}
    Assuming the polynomial quantum hardness of LWE, there exists a non-adaptively computationally sound, non-adaptively statistically zero-knowledge non-interactive protocol for $\NP$ having a uniform random string (\cref{def:nizk-np}).
\end{theorem}

\begin{corollary}[Post-Quantum NISZK sub-exp argument for $\NP$ with URS]
    \label{cor:subexp-nizk-np-urs}
    Assuming the sub-exponential quantum hardness of LWE, there exists a non-adaptively computationally sound, non-adaptively statistically zero-knowledge non-interactive protocol for $\NP$ with sub-exponential computational soundness error having a uniform random string (\cref{def:nizk-np}).
\end{corollary}

\begin{proof}
    This follows from \cref{thm:niszk-np-urs}.
\end{proof}

\subsection{Obfuscation}

\begin{definition}[Indistinguishability obfuscation]\label{def:iO} An indistinguishability obfuscator has the following syntax.

\begin{itemize}
        \item $\Obf(1^\lambda,C) \to \widetilde{C}$. The obfuscation algorithm takes as input the security parameter and a circuit $C$, and outputs an obfuscated circuit $\widetilde{C}$.
        \item $\Eval(\widetilde{C},x) \to y$. The evaluation algorithm takes as input an obfuscated circuit $\widetilde{C}$ and an input $x$ and outputs $y$.
    \end{itemize}

    It should satisfy the following properties.

    \begin{itemize}
        \item \textbf{Functionality-preservation}. For any circuit $C$, $\widetilde{C} \in \Obf(1^\lambda,C)$, and $x$, $\Eval(\widetilde{C},x) = C(x)$.
        \item \textbf{(Sub-exponential) security}. There exists a constant $\epsilon > 0$ such that for any QPT adversary $\cA$ and $C_0,C_1$ such that $C_0 \equiv C_1$, 
        \[\Bigg|\Pr\left[\cA\left(\Obf(1^\lambda,C_0)\right) = 1\right] - \Pr\left[\cA\left(\Obf(1^\lambda,C_1)\right) = 1\right]\Bigg| \leq 2^{-\lambda^\epsilon}.\]
    \end{itemize}

\end{definition}

Before stating the next imported theorem, we introduce the following notation.  For any set $S$, define $C[S]$ to the membership-checking circuit that, on input a vector $v \in \bbF_2^n$, outputs 1 if $v \in S$, and outputs 0 otherwise.  

\begin{theorem}[Subspace-hiding obfuscation \cite{10.1007/s00145-020-09372-x}]\label{lemma:shO} Let $(\Obf,\Eval)$ be a sub-exponentially secure indistinguishability obfuscator, and suppose that sub-exponentially secure injective one-way functions exist. Let $S \subset \bbF_2^n$ be a subspace of $\bbF_2^n$ of dimension $d_0$, let $d_1$ be such that $d_0 < d_1 < n$, and define $\lambda = n-d_1$. There exists a polynomial $p(\cdot)$ such that for any QPT adversary $\cA$, 

\[\Bigg|\Pr\left[\cA\left(\Obf(1^{p(\lambda)},C[S])\right) = 1\right] - \left[\cA\left(\Obf(1^{p(\lambda)},C[T])\right) = 1 : T \gets \Sup_{d_1}(S)\right]\bigg| = 2^{-\Omega(\lambda)},\]

where $\Sup_{d_1}(S)$ is the set of superspaces of $S$ of dimension $d_1$.
    
\end{theorem}

We remark that \cite{10.1007/s00145-020-09372-x} proves the slightly different statement that, assuming polynomially-secure iO and injective one-way functions, the above advantage is at most negligible in some parameter $\lambda$, as long as $n-d_1$ is linear in $\lambda$. It is straightforward to port their proof to our setting of sub-exponential security.

Finally, we note that the following notion of point-function obfuscation follows as a corollary.

\begin{theorem}[Point-function obfuscation]\label{thm:point-obf}
    Let $(\Obf,\Eval)$ be a sub-exponentially secure indistinguishability obfuscator, and suppose that sub-exponentially secure injective one-way functions exist. There exists a polynomial $p(\cdot)$ such that for any QPT adversary $\cA$, 
    
    \[\Bigg|\Pr\left[\cA\left(\Obf(1^{p(\lambda)},C[\{\}])\right) = 1\right] - \left[\cA\left(\Obf(1^{p(\lambda)},C[\{x\}])\right) = 1 : x \gets \{0,1\}^\secp\right]\bigg| = 2^{-\Omega(\lambda)}.\]
\end{theorem}

\subsection{The (Q)PrO model}

First, we define the quantum-accessible pseudorandom oracle (QPrO) model, which extends the psueudorandom oracle model introduced in \cite{JLLW23} to allow for quantum queries.

\begin{definition}[QPrO Model]\label{def:QPRO}
    Let $F = \{f_k\}_k$ be a pseudorandom function. The quantum-accessible pseudorandom oracle model for $F$ consists of the following interface, which internally use a uniformly random permutation $\pi: \{0,1\}^\lambda \to \{0,1\}^\lambda$, and may be queried in quantum superposition.
    \begin{itemize}
        \item $\QPRO(\Gen,k) \to \pi(k)$
        \item $\QPRO(\Eval,h,x) \to f_{\pi^{-1}(h)}(x)$
    \end{itemize}
    The $p(\lambda)$-QPrO model allows the querier to access independent $p(\lambda)-$QPrO oracles for some polynomial $p$, i.e., oracle access to $\QPRO$ is shorthand for allowing query access to $p(\lambda)$-independent $\QPRO$ instantiations $\QPRO_0, \QPRO_1, \ldots, 
    \QPRO_{p(\lambda)-1}$.
\end{definition}

Next, we present the construction of obfuscation in the pseudorandom oracle model due to \cite{JLLW23}. While \cite{JLLW23} show that this scheme satisfies ideal obfuscation in the PrO, we will show in \cref{sec:PROM} that this scheme in fact satisfies post-quantum ideal obfuscation in the QPrO (as long as the building blocks are post-quantum). Before presenting the obfuscator, we define ideal obfuscation (in an oracle model). We use $C^\bullet$ to denote an oracle-aided circuit.

\begin{definition}[Ideal obfuscation]\label{def:ideal-obf}
    An obfuscation scheme in an idealized model with oracle $\cO$ is an efficient algorithm $\Obf^\cO(1^\secp,C)$ that, given a circuit $C$ as input, outputs an oracle circuit $\widehat{C}^\bullet$. 
    The scheme must be \textbf{correct}, i.e.\ for all $\secp \in \bbN$, circuit $C : \{0,1\}^D \to \{0,1\}^*$, and input $x \in \{0,1\}^D$, it holds that 
    \[\Pr\left[\widehat{C}^\cO(x) = C(x) : \widehat{C}^\bullet \gets \Obf^\cO(1^\secp,C)\right] = 1.\]
    It satisfies (post-quantum) \textbf{ideal obfuscation relative to an oracle} $\cR$ if there exists a QPT simulator $\cS = (\cS_1,\cS_2,\cS_3)$ such that for all QPT adversaries $\cA = (\cA_1,\cA_2)$,
    \[
        \left| 
        \Pr\left[ 
            \cA_2^{\cO}(\widehat{C}^\bullet) = 1 : 
            \begin{array}{r}
                C \gets \cA_1^{\cO(\secp), \cR} 
                \\ \widehat{C}^\bullet \gets \Obf^{\cO, \cR}(1^\secp,C)
            \end{array}\right] 
        - \Pr\left[
            \cA_2^{\cS_3^{C}, \cR}(\widetilde{C}^\bullet) = 1 : 
            \begin{array}{r}
                C \gets \cA_1^{\cS_1, \cR}(1^\secp) \\
                \widetilde{C}^\bullet \gets \cS_2^{C}(1^\secp,1^D,1^S)
            \end{array}
        \right]
        \right| 
        = \negl(\secp),
    \]
    where $D = |x|$ in the input length of $C$ and $S = |C|$ is the circuit size of $C$.

\end{definition}

An important difference from \cite{JLLW23}'s definition of ideal obfuscation is the addition of a relativizing oracle $\cR$. The additional of this oracle is crucial for composability with other primitives that might exist in the PrO.
In the plain model, ideal obfuscation is naturally composeable via a hybrid argument. However, when constructed in an oracle model such as the PrO, the simulator for an individual obfuscation might seize control of the global oracle. Unfortunately, this can interfere with the simulators for the other instances, which \emph{also} need control over the global oracle.

By introducing the relativizing oracle $\cR$, which the simulator is not allowed to claim control of, we ensure that the other simulators can also operate.
As a simple example, one can imagine a world in which multiple hash functions (or a single hash function with different salts) define distinct PrOs.

\paragraph{Construction in the PrO.} Now, we describe the construction of obfuscation in the PrO due to \cite{JLLW23}. We first specify the ingredients:

\begin{itemize}
    \item $D$ the input length of the circuit $C$ to be obfuscated.
    \item $S$ the circuit size of $C$
    \item $L$ the block length (determined as in \cite{JLLW23}).
    \item $B$ the number of blocks (determined as in \cite{JLLW23}).
    \item $H : \{0,1\}^\secp \times \{0,1\}^D \to \{0,1\}^L$ the (quantum-query secure) PRF used by the QPrO model.
    \item $G_{sr}: \{0,1\}^ \secp \to \{0,1\}^{4\secp}$ the (post-quantum) PRG for encryption randomness.
    \item $G_v : \{0,1\}^\secp \to \{0,1\}^L$ the (post-quantum) PRG for decryption result simulation.
    \item $(\Gen,\Enc,\Dec)$ a (post-quantum) 1-key FE scheme (\cref{def:FE}) such that $\Enc$ uses $\secp$-bit uniform randomness.
\end{itemize}

\begin{construction}[JLLW Obfuscator]\label{construction:JLLW}
    The JLLW obfuscator is defined as follow, where $\QPRO$ is the pseudorandom oracle model defined in \cref{def:QPRO}, using PRF $H$.\\
    
    \noindent \underline{$\JLLWObf^\QPRO(1^\secp,C)$:}
    \begin{itemize}
        \item Set up $(D+1)$ FE instances:
        \begin{align*}
            &(\pk_D,\sk_D) \gets \Gen(1^\secp,\Eval),\\
            &(\pk_d,\sk_d) \gets \Gen(1^\secp,\Expand_d[\pk_{d+1}]) \ \ \ \text{for } d = D-1,\dots,0,
        \end{align*}
        where $\Eval$ and $\Expand_d$ are defined below. 
        \item Sample keys of $H$ and obtain their handles:
        \[k_{i,j} \gets \{0,1\}^\secp, \ \ \ h_{i,j} \gets \QPRO(\Gen,k_{i,j}) \ \ \ \text{for } 0 \leq i < D, 1 \leq j \leq B.\]
        \item Sample PRG seed and encryption randomness for the root ciphertext, set its flag and information, and compute $\ct_\epsilon$:
        \begin{align*}
            s_\epsilon &\gets \{0,1\}^\secp, \ \ \ r_\epsilon \gets \{0,1\}^\secp \\ \flag_\epsilon &\coloneqq \normal, \ \ \ \info_\epsilon \coloneqq (C,\{k_{i,j}\}_{0 \leq i < D,1 \leq j \leq B},s_\epsilon)\\ \ct_\epsilon &\coloneqq \Enc(\pk_0,\flag_\epsilon,\epsilon,\info_\epsilon;r_\epsilon).
        \end{align*}
        \item Output the circuit $\widehat{C}^\bullet[\ct_\epsilon,\{\sk_d\}_{0 \leq d \leq D},\{h_{i,j}\}_{0 \leq i < D,1 \leq j \leq B}]$, which operates as follows on input $x$:
        \begin{itemize}
            \item For $d = 0,\dots,D-1$:
            \begin{itemize}
                \item $\chi_d \coloneqq x_{\leq d}$
                \item $v_{\chi_d} \gets \Dec(\sk_d,\ct_{\chi_d})$
                \item $\otp_{\chi_d} \coloneqq \QPRO(\Eval,h_{d,1},\chi_d \| 0^{D-d}) \| \dots \| \QPRO(\Eval,h_{d,B},\chi_d \| 0^{D-d})$
                \item $\ct_{\chi_d \| 0}\|\ct_{\chi_d \| 1} \coloneqq v_{\chi_d} \oplus \otp_{\chi_d}$
            \end{itemize}
            \item Output $\Dec(\sk_D,\ct_x)$.
        \end{itemize}
    \end{itemize}
\end{construction}

Next, we define the helper functions that were used in the definition of the JLLW obfuscation scheme above.\\

\noindent $\Expand_d[\pk_{d+1}](\flag_\chi,\chi,\info_\chi)$:
    \begin{align*}
    \text{Output}\begin{cases}
        \Expand_{d,\normal}[\pk_{d+1}](\chi,\info_\chi) \ \ \text{if } \flag_\chi = \normal,\\
        \Expand_{d,\hybsf}[\pk_{d+1}](\chi,\info_\chi) \ \ \text{if } \flag_\chi = \hybsf,\\
        \Expand_{d,\simsf}(\chi,\info_\chi) \ \ \text{if } \flag_\chi = \simsf
    \end{cases}
    \end{align*}

\noindent $\Eval(\flag_\chi,\chi,\info_\chi):$
    \begin{align*}
    \text{Output}\begin{cases}
        \Eval_\normal(\chi,\info_\chi) \ \ \text{if } \flag_\chi = \normal,\\
        \Eval_\simsf(\chi,\info_\chi) \ \ \text{if } \flag_\chi = \simsf
    \end{cases}
    \end{align*}

\noindent $\Expand_{d,\normal}[\pk_{d+1}](\chi,\info_\chi):$
\begin{itemize}
    \item Parse $\info_\chi = (C,\{k_{i,j}\}_{d \leq i < D,1 \leq j \leq B},s_\chi)$
    \item Set $s_{\chi \| 0} \| r_{\chi \| 0} \| s_{\chi \| 1} \| r_{\chi \| 1} \coloneqq G_{sr}(s_\chi)$
    \item For $\eta = 0,1:$
    \begin{itemize}
        \item $\flag_{\chi \| \eta} \coloneqq \normal$
        \item $\info_{\chi \| \eta} \coloneqq (C,\{k_{i,j}\}_{d+1 \leq i < D,1 \leq j \leq B},s_{\chi \| \eta})$
        \item $\ct_{\chi \| \eta} \coloneqq \Enc(\pk_{d+1},\flag_{\chi \| \eta},\chi \| \eta,\info_{\chi \| \eta} ; r_{\chi \| \eta})$
    \end{itemize}
    \item $\otp_\chi \coloneqq H(k_{d,1},\chi \| 0^{D-d}) \| \dots \| H(k_{d,B},\chi \| 0^{D-d})$
    \item Output $v_\chi \coloneqq (\ct_{\chi \| 0} \| \ct_{\chi \| 1}) \oplus \otp_\chi$
\end{itemize}

$\Eval_\normal(\chi,\info_\chi):$
\begin{itemize}
    \item Parse $\info_\chi = (C,s_\chi)$
    \item Output $C(\chi)$, computed by evaluating a universal circuit at $(C,\chi)$
\end{itemize}

$\Expand_{d,\hybsf}[\pk_{d+1},\info_\chi]:$
\begin{itemize}
    \item Parse $\info_\chi = (C,\{k_{i,j}\}_{d < i < D,1 \leq j \leq B},s_\chi,\beta,\{\sigma_{\chi,j}\}_{1 \leq j < \beta},w_\chi,\{k_{d,j}\}_{\beta < j \leq B})$
    \item Set $s_{\chi \| 0} \| r_{\chi \| 0} \| s_{\chi \| 1} \| r_{\chi \| 1} \coloneqq G_{sr}(s_\chi)$
    \item For $\eta = 0,1:$
    \begin{itemize}
        \item $\flag_{\chi \| \eta} \coloneqq \normal$
        \item $\info_{\chi \| \eta} \coloneqq (C,\{k_{i,j}\}_{d+1 \leq i < D,1 \leq j \leq B},s_{\chi \| \eta})$
        \item $\ct_{\chi \| \eta} \coloneqq \Enc(\pk_{d+1},\flag_{\chi \| \eta},\chi \| \eta,\info_{\chi \| \eta} ; r_{\chi \| \eta})$
    \end{itemize}
    \item Output \begin{align*}v_\chi \coloneqq \ &G_v(\sigma_{\chi,1}) \| \dots \| G_v(\sigma_{\chi,\beta-1}) \| w_\chi \\ &\| [\ct_{\chi \| 0} \| \ct_{\chi \| 1}]_{\beta+1} \oplus H(k_{d,\beta+1},\chi \| 0^{D-d}) \| \dots \\ &\| [\ct_{\chi \|0} \| \ct_{\chi \| 1}]_B \oplus H(k_{d,B},\chi \| 0^{D-d})\end{align*}
\end{itemize}

$\Expand_{d,\simsf}(\chi,\info_\chi):$
\begin{itemize}
    \item Parse $\info_\chi = \{\sigma_{\chi,j}\}_{1 \leq j \leq B}$
    \item Output $v_\chi \coloneqq G_v(\sigma_{\chi,1}) \| \dots \| G_v(\sigma_{\chi,B})$
\end{itemize}

$\Eval_\simsf(\chi,\info_\chi):$
\begin{itemize}
    \item Parse $\info_\chi = y_\chi$
    \item Output $y_\chi$
\end{itemize}

Finally, we have the following theorem, which we will prove in \cref{sec:PROM}.

\begin{theorem}
    The JLLW obfuscation $\JLLWObf^\QPRO(1^\secp,C)$ given in \cref{construction:JLLW} satisfies post-quantum ideal obfuscation (\cref{def:ideal-obf}) in the quantum-accessible pseudorandom oracle model (\cref{def:QPRO})
\end{theorem}

\section{QMA Verification with Strong Completeness}
\label{sec:qma-ver}

We first define a special class of ``ZX'' QMA verifiers satisfying a notion of ''strong'' completeness, which demands that for an honest witness, every ZX measurement the verifier may apply will accept with overwhelming probability.

\begin{definition}[ZX verifier with strong completeness]\label{def:ZX-strong}
    A ZX verifier with strong completeness for a QMA language $(\cL_{yes},\cL_{no})$ consists of, for each instance $H$ (an $n$-qubit Hermitian matrix), a family 
    $\{\theta_{H,i},f_{H,i}\}_{i \in [N]}$
    of binary-outcome ZX measurements (\cref{def:ZX-measurement}), where $N = N(n) = 2^{\poly(n)}$ is some (potentially exponential) function of the instance size. It satisfies the following properties.
    
    \begin{itemize}
        \item \textbf{Strong completeness}. For each $H \in \cL_{yes}$, there exists a state $\ket{\psi}$ such that for all $i \in [N]$,
        \[
            \Big\| M[\theta_{H,i},f_{H,i}]\ket{\psi}\Big\|^2 
            \geq 1-2^{-n^3}.
        \]
        \item \textbf{Soundness}. For each $H \in \cL_{no}$ and any state $\ket{\psi}$,
        \[
            \E_{i \gets [N]} \left[\Big\| M[\theta_{H,i},f_{H,i}]\ket{\psi} \Big\|^2\right] 
            = \negl(n).
        \]
    \end{itemize}
\end{definition}

\begin{theorem}\label{thm:ZX-strong}
    Every language in QMA has a ZX verifier with strong completeness (\cref{def:ZX-strong}).
\end{theorem}

The main ingredient to our proof is a lemma that turns any QMA verifier which applies a random projective measurement to the witness into one with strong correctness. The theorem follows from applying the lemma to the protocol given in \cite{MNS16} for QMA verification via single qubit ZX measurements.
For completeness, we state the protocol for permuting ZX verifiers in \Cref{sec:permutingzx}.

\begin{lemma}[Permuting QMA Verifiers]\label{lem:permuting-QMA}
    Let $\mathcal{L} = (\cL_{yes}, \cL_{no})$ be a QMA language with instance size $n$.
    Let $\{(P_j, I-P_j)\}_{j\in \cJ}$ be a $\poly(n)$-sized set of binary-outcome projective measurements on $n$ qubits and let $(\mathcal{P} = \sum_{j\in \cJ} p_j P_j, \mathcal{Q} = \sum_{j\in \cJ} p_j (I - P_j))$ be a POVM which decides $\mathcal{L}$ with correctness $a$ and soundness error $b$.
    
    Then there is a verifier for $\mathcal{L}$ with strong correctness and soundness error $\negl$ which only performs measurements from $\{P_j, I-P_j\}_{j\in J}$.
\end{lemma}
\begin{proof}
    The permuting verifier operates on $k = \left\lceil \frac{4}{a-b}\left(\frac{n^3\ln(2)}{a-b}+|\cJ|\right)\right\rceil + 1$ registers $\cR_j$ each containing $n$ qubits. Let $\mathsf{List}$ be an ordered list containing each $j\in \cJ$ a total of $\lfloor k p_j \rfloor$ times. 
    The family of possible measurements the verifier can make is given by all possible permutations of $\mathsf{List}$. The deciding function $f$ accepts if at least $k \frac{a+b}{2}$ of the measurement outcomes are $P$. In other words, the family of measurements is
    \[
        \{\sigma(\mathsf{List}), f\}_{\sigma\in \mathsf{Sym}_{k} }
    \]
    The distribution $\mathsf{Samp}(H)$ samples a uniform $\sigma \gets \mathsf{Sym}_{k}$ and outputs $(\sigma(\mathsf{List}), f)$.

    \begin{claim}
        The verifier above has strong completeness.
    \end{claim}
    \begin{proof}
        For any $H\in \mathcal{L}_{yes}$, there exists an $n$-qubit witness $\ket{w}$ such that $\Tr[\cP \ket{w}] \geq a$. The witness for the permuting verifier is $k$ copies of this witness, i.e. $\ket{w}^{\otimes k}$.
        Since the witness is separable across the registers $\cR_j$ and the verifier applies its projectors on disjoint registers, the outcome for each copy is independent.
        Let $S_{k}$ be the random variable representing the number of accepting measurement results (outcome $P_j$).
        Hoeffding's inequality allows us to bound the probability that the sum of the outcomes differs from its expected value by
        \[
            \Pr_{x_i: i\in [k]}\left[S_{k} \leq \mathop{\mathbb{E}}_{x_i:i\in k}[S_{k}] - t \right]
            \leq 2\exp\left(-2\frac{t^2}{k}\right)
        \]
        The expectation of the summation is
        \begin{align*}
            \sum_{i\in [k]} \Tr[P_{\sigma(\List[i])} \ket{w}]
            &= \Tr[\sum_{j\in \cJ} \lfloor k p_j\rfloor P_j \ket{w}]
            \\
            &\geq k \Tr[\sum_{j\in \cJ} p_i P_j \ket{w}] - |\cJ|
            \\
            &= k a - |\cJ|
        \end{align*}
        Setting $t = k\frac{a-b}{2} - |\cJ|$, we have $t \geq 4\frac{n^3\ln(2)}{a-b}$ and $t/k \geq \frac{a-b}{4}$. Therefore the probability of $S_k \leq k\frac{a+b}{2}$, i.e. the verifier rejecting, is 
        \[
            \leq 2\exp\left(-2n^3\ln(2)\right) = 2^{-2n^3+1}
        \]
    \end{proof}
    
    \begin{claim}
        The verifier above has $\negl(n)$ soundness error.
    \end{claim}
    \begin{proof}
        For any $H\notin \cL$, every state $\rho$ satisfies $\Tr[\cP \rho] \leq b$. We will relate the number of accepting repetitions in the $k$-wise parallel repetition of the decision procedure to the number of accepting repetitions in the permuted procedure. 
        
        For any repetition $i$ in the parallel case, the probability of accepting \emph{any} mixed state is at most $b$. Thus, conditioned on any outcome of the other repetitions, the probability of repetition $i$ accepting is still at most $b$.
        Therefore the probability of obtaining $\geq n$ accepts is upper bounded by the probability of sampling $\geq n$ in a binomial distribution with success rate $b$ for any $n$. Let $S_{\mathsf{par}}$ be distributed according to this binomial distribution. Hoeffding's inequality bounds the probability of $S_{\mathsf{par}} \geq k b + t$ as
        \[
            \Pr[S_{\mathsf{par}} \geq n^3 b + t] \leq \exp(-2t^2/k)
        \]
    
        Now we show how this relates to the number of accepting repetitions in the permuted procedure. Observe that the parallel procedure can be equivalently stated as sampling a vector random variable $\mathsf{COUNT}\in \bbN^{|\cJ|}$ where $\mathsf{COUNT}[j]$ determines the number of times $(P_j, I-P_j)$ is applied, then randomly permuting the corresponding list of projectors. The probability density function is 
        \[
            \Pr[\mathsf{COUNT} = \mathsf{count}] = \prod_{j\in [|\cJ|]} p_j^{\mathsf{count}[j]}
        \]
        Let $\mathsf{count}^*$ be the count corresponding to the permuted procedure, i.e. 
        $\mathsf{count}^*[i] = \lfloor k p_i \rfloor$.
        Note that $\expect[\mathsf{COUNT}[i]] = k p_i$, so
        \[
            \|\expect[\mathsf{COUNT}[i]] - \mathsf{count}^*\|_1 
            \leq |\cJ|
        \]
        We claim that with overwhelming probability, 
        \[
            \| \mathsf{COUNT} - \bbE[\mathsf{COUNT}] \|_1 \leq \frac{a-b}{4} k
        \]
        To see this, consider each parallel repetition to sample an indicator vector indicating which term $j\in \cJ$ is chosen. Then \Cref{coro:vector-bernstein} implies that the probability of $\| \mathsf{COUNT} - \bbE[\mathsf{COUNT}] \|_1 \geq \frac{a-b}{4} k$ is at most
        \[
            |\cJ| \exp(\frac{-\left(\frac{a-b}{4} k\right)^2}{2(2k + \left(\frac{a-b}{4} k\right)/3)})
        \]
        since $k> n^3/(a-b)$ and $a-b\leq 1$, this probability is negligible in $n$.

        By triangle inequality, with overwhelming probability
        \begin{equation}\label{eq:permuting-difference}
            \| \mathsf{COUNT} - \expect[\mathsf{COUNT}] \|_1 
            \leq \lambda(a-b)/4 + |\cJ|
        \end{equation}

        \begin{claim}
            Consider two sequences of projective measurements $(\Pi_1^i)_{i\in [\lambda]}$ and $(\Pi_2^i)_{i\in [\lambda]}$ which are respectively applied to (disjoint) registers $\cR_i)_{i\in [\lambda]}$. Let $X_1$ and $X_2$ be the random variables denoting the number of times the measurement result is $1$ (corresponding to $\Pi_1^i$ or $\Pi_2^i$), respectively. If the number of indices $i$ such that $\Pi_1^i \neq \Pi_2^i$ is at most $\ell$, then for any state $\rho$ and any $n\in \bbN$,
            \[
                \Pr[X_1 \geq n] \leq \Pr[X_2 \geq n - \ell]
            \]
        \end{claim}
        \begin{proof}
            The two experiments can be thought of as performing two steps. In the first step we measure all indices where the two sequences match, and then in the second step we measure the remaining indices according to the first sequence in the first experiment, and according to the second sequence in the second experiment. We can obtain greater than equal to $n$ accepting measurements in the first experiment only if in the first step we obtain atleast $n-k$ accepting measurements. Let $X$ represent the number of accepting measurements in the first step. Therefore,
            \[
                \Pr[X_1 \geq n] \leq \Pr[X \geq n - \ell]
            \]
            In the second experiment, if we obtain $n-\ell$ accepting measurements in the first step, the final number of accepting measurements will be at least $n-k$. Therefore,
            \[
                \Pr[X_2 \geq n-k] \geq \Pr[X \geq n - \ell]
            \]
            Putting both together concludes the proof of the claim.
        \end{proof}

        For any permutation $\sigma$, state $\rho$, and vectors $\mathsf{count^*}$ and $\mathsf{COUNT}$,  consider the following experiments. 
        In the first experiment, we permute the list of measurements specified by $\mathsf{count^*}$ using the permutation $\sigma$ and apply the measurements to $\rho$. Let $v_0$ be the number of accepting measurements. 
        In the second experiment, we perform the parallel procedure: first sample $\mathsf{COUNT}$, then permute it randomly. Let $v_1$ be the number of accepting measurements in the second experiment.
        Let 
        \[
            \delta:=\| \mathsf{COUNT} - \mathsf{count}^* \|_1.
        \]
        By the above claim, for all $t$ and all $\delta$,  
        \[
            \Pr[v_0 \geq k b + t ] 
            \leq \Pr[v_1 \geq k b + t - \delta].
        \]
        
        Now consider sampling $\mathsf{COUNT}$ as in the parallel repeated experiment. 
        Recall from \cref{eq:permuting-difference} that with overwhelming probability
        \[
            \delta \leq k\frac{a-b}{4} + |\cJ|.
        \]
        Additionally, when $\sigma$ is also sampled randomly, $v_1$ is distributed as $S_{\text{par}}$ which means that 
        \[
            \Pr[v_1 \geq k b + t - \delta] \leq \exp\left(\frac{-2(t-\delta)^2}{k} \right).
        \]
        Setting $t$ to be $k \frac{a-b}{3} - |\cJ|$
        we get that with overwhelming probability 
        \[
            v_0 \leq k(b + a/3 - b/3) < k \frac{a+b}{2}.
        \]
        Finally, by noting that the probability that the permuting verifier accepts is at most the probability that $v_0$ exceeds $k\frac{a+b}{2}$, we obtain that the permuting verifier has negligible soundness error.
    \end{proof}
\end{proof}

\subsection{Permuting ZX Verifier for QMA}\label{sec:permutingzx}

We state here the full verification procedure for the ZX verifier with strong completeness that is guaranteed by \Cref{thm:ZX-strong} for every QMA language. As a corollary of \Cref{lem:permuting-QMA} and the ZX verifier from \cite{MNS16}, the following is a ZX verifier with strong completeness for any QMA language $\cL = (\cL_{yes},\cL_{no})$.

\begin{construction}[Permuting ZX Verifier]
    \label{def:pzxham}
    Let $H$ be an instance of the language $(\cL_{yes}, \cL_{no})$ with completeness and soundness energy thresholds $a',b' \in [-1,1]$. Without loss of generality~\cite{BL08}, $H$ is a ZX Hamiltonian
    \[
        H = \sum_{\substack{i < j < \ell \\ S \in \{Z, X\}}} p_{ij}P_{i, j, S} 
    \]
    on $\ell = \ell(n)$ qubits where $p_{ij} \in [0, 1]$ with $\sum_{i < j} 2p_{ij} = 1$ and $P_{i, j, S} = \frac{\identity + (-1)^{\beta_{i, j}}S_iS_j}{2}$ for $\beta_{i, j} \in \zo$.

    For each instance size $n \in \bbN$, define the following.
    \begin{itemize}
        \item Completeness and soundness thresholds $a = \frac{1}{2}\left(1 - \frac{a'}{\sum_{i<j}p_{i,j}}\right)$ and $b = \frac{1}{2}\left(1 - \frac{b'}{\sum_{i<j}p_{i,j}}\right)$ for \cite{MNS16}'s protocol. Note that $a, b \in [0,1]$.
        \item A number of repetitions $k\left\lceil \frac{4}{a-b}\left(\frac{n^3\ln(2)}{a-b}+ (2\ell^2) \right)\right\rceil + 1$.
        \item $\List_{H,k}$: A list of $(\theta_{i,j,S}, f_{i, j, S})$ for
        \begin{itemize}
            \item the basis $\theta_{i,j,S} = 0^n$ if $S = Z$, and $\theta_{i,j,S}= 1^n$ otherwise, and
            \item the function $f_{i,j,S}(m_1\Vert \ldots\Vert m_n)$ outputs $1$ iff  $m_i \xor m_j \xor \beta_{i,j} = 1$,
        \end{itemize}
        where each $(\theta_{i,j,S}, f_{i, j, S})$ appears $\lfloor p_{ij}k\rfloor$ times.
        
        \item $(\theta_{H, k, r}, f_{H, k, r}) \gets \Samp(H; r)$: On input an instance $H$ and randomness $r$, $\Samp$ samples a random permutation $\sigma \gets \mathsf{Sym}_k$ using randomness $r$, computes $\mathsf{PermList} = \sigma(\List_{H, k})$, and outputs $(\theta_{H, k, r}, f_{H,k, r})$ where
        \begin{itemize}
            \item $\theta_{H, k, r}$ as a concatenation of all $\theta_{i,j,S}$ in $\mathsf{PermList}$, and
            \item $f_{H, k, r}$ as dividing its input in-order amongst the $f_{i,j,S}$ in $\mathsf{PermList}$ and outputting $1$ iff at least $k \frac{a+b}{2}$ of $f_{i,j,S}$ accept their respective inputs.
        \end{itemize}

        \item For $H\in \cL_{yes}$, let $\ket{\psi}$ be a state such that $\Tr[\bra{\psi} H \ket{\psi}] \geq a$. Then $\ket{\psi}^{\otimes k}$ is a witness for the permuting ZX verifier.
    \end{itemize}
\end{construction}

\section{Coset State Authentication}
\label{sec:csa}

We recall the coset state authentication scheme, first introduced by \cite{BBV24}. We describe a variant of the scheme that does not involve CNOT-homomorphism, and where each qubit is encoded with an independently sampled subspace. We will use the following notation. Given a subspace $S \subset \bbF_2^{2\lambda+1}$ and a vector $\Delta \in \bbF_2^{2\lambda+1} \setminus S$, define the subspace $$S_\Delta \coloneqq S \cup (S+\Delta).$$ Let the dual subspace of $S_{\Delta}$ be $\widehat{S} \coloneqq S_\Delta^\bot$, let $\widehat{\Delta}$ be an arbitrary choice of a vector such that $S^\bot = \widehat{S} \cup (\widehat{S} + \widehat{\Delta})$, and define $$\widehat{S}_{\widehat{\Delta}} \coloneqq S^\bot = \widehat{S} \cup (\widehat{S} + \widehat{\Delta}).$$

Finally, given a projector $\Pi$ and a state $\ket{\psi}$, we write $\ket{\psi} \in \mathsf{im}(\Pi)$ to indicate that $\Pi\ket{\psi} = \ket{\psi}$.

\subsection{Construction}

\begin{construction}[Coset state authentication]
    \label{def:CSA}
    The coset state authentication scheme is defined by the following algorithms.
    \begin{itemize}
        \item $\KeyGen(1^\lambda,1^n)$: For each $i \in [n]$, sample a random subspace $S_i \subset \bbF_2^{2\lambda+1}$ of dimension $\lambda$, a random vector $\Delta \in \bbF_2^{2\lambda+1} \setminus S_i$, and random vectors $x_i,z_i \in \bbF_2^{2\lambda+1}$. Output $k \coloneqq \{S_i,\Delta_i,x_i,z_i\}_{i \in [n]}$. 
        \item $\Enc_k(\ket{\psi})$: Parametrized by a key $k = \{S_i,\Delta_i,x_i,z_i\}_{i \in [n]}$, the encoding algorithm is an $n$-qubit to $(2\secp+1)n$-qubit isometry that first applies \[\bigotimes_i\ket{b_i} \to \bigotimes_i \ket{S_i + b_i\Delta_i},\] and then applies the quantum one-time pad $X^xZ^z$, where $x = (x_1,\dots,x_n)$ and $z = (z_1,\dots,z_n)$.
        \item $\Dec_{k,\theta,f}(v) \to \{0,1\}$: Parameterized by a key $k = \{S_i,\Delta_i,x_i,z_i\}_{i \in [n]}$ and the description of a ZX measurement $\theta,f$, the decode algorithm takes as input a vector $v \in \bbF_2^{n \cdot (2\lambda+1)}$ and does the following. Parse $v = (v_1,\dots,v_n)$ where each $v_i \in \bbF_2^{2\lambda+1}$, and, for each $i \in [n]$, compute 
        \[m_i \coloneqq \begin{cases}0 \ \ \ \  \text{if} \ (\theta_i = 0 \ \text{and} \ v_i \in S_i+x_i) \ \text{or} \ (\theta_i = 1 \ \text{and} \ v_i \in \widehat{S}_i + z_i) \\ 1 \ \ \ \  \text{if} \ (\theta_i = 0 \ \text{and} \ v_i \in S_i+\Delta_i + x_i) \ \text{or} \ (\theta_i = 1 \ \text{and} \ v_i \in \widehat{S}_i + \widehat{\Delta}_i + z_i) \\ \bot \ \ \ \ \text{otherwise}\end{cases}.\] If any $m_i = \bot$, then output $0$. Otherwise output $f(m)$.
        \item $\Ver_{k,\theta}(v) \to \{0,1\}$: Parameterized by a key $k = \{S_i,\Delta_i,x_i,z_i\}_{i \in [n]}$ and bases $\theta \in \{0,1\}^n$, the verification algorithm takes as input a vector $v \in \bbF_2^{n \cdot (2\lambda+1)}$ and does the following. Parse $v = (v_1,\dots,v_n)$ where each $v_i \in \bbF_2^{2\lambda+1}$, and, for each $i \in [n]$, output 0 if $\theta_i = 0$ and $v_i \notin S_{i,\Delta_i}+x_i$ or $\theta_i = 1$ and $v_i \notin \widehat{S}_{i,\widehat{\Delta}_i}+z_i$. Otherwise, output 1.
    \end{itemize}
    
\end{construction}

\subsection{Properties}

We introduce new properties of this authentication scheme. First, we state some imported lemmas that follow from \cite{BBV24}.

\begin{lemma}[Correctness]\label{def:coset-auth-correctness}
    For any bases $\theta \in \{0,1\}^n$, function $f : \{0,1\}^n \to \{0,1\}$, and key $k \in \KeyGen(1^\secp,1^n)$, 
    \[\Enc_k^\dagger (H^{\otimes 2\secp+1})^\theta \left(\sum_{v : \Dec_{k,\theta,f}(v) = 1}\ketbra{v}\right) (H^{\otimes 2\secp+1})^\theta \Enc_k = M[\theta,f].\] 
\end{lemma}
\begin{lemma}[Hiding]\label{lemma:csa-privacy}
    Let $\Obf$ be a sub-exponentially secure indistinguishability obfuscator (\cref{def:iO}), and suppose that sub-exponentially secure injective one-way functions exist. Then there exists polynomials $d(\cdot,\cdot), q(\cdot,\cdot)$ such that for any two $n$-qubit states $\ket{\psi_0},\ket{\psi_1}$, and QPT adversary $\cA$,
    \begin{align*}&\Bigg|\Pr\left[\cA(\ket*{\widetilde{\psi}_0},\widetilde{\Ver}) = 1 : \begin{array}{r} k \gets \KeyGen(1^d,1^n) \\ \ket*{\widetilde{\psi}_0} \gets \Enc_k(\ket{\psi_0}) \\ \widetilde{\Ver} \gets \Obf(1^q,\Ver_{k,(\cdot)}(\cdot)) \end{array}\right]\\ &- \Pr\left[\cA(\ket*{\widetilde{\psi}_1},\widetilde{\Ver}) = 1 : \begin{array}{r} k \gets \KeyGen(1^d,1^n) \\ \ket*{\widetilde{\psi}_1} \gets \Enc_k(\ket{\psi_1}) \\ \widetilde{\Ver} \gets \Obf(1^q,\Ver_{k,(\cdot)}(\cdot)) \end{array}\right]\Bigg| = 2^{-\Omega(\lambda)},\end{align*} where $d \coloneqq d(\secp,n)$, and $q \coloneqq q(\secp,n)$.
\end{lemma}

We next show the following characterizing the codespace.

\begin{lemma}
    \label{lem:csa-soundness}
    For any key $k \in \KeyGen(1^\secp,1^n)$, define \[\Pi_k \coloneqq X^xZ^z\ketbra{S}X^xZ^z + X^xZ^z\ketbra{S+\Delta}X^xZ^z\] to be projector onto the image of the isometry $\Enc_k$. Then \[\Pi_k = \left(H^{\otimes 2\secp+1}\right)^{1^n}\left(\sum_{v : \Ver_{k,1^n}(v) = 1}\ketbra{v}\right)\left(H^{\otimes 2\secp+1}\right)^{1^n}\left(\sum_{v : \Ver_{k,0^n}(v) = 1}\ketbra{v}\right).\]
\end{lemma}

\begin{proof}
    We show the claim for $n=1$, which naturally generalizes to any $n$. Given a key $k = (S,\Delta,x,z)$, define $\Pi[S_\Delta]$ to be the projector onto $v \in S_\Delta$, define $\Pi[S^\bot]$ analogously, and re-write the RHS on the final line of the claim as \[H^{\otimes 2\secp+1}X^z\Pi[S^\bot]X^zH^{\otimes 2\secp+1}X^x\Pi[S_\Delta]X^x \coloneqq V_k.\] Then it suffices to show that (i) $V_kX^xZ^z\ket{S} = X^xZ^z\ket{S}$, (ii) $V_kX^xZ^z\ket{S+\Delta} = X^xZ^z\ket{S+\Delta}$, and (iii) for any $\ket{\psi}$ such that $\Pi_k\ket{\psi} = 0$, $V_k\ket{\psi} = 0$. 

    The first two follow by inspection, so we just show (iii). Writing \[X^xZ^z\ket{\psi} = \sum_{v} \alpha_v\ket{v},\] we have that $\sum_{v \in S}\alpha_v = 0$ and $\sum_{v \in S+\Delta}\alpha_v = 0$. Then
    \begin{align*}
        V_k\ket{\psi} &= H^{\otimes 2\secp+1}X^z\Pi[S^\bot]X^zH^{\otimes 2\secp+1}X^x\Pi[S_\Delta]X^x\ket{\psi} \\
        &= H^{\otimes 2\secp+1}X^z Z^x\Pi[S^\bot]H^{\otimes 2\secp+1}\Pi[S_\Delta]\sum_{v}\alpha_v\ket{v} \\
        &=H^{\otimes 2\secp+1}X^z Z^x\Pi[S^\bot]H^{\otimes 2\secp+1}\sum_{v \in S_\Delta}\alpha_v\ket{v} \\
        &=H^{\otimes 2\secp+1}X^z Z^x\Pi[S^\bot]\sum_{v \in S_\Delta}\alpha_v\sum_{w}(-1)^{v \cdot w}\ket{w} \\
        &=H^{\otimes 2\secp+1}X^z Z^x\sum_{w \in S^\bot}\left(\left(\sum_{v \in S}\alpha_v\right) + (-1)^{\Delta \cdot w}\left(\sum_{v \in S+\Delta}\alpha_v\right)\right)\ket{w} \\
        &= 0.
    \end{align*}
\end{proof}

Finally, we prove a new privacy property of the coset state authentication scheme. 
\begin{theorem}[Measurement Indistinguishability]\label{thm:auth-security}
    Let $\Obf$ be a sub-exponentially secure indistinguishability obfuscator (\cref{def:iO}), and suppose that sub-exponentially secure injective one-way functions exist. Then there exists polynomials $d(\cdot,\cdot), q(\cdot,\cdot)$ such that for any bases $\theta \in \{0,1\}^n$, functions $f_0,f_1: \{0,1\}^n \to \{0,1\}$, and $n$-qubit states $\ket{\psi_0},\ket{\psi_1}$ such that $\ket{\psi_0} \in \mathsf{im}\left(\cI - M[\theta,f_0]\right)$ and $\ket{\psi_1} \in \mathsf{im}\left(\cI - M[\theta,f_1]\right)$, it holds that for any QPT adversary $\cA$, 
     \begin{align*}&\Bigg| \Pr\left[\cA(\ket*{\widetilde{\psi}_0},\widetilde{\Ver},\widetilde{\Dec}_0) = 1 : \begin{array}{r}k \gets \KeyGen(1^d,1^n) \\ \ket*{\widetilde{\psi}_0} \gets \Enc_k(\ket{\psi_0}) \\ \widetilde{\Ver} \gets \Obf(1^q,\Ver_{k,(\cdot)}(\cdot)) \\ \widetilde{\Dec}_0 \gets \Obf(1^q,\Dec_{k,\theta,f_0}(\cdot))\end{array}\right]\\ &- \Pr\left[\cA(\ket*{\widetilde{\psi}_1},\widetilde{\Ver},\widetilde{\Dec}_1) = 1 : \begin{array}{r}k \gets \KeyGen(1^d,1^n) \\ \ket*{\widetilde{\psi}_1} \gets \Enc_k(\ket{\psi_1}) \\ \widetilde{\Ver} \gets \Obf(1^q,\Ver_{k,(\cdot)}(\cdot)) \\ \widetilde{\Dec}_1 \gets \Obf(1^q,\Dec_{k,\theta,f_1}(\cdot))\end{array}\right]\Bigg| = 2^{-\Omega(\lambda)},\end{align*} where $d \coloneqq d(\secp,n)$, and $q \coloneqq q(\secp,n)$.
\end{theorem}

\begin{proof}

Let $d = 2 \cdot \max\{n^2,\secp\}$ and $q = p(d)$, where $p$ is the polynomial from \cref{lemma:shO}. Now, for each $b \in \{0,1\}$, we proceed via the following sequence of hybrids.

\begin{itemize}
    \item $\Hyb_{0,b}$: This is the distribution over the output of $\cA$ as defined in the lemma statement using $f_b$ and $\ket{\psi_b}$.
    \item $\Hyb_{1,b}$: We ``bloat'' the subspaces used by $\Ver$. Given $k = \{S_i,\Delta_i,x_i,z_i\}_{i \in [n]}$, define $k' \gets \Bloat(k)$ to be the following procedure.

    \begin{itemize}
        \item For each $i \in [n]$, sample $T_i \gets \Sup_{3d/2+1}(S_{i,\Delta_i}),R_i \gets \Sup_{3d/2+1}(\widehat{S}_{i,\widehat{\Delta}_i})$.
        \item Output $k' \coloneqq \{T_i,R_i,x_i,z_i\}_{i \in [n]}$.
    \end{itemize}

    Then, define $\Ver'_{k',\theta}(v)$ as follows.
    \begin{itemize}
        \item Parse $v = (v_1,\dots,v_n)$.
        \item For each $i \in [n]$, output $0$ if $\theta_i = 0$ and $v_i \notin T_i + x_i$ or $\theta_i = 1$ and $v_i \notin R_i+z_i$. Otherwise, output 1. 
    \end{itemize}

    Finally, this hybrid is defined as follows.

    \begin{itemize}
        \item $k \gets \KeyGen(1^d,1^n)$
        \item $\ket*{\widetilde{\psi}_b} \gets \Enc_k(\ket{\psi_b})$
        \item $k' \gets \Bloat(k)$
        \item $\widetilde{\Ver'} \gets \Obf(1^q,\Ver'_{k',(\cdot)}(\cdot))$ 
        \item $\widetilde{\Dec}_b \gets \Obf(1^q,\Dec_{k,\theta,f_b}(\cdot))$
        \item Output $\cA(\widetilde{H^\theta\ket{x}},\widetilde{\Ver'},\widetilde{\Dec}_b)$
    \end{itemize}

    \item $\Hyb_{2,b}$: We measure $\ket{\psi_b}$ in the $\theta$-basis before encoding. Let $M_\theta$ denote the $n$-qubit measurement that measures the $i$'th qubit in basis $\theta_i$, and $H^\theta\ket{x} \gets M_\theta(\ket{\psi_b})$ denote the process of applying $M_\theta$ to the state $\ket{\psi_b}$. Then $\Hyb_{2,b}$ is defined as follows.
    \begin{itemize}
        \item $k \gets \KeyGen(1^d,1^n)$
        \item $H^\theta\ket{x} \gets M_\theta(\ket{\psi_b})$
        \item $\widetilde{H^\theta\ket{x}} \gets \Enc_k(H^\theta\ket{x})$
        \item $k' \gets \Bloat(k)$
        \item $\widetilde{\Ver'} \gets \Obf(1^q,\Ver'_{k',(\cdot)}(\cdot))$ 
        \item $\widetilde{\Dec}_b \gets \Obf(1^q,\Dec_{k,\theta,f_b}(\cdot))$
        \item Output $\cA(\widetilde{H^\theta\ket{x}},\widetilde{\Ver'},\widetilde{\Dec}_b)$
    \end{itemize}
    \item $\Hyb_{3,b}$: Sample $\widetilde{\Dec}_b$ as $\widetilde{\Dec}_b \gets \Obf(1^q,\mathsf{null})$, where $\mathsf{null}$ is the circuit that always outputs 0.
    \item $\Hyb_{4,b}$: Undo the measurement from $\Hyb_{2,b}$.
    \item $\Hyb_{5,b}$: Undo the change in $\Ver_{k,(\cdot)}(\cdot)$ from $\Hyb_{1,b}$.
\end{itemize}

The proof follows by combining the following sequence of claims.

\begin{claim}\label{claim:hyb1}
    For any $b \in \{0,1\}$, 
    \[|\Pr[\Hyb_{0,b} = 1] - \Pr[\Hyb_{1,b} = 1]| = 2^{-\Omega(\secp)}.\]
\end{claim}

\begin{proof}
    This follows directly by applying the security of subspace-hiding obfuscation (\cref{lemma:shO}) for each $i \in [n]$, and using the fact that $d/2 \geq \secp$.
\end{proof}

\begin{claim}\label{claim:hyb2}
    For any $b \in \{0,1\}$, 
    \[|\Pr[\Hyb_{1,b} = 1] - \Pr[\Hyb_{2,b} = 1]| = 2^{-\Omega(\secp)}.\]
\end{claim}

\begin{proof}
    Fix any choice of $y \in \{0,1\}^n$ such that $f_b(y) = 0$, any choice of $y' \neq y,$ and consider the following experiment $\Exp_0$.\\
    
    \noindent\underline{$\Exp_0$}

    \begin{itemize}
        \item $k \gets \KeyGen(1^d,1^n)$
        \item $\widetilde{H^\theta\ket{y}} \gets \Enc_k(H^\theta\ket{y})$
        \item $k' \gets \Bloat(k)$
        \item $\widetilde{\Ver'} \gets \Obf(1^q,\Ver'_{k',(\cdot)}(\cdot))$
        \item $\widetilde{\Dec}_b \gets \Obf(1^q,\Dec_{k,\theta,f_b}(\cdot))$
        \item $(v_1,\dots,v_n) \gets \cA(\widetilde{H^\theta\ket{y}},\widetilde{\Ver'},\widetilde{\Dec}_b)$
        \item Output 1 if 
        \begin{itemize}
            \item for all $i : \theta_i = 0$, $v_i \in S_i + y'_i \cdot \Delta_i + x_i$, and 
            \item For all $i : \theta_i = 1$, $v_i \in \widehat{S}_i + y'_i \cdot \widehat{\Delta}_i + z_i$.
        \end{itemize}
    \end{itemize}

    By \cref{lemma:state-decomp}, it suffices to show that \[\sqrt{\Pr[\Exp_0 = 1]} \cdot 2^n = 2^{-\Omega(\secp)}.\]
    Next, we'll define experiment $\Exp_1$ with a less-restrictive win condition.

    \noindent\underline{$\Exp_1$}

    \begin{itemize}
        \item $k \gets \KeyGen(1^d,1^n)$
        \item $\widetilde{H^\theta\ket{y}} \gets \Enc_k(H^\theta\ket{y})$
        \item $k' \gets \Bloat(k)$
        \item $\widetilde{\Ver'} \gets \Obf(1^q,\Ver'_{k',(\cdot)}(\cdot))$
        \item $\widetilde{\Dec}_b \gets \Obf(1^q,\Dec_{k,\theta,f_b}(\cdot))$
        \item $(v_1,\dots,v_n) \gets \cA(\widetilde{H^\theta\ket{y}},\widetilde{\Ver'},\widetilde{\Dec}_b)$
        \item Output 1 if there exists an $i \in [n]$ such that 
        \begin{itemize}
            \item if $\theta_i = 0$, then $v_i \in S_i + y'_i \cdot \Delta_i + x_i$, or
            \item if $\theta_i = 1$, then $v_i \in \widehat{S}_i + y'_i \cdot \widehat{\Delta}_i + z_i$.
        \end{itemize}
    \end{itemize}

    Clearly, we have that \[\Pr[\Exp_0 = 1] \leq \Pr[\Exp_1 = 1].\]
    
    Next, we change variables to make the notation more convenient. Define $f[y] \coloneqq f_b \oplus y$, let $H^{\theta,d} \coloneqq (H^{\otimes 2d+1})^{\theta_1} \otimes \dots \otimes (H^{\otimes 2d+1})^{\theta_n}$, and consider the following experiment.\\

    \noindent\underline{$\Exp_2$}

    \begin{itemize}
        \item $k \gets \KeyGen(1^d,1^n)$
        \item $\widetilde{\ket{0^n}} \gets \Enc_k(\ket{0^n})$
        \item $k' \gets \Bloat(k)$
        \item $\widetilde{\Ver'} \gets \Obf(1^q,\Ver'_{k',(\cdot)}(\cdot))$
        \item $\widetilde{\Dec} \gets \Obf(1^q,\Dec_{k,0^n,f[y]}(\cdot))$
        \item $(v_1,\dots,v_n) \gets \cA(H^{\theta,d}\widetilde{\ket{0^n}},\widetilde{\Ver},\widetilde{\Dec})$.
        \item Output 1 if there exists an $i \in [n]$ such that $v_i \in S_i + \Delta_i + x_i$
    \end{itemize}

    Since this is just a change of variables, we have that 

    \[\Pr[\Exp_1 = 1] = \Pr[\Exp_2 = 1].\]

    Next, we consider $n$ experiments $\Exp_{2,1},\dots,\Exp_{2,n}$, where in $\Exp_{2,j}$, the circuit $\Dec_{k,0^n,f[y],j}$, defined as follows, is obfuscated.\\

    \noindent\underline{$\Dec_{k,0^n,f[y],j}$}

    \begin{itemize}
        \item Parse $v = (v_1,\dots,v_n)$ where each $v_i \in \bbF_2^{2d+1}$
        \item For each $i \in [0,\dots,j]$, compute 
        \[m_i \coloneqq \begin{cases}0 \ \ \ \  \text{if} \ v_i \in S_i+x_i \\ \bot \ \ \ \ \text{otherwise}\end{cases}.\] 
        \item For each $i \in [j+1,\dots,n]$, compute 
        \[m_i \coloneqq \begin{cases}0 \ \ \ \  \text{if} \ v_i \in S_i+x_i \\ 1 \ \ \ \  \text{if} \ v_i \in S_i+\Delta_i + x_i  \\ \bot \ \ \ \ \text{otherwise}\end{cases}.\] 
        \item If any $m_i = \bot$, then output $0$. Otherwise output $f[y](m)$.
    \end{itemize}

    Note that $\Delta_i$ is sampled as a uniformly random coset of $S_i$ within $T_i$, which is a set of size $d/2 + 1$. Thus, by \cref{thm:point-obf}, we have that \[|\Pr[\Exp_{2,j-1} = 1] - \Pr[\Exp_{2,j} = 1]| = 2^{-\Omega(d)}\] for each $j \in [n]$.

    Next, since $f[y](0^n) = f(y) = 0$, $\Dec_{k,0^n,f[y],n}$ is the null circuit, and thus \[\Pr[\Exp_{2,n}=1] \leq \frac{|S_i|}{|T_i \setminus S_i|} = 2^{-\Omega(d)}.\] 

    Combining everything so far, we have that $\Pr[\Exp_0 = 1] = 2^{-\Omega(d)}$, which implies that 
    \[\sqrt{\Pr[\Exp_0 = 1]} \cdot 2^n = 2^{-\Omega(d)} \cdot 2^n = 2^{-\Omega(\secp)},\] since $d \geq n^2$.

\end{proof}

\begin{claim}\label{claim:hyb3}
    For any $b \in \{0,1\}$, 
    \[|\Pr[\Hyb_{2,b} = 1] - \Pr[\Hyb_{3,b} = 1]| = 2^{-\Omega(\secp)}.\]
\end{claim}

\begin{proof}
    This follows via the same sequence of hybrids $\Exp_{2,1},\dots,\Exp_{2,n}$ used in the proof of \cref{claim:hyb2}.
\end{proof}

\begin{claim}\label{claim:hyb4}
    For any $b \in \{0,1\}$, 
    \[|\Pr[\Hyb_{3,b} = 1] - \Pr[\Hyb_{4,b} = 1]| = 2^{-\Omega(\secp)}.\]
\end{claim}

\begin{proof}
    This follows from the same argument as the proof of \cref{claim:hyb2}
\end{proof}

\begin{claim}\label{claim:hyb5}
    For any $b \in \{0,1\}$, 
    \[|\Pr[\Hyb_{4,b} = 1] - \Pr[\Hyb_{5,b} = 1]| = 2^{-\Omega(\secp)}.\]
\end{claim}

\begin{proof}
    This follows from the same argument as the proof of \cref{claim:hyb1}
\end{proof}

\begin{claim}\label{claim:hybfinal} 
    \[|\Pr[\Hyb_{5,0} = 1] - \Pr[\Hyb_{5,1} = 1]| = 2^{-\Omega(\secp)}.\]
\end{claim}

\begin{proof}
    This follows directly from \cref{lemma:csa-privacy}, since we have exactly the same setup except that the adversary also receives an obfuscated null circuit.
\end{proof}

\end{proof}

\begin{corollary}
    \label{cor:csa-decoder-privacy}
    Let $\Obf$ be a sub-exponentially secure indistinguishability obfuscator (\cref{def:iO}), and suppose that sub-exponentially secure injective one-way functions exist. Then there exists a polynomial $p(\cdot,\cdot)$ such that for any bases $\theta \in \{0,1\}^n$, functions $f_0,f_1: \{0,1\}^n \to \{0,1\}$, and $n$-qubit states $\ket{\psi_0},\ket{\psi_1}$ such that $\ket{\psi_0} \in \mathsf{im}\left(M[\theta,f_0]\right)$ and $\ket{\psi_1} \in \mathsf{im}\left(M[\theta,f_1]\right)$, it holds that for any QPT adversary $\cA$, 
     \begin{align*}&\Bigg| \Pr\left[\cA(\ket*{\widetilde{\psi}_0},\widetilde{\Ver},\widetilde{\Dec}_0) = 1 : \begin{array}{r}k \gets \KeyGen(1^{p(\lambda,n)},1^n) \\ \ket*{\widetilde{\psi}_0} \gets \Enc_k(\ket{\psi_0}) \\ \widetilde{\Ver} \gets \Obf(\Ver_{k,(\cdot)}(\cdot)) \\ \widetilde{\Dec}_0 \gets \Obf(\Dec_{k,\theta,f_0}(\cdot))\end{array}\right]\\ &- \Pr\left[\cA(\ket*{\widetilde{\psi}_1},\widetilde{\Ver},\widetilde{\Dec}_1) = 1 : \begin{array}{r}k \gets \KeyGen(1^{p(\lambda,n)},1^n) \\ \ket*{\widetilde{\psi}_1} \gets \Enc_k(\ket{\psi_1}) \\ \widetilde{\Ver} \gets \Obf(\Ver_{k,(\cdot)}(\cdot)) \\ \widetilde{\Dec}_1 \gets \Obf(\Dec_{k,\theta,f_1}(\cdot))\end{array}\right]\Bigg| = 2^{-\Omega(\lambda)}.\end{align*}
\end{corollary}

\begin{proof}
    Suppose there exist $\theta,f_0,f_1,\ket{\psi_0},\ket{\psi_1},\cA$ for which the above does not hold. Define $\overline{f}_0,\overline{f}_1$ to be the complements of $f_0,f_1$, and note that $\ket{\psi_0} \in \mathsf{im}\left(\cI-M[\theta,\overline{f}_0]\right)$, $\ket{\psi_1} \in \mathsf{im}\left(\cI-M[\theta,\overline{f}_1]\right)$. For any $b \in \{0,1\}$, given $\ket*{\widetilde{\psi}_b},\widetilde{\Ver},\widetilde{\Dec}'_b$, where $k \gets \KeyGen(1^{p(\lambda,n)},1^n)$, $\ket*{\widetilde{\psi}_b} \gets \Enc_k(\ket{\psi_b})$, $\widetilde{\Ver} \gets \Obf(\Ver_{k,(\cdot)}(\cdot))$, and $\widetilde{\Dec}'_b \gets \Obf(\Dec_{k,\theta,\overline{f}_b}(\cdot))$, consider a reduction that samples \[\widetilde{\Dec_b}(\cdot) \gets \Obf\left(1 \ \text{if} \ \widetilde{\Dec}'_b(\cdot) = 0 \wedge \widetilde{\Ver}(\cdot) = 1\right)\] and runs $\cA$ on $\ket*{\widetilde{\psi}_b},\widetilde{\Ver},\widetilde{\Dec}_b$. Note that the program obfuscated is functionally-equivalent to $\Dec_{k,\theta,f_b}(\cdot)$, and thus by the security of $\Obf$, this reduction violates \cref{thm:auth-security}. 
\end{proof}

\section{Post-Quantum NIZK Arguments of Knowledge for NP}
\label{sec:nizk-np}

\subsection{Knowledge Soundness Definition}

\begin{definition}[Post-Quantum NIZKPoK (AoK) for $\NP$ in CRS Model]
\label{def:nizkpok-np}
Let $\NP$ relation $\cR$ with corresponding language $\cL$ be given such that they can be indexed by a security parameter $\lambda \in \bbN$.

$\Pi = (\Setup, \sP, \sV)$ is a post-quantum, non-interactive, zero-knowledge proof (argument) of knowledge for $\NP$ in the CRS model if it has the following syntax and properties.

\noindent \textbf{Syntax.}
The input $1^\lambda$ is left out when it is clear from context.
\begin{itemize}
    \item $\crs \gets \Setup(1^\lambda)$: The probabilistic polynomial-size circuit $\Setup$ on input $1^\lambda$ outputs a common reference string $\crs$.
    \item $\pi \gets \sP(1^\lambda, \crs, x, w)$: The probabilistic polynomial-size circuit $\sP$ on input a common random string $\crs$ and instance and witness pair $(x, w) \in \cR_\lambda$, outputs a proof $\pi$.
    \item $\sV(1^\lambda, \crs, x, \pi) \in \zo$: The probabilistic polynomial-size circuit $\sV$ on input a common random string $\crs$, an instance $x$, and a proof $\pi$ outputs $1$ iff $\pi$ is a valid proof for $x$.
\end{itemize}

\noindent \textbf{Properties.}
\begin{itemize}
    \item {\bf Uniform Random String.}
    $\Pi$ satisfies the uniform random string property of \cref{def:nizk-np}.

    \item {\bf Perfect Completeness.}
    $\Pi$ satisfies the perfect completeness property of \cref{def:nizk-np}.

    \item {\bf Adaptive $T$-Proof (Argument) of Knowledge.}
    There exists a polynomial-size circuit extractor $\Ext = (\Ext_0, \Ext_1)$ and a negligible functions $\negl_0(\cdot), \negl_1(\cdot)$ such that:
    \begin{enumerate}
        \item for every unbounded (polynomial-size) quantum circuit $\cD$, every sufficiently large $\lambda \in \bbN$, 
        \begin{equation*}
            \left|\Pr_{\crs \gets \Setup(1^\lambda)}[\cD(\crs) = 1] - \Pr_{(\crs,\td) \gets \Ext_0(1^\lambda)}[\cD(\crs) = 1] \right| \le \negl_0(T(\lambda))
        \end{equation*}
        \item and, for every unbounded (polynomial-size) quantum circuit $\cA$, every sufficiently large $\lambda \in \bbN$,
        \begin{equation*}
            \Pr_{\substack{(\crs, \td) \gets \Ext_0(1^\lambda) \\ (x, \pi) \gets \cA(\crs) \\ w \gets \Ext_1(\crs, \td, x, \pi)}}\left[\sV(\crs, x, \pi) = 1 \wedge (x, w) \not\in \cR_\lambda\right] \le \negl_1(T(\lambda)).
        \end{equation*}
    \end{enumerate}

    \item {\bf Adaptive Computational (Non-Adaptive Statistical) Zero-Knowledge.}
    $\Pi$ satisfies the adaptive computational (non-adaptive statistical) zero-knowledge property of \cref{def:nizk-np}.
\end{itemize}
\end{definition}

\subsection{Proofs of Knowledge for NP with CRS}

Let $\NP$ relation $\cR$ with corresponding language $\cL$ be given such that they can be indexed by a security parameter $\lambda \in \bbN$. Let $\Pi$ be a post-quantum NIZK for $\NP$ (\cref{def:nizk-np}). Let $(\Gen, \Enc, \Dec)$ be a IND-CPA post-quantum encryption scheme.

$\Setup(1^{\lambda})$:
\begin{enumerate}
    \item Generate a CRS for the NIZK scheme $\Pi$. Formally,
    \begin{enumerate}
        \item $\crs_\Pi \gets \Pi.\Setup(1^\lambda)$.
    \end{enumerate}
    \item Sample a public and secret key for the encryption scheme.
    \begin{enumerate}
        \item $(\pk, \sk) \gets \Gen(1^\lambda)$.
    \end{enumerate}
    \item Output $\crs = (\crs_\Pi, \pk)$.
\end{enumerate}

$\sP(\crs, x, w)$:
\begin{enumerate}
    \item Compute an encryption of the witness. Formally,
    \begin{enumerate}
        \item Compute $\ct = \Enc(\pk, w; r)$ for uniformly random $r$.
    \end{enumerate}
    
    \item Compute a NIZK to prove that the ciphertext contains a witness for the instance. Formally,
    \begin{enumerate}
        \item Let $\cL_\Enc$ be an $\NP$ language defined as follows
        \begin{equation*}
            \cL_\Enc = \left\{(x, \ct) \st \exists (w, r), 
            \begin{array}{c}
                (x, w) \in \cR_\lambda \text{ and}\\
                \ct = \Enc(\pk, w; r)
            \end{array}\right\}.
        \end{equation*}
        \item Compute $\pi_\Pi \gets \Pi.\sP(\crs_\Pi, (x, \ct), (w, r))$ with respect to language $\cL_\Enc$.
    \end{enumerate}
    \item Output $\pi = (\ct, \pi_\Pi)$.
\end{enumerate}

$\sV(\crs, x, \pi)$:
\begin{enumerate}
    \item Verify that the NIZK verifier accepts the NIZK proof. Formally,
    \begin{enumerate}
        \item Verify that $\Pi.\sV(\crs_\Pi, (x, \ct), \pi_\Pi) = 1$.
    \end{enumerate}
    \item Output $b$.
\end{enumerate}

\begin{theorem}
    \label{thm:nizkpok-np}
    Given that
    \begin{itemize}
        \item $\Pi$ is a post-quantum adaptively statistically sound, (adaptively) computationally zero-knowledge NIZK protocol for $\NP$ with common reference string (\cref{def:nizk-np}) and
        \item  $(\Gen, \Enc, \Dec)$ is a perfectly-correct post-quantum IND-CPA encryption scheme,
    \end{itemize}
    then this construction is a post-quantum adaptive proof of knowledge, (adaptively) computationally zero-knowledge NIZK for $\NP$ with common reference string (\cref{def:nizkpok-np}).
\end{theorem}

\begin{proof}
    \textbf{Perfect Completeness.}
    This follows from perfect completeness of $\Pi$.

    \textbf{Adaptive Proof of Knowledge.}
    We define $\Ext_0$ as follows:
    \begin{addmargin}[2em]{2em} % left, right
        \noindent \emph{Input}: $1^\lambda$.

        \noindent \textbf{(1)} $\crs_\Pi \gets \Pi.\Setup(1^\lambda)$.
        
        \noindent \textbf{(2)} $(\pk, \sk) \gets \Gen(1^\lambda)$.

        \noindent \textbf{(3)} Output $\crs = (\crs_\Pi, \pk)$ and $\td = \sk$.
    \end{addmargin}
    We define $\Ext_1$ as follows:
    \begin{addmargin}[2em]{2em} % left, right
        \noindent \emph{Input}: $\crs, \td, x, \pi$.
        
        \noindent \textbf{(1)} Output $w \defeq \Dec(\sk, \ct)$.
    \end{addmargin}

    Since $\Ext_0$ and $\Setup$ output $\crs$ from identical distributions, we have that for every unbounded (polynomial-size) quantum circuit $\cD$, every sufficiently large $\lambda \in \bbN$,
        \begin{equation*}
            \left|\Pr_{\pp \gets \Setup(1^\lambda)}[\cD(\pp) = 1] - \Pr_{(\pp,\td) \gets \Ext_0(1^\lambda)}[\cD(\pp) = 1] \right| = 0.
        \end{equation*}

    We now argue by contradiction that the extractor $(\Ext_0, \Ext_1)$ satisfies the second property. Let a polynomial $p(\cdot)$ and an oracle-aided polynomial-size quantum circuit $\cA$ be given such that for every sufficiently large $\lambda \in \bbN$,
    \begin{equation*}
        \Pr_{\substack{(\crs, \td) \gets \Ext_0(1^\lambda) \\ (x, \pi) \gets \cA(\crs) \\ w \gets \Ext_1(\crs, \td, x, \pi)}}\left[\sV(\crs, x, \pi) = 1 \wedge (x, w) \not\in \cR_\lambda \right] \ge \frac{1}{p(\lambda)}.
    \end{equation*}

    We consider two scenarios: either $(x, \ct) \not\in \cL_\Enc$, or $(x, \ct) \in \cL_\Enc$. At least one of these scenarios must occur with at least $1/(2p(\lambda))$ probability. We will show that both scenarios reach a contradiction.

    \noindent \underline{Scenario One}

    Consider that 
    \begin{equation*}
        \Pr_{\substack{(\crs, \td) \gets \Ext_0(1^\lambda) \\ (x, \pi) \gets \cA(\crs) \\ w \gets \Ext_1(\crs, \td, x, \pi)}}\left[\sV(\crs, x, \pi) = 1 \wedge (x, w) \not\in \cR_\lambda \wedge (x, \ct) \not\in \cL_\Enc\right]  \ge \frac{1}{2p(\lambda)}.
    \end{equation*}

    If the verifier $\sV$ accepts a proof $\pi$, then the verifier of the NIZK $\Pi.\sV$ accepts the proof $\pi_\Pi$. Hence,
    \begin{equation*}
        \Pr_{\substack{(\crs, \td) \gets \Ext_0(1^\lambda) \\ (x, \pi) \gets \cA(\crs) \\ w \gets \Ext_1(\crs, \td, x, \pi)}}\left[\Pi.\sV(\crs_\Pi, (x, \ct), \pi_\Pi) = 1 \wedge (x, w) \not\in \cR_\lambda \wedge  (x, \ct) \not\in \cL_\Enc)\right]  \ge \frac{1}{2p(\lambda)}.
    \end{equation*}

    However, this contradicts the adaptive statistical (computational) soundness of $\Pi$.
    
    \noindent \underline{Scenario Two}

    Consider that 
    \begin{equation*}
        \Pr_{\substack{(\crs, \td) \gets \Ext_0(1^\lambda) \\ (x, \pi) \gets \cA(\crs) \\ w \gets \Ext_1(\crs, \td, x, \pi)}}\left[\sV(\crs, x, \pi) = 1 \wedge (x, w) \not\in \cR_\lambda \wedge (x, \ct) \in \cL_\Enc\right]  \ge \frac{1}{2p(\lambda)}.
    \end{equation*}

    If $(x, \ct) \in \cL_\Enc$ then there exists $(w, r)$ such that $(x, w) \in \cR_\lambda$ and $\ct = \Enc(\pk, w; r)$. Coupled with the perfect correctness of the encryption scheme, this means that $(x, \Dec(\sk, \ct)) \in \cR_\lambda$. However, by the definition of $\Ext_1$, this contradicts with $(x, w) \not\in \cR_\lambda$.

    Therefore, our protocol must be an adaptive proof of knowledge.

    \textbf{(Adaptive) Computational Zero-Knowledge.}
    Let $\Pi.\Sim = (\Pi.\Sim_0, \Pi.\Sim_1)$ be the (adaptive) computational zero-knowledge simulator of the NIZK $\Pi$.
    We define $\Sim_0$ as follows:
    \begin{addmargin}[2em]{2em} % left, right
        \noindent \emph{Input}: $1^\lambda$
    
        \noindent \textbf{(1)} Compute $\crs_\Pi \gets \Pi.\Sim_0(1^\lambda)$.
        
        \noindent \textbf{(2)} Sample $(\pk, \sk) \gets \Gen(1^\lambda)$.

        \noindent \textbf{(3)} Output $\crs = (\crs_\Pi, \pk)$ and $\td = \sk$.
    \end{addmargin}
    We define $\Sim_1$ as follows:
    \begin{addmargin}[2em]{2em} % left, right
        \noindent \emph{Input}: $\crs$, $x$
    
        \noindent \textbf{(1)} Compute $\ct \gets \Enc(\pk, 0)$.
        
        \noindent \textbf{(2)} Compute $\pi_\Pi \gets \Pi.\Sim_1(\crs_\Pi, (x, \ct))$.

        \noindent \textbf{(3)} Output $\pi = (\ct, \pi_\Pi)$.
    \end{addmargin}

    Let a polynomial-size quantum circuit $\cD = (\cD_0, \cD_1)$, and sufficiently large $\lambda \in \bbN$ be given. We construct the following series of hybrids to argue computational indistinguishability of the honest $\cH_0$ and simulated $\cH_3$ distributions:
    
    \paragraph{$\cH_0$ :}
    $(\crs, \td) \gets \Setup(1^\lambda)$. $(x, w, \zeta) \gets \cD_0(\crs)$. $\pi \gets \sP(\crs, x, w)$.

    \paragraph{$\cH_1$ :}
    $\crs_\Pi \gets \Pi.\Sim_0(1^\lambda)$. $(\pk, \sk) \gets \Gen(1^\lambda)$. $\crs = (\crs_\Pi, \pk)$. $(x, w, \zeta) \gets \cD_0(\crs)$. $\ct \gets \Enc(\pk, w)$. $\pi_\Pi \gets \Pi.\Sim_1(\crs_\Pi, (x, \ct))$. $\pi = (\ct, \pi_\Pi)$.
    
    \paragraph{$\cH_2$ :}
    $\crs_\Pi \gets \Pi.\Sim_0(1^\lambda)$. $(\pk, \sk) \gets \Gen(1^\lambda)$. $\crs = (\crs_\Pi, \pk)$. $(x, w, \zeta) \gets \cD_0(\crs)$. $\ct \gets \Enc(\pk, 0)$. $\pi_\Pi \gets \Pi.\Sim_1(\crs_\Pi, (x, \ct))$. $\pi = (\ct, \pi_\Pi)$.
    
    \paragraph{$\cH_3$ :}
    $\crs \gets \Sim_0(1^\lambda)$. $(x, w, \zeta) \gets \cD_0(\crs)$. $\pi \gets \Pi.\Sim(\crs, x)$.

    $\cH_0$ and $\cH_1$ are computationally indistinguishable by the post-quantum adaptive computational zero-knowledge of $\Pi$. $\cH_1$ and $\cH_2$ are computationally indistinguishable by the post-quantum IND-CPA property of encryption. $\cH_2$ and $\cH_3$ are identical. Therefore, our protocol is adaptive computational zero-knowledge.
\end{proof}

\begin{corollary}
    \label{cor:nizkpok-np}
    Assuming the polynomial quantum hardness of LWE, there exists an adaptive proof of knowledge, adaptively computationally zero-knowledge NIZKPoK for $\NP$ having a common reference string (\cref{def:nizkpok-np}).
\end{corollary}

\begin{proof}
    This follows from \cref{thm:nizkpok-np}, \cref{thm:nizk-np-crs}, and \cref{thm:enc}.
\end{proof}

\subsection{Adaptively Sound Arguments for NP with URS}

It is well-known that any NIZK with non-adaptive, but sub-exponential, soundness can be compiled to an adaptively sound NIZK using complexity leveraging. We present its proof for completeness.

\begin{theorem}
    \label{thm:adapt-niszk-np-urs}
    Let $\cR$ be an $\NP$ relation indexed by $\lambda \in \bbN$ such that $\cR_\lambda$ has instance size $\lambda^c$ for some parameter $c < 1$.
    Assuming a non-interactive zero knowledge argument for $\NP$ with sub-exponential computational soundness error and statistical zero knowledge, then 
    for every $c < 1$, there exists an \emph{adaptively} sub-exponential computationally sound, (not necessarily adaptively) statistically zero-knowledge non-interactive protocol for $\cR$ having a uniform random string (\cref{def:nizk-np}).
\end{theorem}

\begin{proof}
    Let $\Pi$ be the non-interactive zero knowledge argument for $\NP$ with sub-exponential computational soundness error and statistical zero knowledge.
    Set $T(\lambda) = 2^{\lambda^c}$ and $\negl(T(\lambda))$ for the non-adaptive computational $T$-soundness of $\Pi$. This is possible since the NIZK $\Pi$ has sub-exponential soundness error. We show that $\Pi$ with these parameters is adaptively sound by reducing to the $\negl(2^{\lambda^c})$ soundness error.
    
    We can decompose any adversary's computational advantage into their advantage when conditioning on the instance they output to see that
    \begin{align*}
        &\Pr_{\substack{\crs \gets \Setup(1^\lambda) \\ (x, \pi) \gets \cA(\crs)}}\left[\begin{array}{c}
             \sV(\crs, x, \pi) = 1 \\
             \wedge\ x \not\in \cL_\lambda
        \end{array} \right]  \\
        &= \sum_{x\in \{0,1\}^{\lambda^c}: x\notin \cL_\lambda} \Pr_{\substack{\crs \gets \Setup(1^\lambda) \\ (x', \pi) \gets \cA(\crs)}}\left[\begin{array}{c}
             \sV(\crs, x', \pi) = 1 \\
             \wedge\ x' \not\in \cL_\lambda
        \end{array} \ \middle|\  x' = x \right] \cdot \Pr_{\substack{\crs \gets \Setup(1^\lambda) \\ (x', \pi) \gets \cA(\crs)}}[x = x']
        \\
        &\leq \sum_{x\in \{0,1\}^{\lambda^c}: x\notin \cL_\lambda} \negl(2^{\lambda^c}) \cdot \Pr_{\substack{\crs \gets \Setup(1^\lambda) \\ (x', \pi) \gets \cA(\crs)}}[x = x']
        \\
        &\leq \negl(2^{\lambda^c})
    \end{align*}
    where the first inequality follows from the non-adaptive $\negl(2^{\lambda^c})$ soundness error.
\end{proof}

\subsection{Arguments of Knowledge for NP with URS}

\begin{theorem}
    \label{thm:nizkaok-np}
    Let $\cR$ be an $\NP$ relation indexed by $\lambda \in \bbN$ such that $\cR_\lambda$ has instance size $\lambda^c$ for some parameter $c < 1$.
    Given that
    \begin{itemize}
        \item $\Pi$ is a post-quantum adaptively (sub-exponentially) computationally sound, computationally zero-knowledge NIZK protocol for $\cR$ with uniformly random string (\cref{def:nizk-np}) and
        \item  $(\Gen, \Enc, \Dec)$ is a post-quantum IND-CPA encryption scheme with uniformly random public keys,
    \end{itemize}
    then for every $c < 1$, this construction is a post-quantum adaptive (sub-exponentially) argument of knowledge, computationally zero-knowledge NIZKAoK for $\cR$ with uniformly random string (\cref{def:nizkpok-np}).
\end{theorem}

\begin{proof}
    This follows by observing the proof of \cref{thm:nizkpok-np}.
\end{proof}

\begin{corollary}
    \label{cor:nizkaok-np}
    Let $\cR$ be an $\NP$ relation indexed by $\lambda \in \bbN$ such that $\cR_\lambda$ has instance size $\lambda^c$ for some parameter $c < 1$.
    Assuming the sub-exponential quantum hardness of LWE, for every $c < 1$, there exists an adaptive sub-exponential argument of knowledge, statistically zero-knowledge non-interactive protocol for $\cR$ having a uniform random string (\cref{def:nizkpok-np}).
\end{corollary}

\begin{proof}
    This follows from \cref{thm:nizkaok-np}, \cref{thm:adapt-niszk-np-urs}, \cref{cor:subexp-nizk-np-urs}, and \cref{thm:enc}.
\end{proof}

\section{Provably-Correct Obfuscation}
\label{sec:prov-obf}
 \subsection{Definition}
We define a notion of provably-correct obfuscation with two security properties: (standard) indistinguishability security and a notion of \emph{secure composition} for obfuscating evasive families of circuits.

\begin{definition}[Provably-correct obfuscation]\label{def:provable-obfuscation}
    A provably-correct obfuscator has the following syntax.
    \begin{itemize}
        \item $\pp \gets \Setup(1^{\lambda_\ZK})$. 
        The setup algorithm takes as input the zero-knowledge security parameter $1^{\lambda_\ZK}$ and outputs public parameters $\pp$. We say that the obfuscator is in the \emph{URS} (uniform random string) model if $\pp$ just consists of uniform randomness.
        \item $\ObfC \gets \Obf(1^{\lambda_\OBF}, \pp, \pred, C)$.
        The obfuscation algorithm takes as input the obfuscation security parameter $1^{\lambda_\OBF}$, public parameters $\pp$, a predicate $\pred$, and a circuit $C$. It outputs an obfuscated circuit $\ObfC$ that satisfies predicate $\pred$.
        \item $y \gets \Eval(\ObfC,x)$. 
        The evaluation algorithm takes as input an obfuscated circuit $\ObfC$ and an input $x$ and outputs $y$.
        
        \item $\Ver(\pp, \pred, \ObfC) \in \zo$. 
        The verification algorithm takes as input the public parameters $\pp$, a predicate $\pred$, and an obfuscated circuit $\ObfC$, and outputs either accept or reject.
    \end{itemize}
    It should satisfy the following properties.

    \begin{itemize}
        \item \textbf{Functionality-preservation}. Let $C$ be any circuit and $\pred$ be any predicate such that $\pred(C) = 1$. For all $\pp$ and randomness $r$,  $\ObfC:= \Obf(1^{\lambda_\OBF}, \pp, \pred, C; r)$, and $x$, it holds that $\Eval(\ObfC,x) = C(x)$.
        
        \item \textbf{Completeness}. For any circuit $C$,
        \begin{equation*}
            \Pr_{\substack{\pp \gets \Setup(1^{\lambda_\ZK}) \\ \ObfC \gets \Obf(1^{\lambda_\OBF}, \pp, \pred, C)}}\left[\Ver(\pp, \pred, \ObfC) = 1 \right] = 1.
        \end{equation*}
        
        \item \textbf{Adaptive Knowledge Soundness}. There exist PPT algorithms $\Ext_0,\Ext_1$ such that \[\left\{\pp : \pp \gets \Setup(1^{\lambda_\ZK})\right\} \approx_{\negl(\secp_\ZK)} \left\{\pp : (\pp,\td) \gets \Ext_0(1^{\lambda_\ZK})\right\},\] and, for any QPT adversary $\cA$,
           \[\Pr\left[\begin{array}{l}\left(\Ver(\pp,\pred,\ObfC) = 0\right) \vee \\\left(\forall x, C(x) = \Eval(\ObfC,x) \wedge \pred(C) = 1\right)\end{array} : \begin{array}{r}(\pp,\td) \gets \Ext_0(1^{\lambda_\ZK}) \\ \ObfC \gets \cA(\pp) \\ C \gets \Ext_1(\pp, \td,\pred,\ObfC)\end{array}\right] = 1-\negl(\lambda_\ZK).\] If this property holds against \emph{unbounded} adversaries $\cA$, then we say the scheme has \textbf{statistical knowledge soundness}.
        
        \item {\bf Simulation Security}
        There exists a PPT simulator $\Sim = (\SimGen, \SimObf)$ such that the following properties hold.
        \begin{itemize}

            \item \textbf{Honest-to-Simulated Indistinguishability.}
            For every circuit $C$ and every polynomial-size predicate $\pred$ such that $\pred(C) = 1$,
            \begin{equation*}
                \left\{(\pp, \ObfC) : 
            \begin{array}{r}
                 \pp \gets \Setup(1^{\lambda_\ZK}) \\
                 \ObfC \gets \Obf(1^{\secp_\OBF},\pp, \pred, C)
            \end{array}
            \right\} \approx_{\negl(\secp_\ZK)}
            \left\{(\pp, \ObfC) : 
            \begin{array}{r}
                 (\pp, \td) \gets \SimGen(1^{\lambda_\ZK}) \\
                 \ObfC \gets \SimObf(1^{\secp_\OBF},\pp, \td, \pred, C)
            \end{array}
            \right\}.
            \end{equation*}

            \item \textbf{(Sub-exponential) Simulated-Circuit Indistinguishability}. There exists a constant $\epsilon > 0$ such that for every quantum polynomial-size circuit $\cA$, every $C_0,C_1$ such that $C_0 \equiv C_1$, and every polynomial-size predicate $\pred$,
            \begin{align*}
                &\Bigg|\Pr_{\substack{(\pp, \td) \gets \SimGen(1^{\lambda_\ZK})}}\left[\cA\left(\SimObf(1^{\secp_\OBF},\pp, \td,\pred,C_0)\right) = 1\right] \\
                &- \Pr_{\substack{(\pp, \td) \gets \SimGen(1^{\lambda_\ZK})}}\left[\cA\left(\SimObf(1^{\secp_\OBF},\pp,\td,\pred,C_1)\right) = 1\right]\Bigg| \leq 2^{-\lambda_\OBF^\epsilon}.
            \end{align*}

            \item \textbf{$\cS$-evasive composability}. This property is parameterized by a sampler $S$ that outputs a set of circuits $\{C_i\}_{i \in [N]}$ along with side information in the form of circuit $C$ and state $\ket{\psi}$. Let $C_{\mathsf{null}}$ be the null circuit, and, given a set of circuits $\{C_i\}_i$ and auxiliary circuit $C$, define the circuit $C||\Combine(\{C_i\}_i)$ to map $(i,x) \to C_i(x)$ for $i > 0$ and $(0,x)$ to $C(x)$. Likewise, let $C||C'$ be the circuit that maps $(0,x)$ to $C(x)$ and $(1,x)$ to $C'(x)$.  \\ 
        
        \textbf{IF} for any $i \in [N]$, any predicate $\pred$, and any QPT adversary $\cA$:
        \begin{align*}
        &\bigg|\Pr_{\substack{(\pp, \td) \gets \SimGen(1^{\lambda_\ZK}) \\ \left(\ket{\psi},C,\{C_j\}_j\right) \gets \cS}}\left[\cA\left(\ket{\psi},\SimObf\left(1^{\secp_\OBF},\pp, \td, \pred, C||C_i\right)\right) = 1\right]\\ &- \Pr_{\substack{(\pp, \td) \gets \SimGen(1^{\lambda_\ZK}) \\ \left(\ket{\psi},\{C_j\}_j\right) \gets \cS}}\left[\cA\left(\ket{\psi},\SimObf\left(1^{\secp_\OBF},\pp, \td,\pred,C||C_{\mathsf{null}}\right)\right) = 1\right]\bigg| = \negl(\lambda_\OBF)/N,\end{align*}
        \textbf{THEN}: 
        \begin{align*}
            &\bigg|\Pr_{\substack{(\pp, \td) \gets \SimGen(1^{\lambda_\ZK}) \\ \left(\ket{\psi},C,\{C_j\}_j\right) \gets \cS}}\left[\cA\left(\ket{\psi},\SimObf\left(1^{\secp_\OBF},\pp, \td, \pred,C||\Combine\left(\{C_i\}_i\right)\right)\right) = 1\right] \\ &- \Pr_{\substack{(\pp, \td) \gets \SimGen(1^{\lambda_\ZK}) \\ \left(\ket{\psi},\{C_j\}_j\right)\gets \cS}}\left[\cA\left(\ket{\psi},\SimObf\left(1^{\secp_\OBF},\pp, \td, \pred,C||\Combine(\{C_{\mathsf{null}}\}_i)\right)\right) = 1\right]\bigg| = \negl(\secp_\OBF).
        \end{align*}
        \end{itemize}
        
    \end{itemize}
\end{definition}

\subsection{Construction in QPrO Model}\label{subsec:qpro-construction}

In this section, we show how to modify the JLLW construction \cite{JLLW23} to permit provable correctness while still satisfying ideal obfuscation. The main technical nuance is that the keys and handles used in the construction may only be verified through an oracle interface.

First, we confirm in \Cref{lem:idealobfs-is-evasive} below that any \emph{ideal} obfuscator satisfies $\cS$-evasively composability for any choice of sampler $\cS$.

\begin{lemma}\label{lem:idealobfs-is-evasive}
    If $\SimObf$ is an ideal obfuscator, then it satisfies $S$-evasive composeability for all samplers $S$.
\end{lemma}
\begin{proof}
    Observe that the pre-condition of $\cS$-evasive composability implies that no adversary given $\hat{C} \gets \SimObf(1^{\secp_\OBF}, \pp, \td, \varphi, C_i)$ can find an input $x$ such that $C(x) \neq 0$, except with $\negl(\secp_\OBF)/N$ probability.
    
    Consider the following sequence of hybrids.
    \begin{itemize}
        \item $\Hyb_0$: The original experiment where the adversary is run on $\SimObf(1^{\secp_\OBF}, \pp, \td, \varphi, \Combine(\{C_i\}_i)$.
        \item $\Hyb_1$: Instead of running the adversary using the obfuscation, run it using the ideal-world simulator $\cS$. $\cS$ only has oracle access to $\Combine(\{C_j\}_j$.

        \item $\Hyb_2$: Replace the simulator's oracle access to $\Combine(\{C_j\}_j)$ by oracle access to the null program $C_\bot$.

        \item $\Hyb_3$: Replace the simulator by $\SimObf(1^{\secp_\OBF}, \pp, \td, \varphi, C_\bot)$.
    \end{itemize}

    $\Hyb_0 \approx \Hyb_1$ and $\Hyb_2\approx \Hyb_3$ by the security of ideal obfuscation. The main step is to show that $\Hyb_1 \approx \Hyb_2$. \Cref{lemma:oracle-ind} shows that any adversary distinguishing the two must have noticeable probability of querying an input where the oracles $\Combine(\{C_j\}_j$ and $C_\bot$ differ -- this holds in particular when considering the queries the adversary makes to $C_\bot$. Let $\epsilon$ be the probability of this occurring. Any differing input has the form $(i, x)$ such that $C_i(x) \neq \bot$. Since there are $N$ possible choices of $i$, there is at least one $i^*$ such that the probability of outputting a differing input with $i=i^*$ is $\geq \epsilon/N$. But then the adversary could find an input $x$ such that $C_{i^*}(x) \neq \bot$ with probability better than $\negl(\secp_\OBF)/N$ using only $\hat{C} \gets \SimObf(1^{\secp_\OBF}, \pp, \td, \varphi, C_{i^*})$, contradicting our earlier observation.
\end{proof}

Throughout the remainder of the section, we set $\secp \coloneqq \secp_\ZK$.

\begin{construction}[Provably-Correct Obfuscation in the $(\lambda+1)$-QPrO Model]\label{construction:provable-obfuscator-QPRO}
  We construct a provably-correct obfuscator in the quantum-accessible pseudorandom oracle model (\cref{def:QPRO}). The obfuscator is defined as follows in the $(\lambda+1)$-QPrO model defined in \cref{def:QPRO}, using PRF $H$.

 We first specify the ingredients:

\begin{itemize}
    \item Let $\Com$ be a sub-exponentially secure statistically binding non-interactive commitment scheme. 
    \item Let $\JLLWObf$ be the (sub-exponentially secure) obfuscation scheme specified in Construction \ref{construction:JLLW}. We will make the following modification to the definition of the obfuscator. Instead of sampling key and handle pairs $(k_{ij}, h_{ij})$ internally in the second step, $\JLLWObf$ instead takes the set of key and handle pairs as an additional input. Note that if these pairs are sampled honestly this does not affect correctness or security. 
    \item Let $(\NIZK.\Setup, \NIZK.\mathsf{P}, \NIZK.\mathsf{V})$ be a post-quantum non-interactive zero-knowledge argument of knowledge (\cref{def:nizkpok-np}) for an NP relation $\cR_\lambda$ to be specified later.
\end{itemize}
We now define the algorithms for the obfuscation scheme:
\begin{itemize}
    \item $\Setup^\QPRO(1^{\secp})$:
    \begin{itemize}
        \item $\crs \gets \NIZK.\Setup(1^{\secp})$
        \item $h^* \gets \{0,1\}^{\secp}$
        \item Return $\pp := (\crs, h^*)$
    \end{itemize}
    \item $\Obf^\QPRO(1^{\secp_\OBF}, \pp,C)$:
    \begin{itemize}
        \item Parse $\pp$ as $(\crs, h^*)$
        \item Let $D$ be the input length of  $C$ and let $B$ be the number of blocks (determined as in \cite{JLLW23}). 
        \item For each $t \in [\secp]$:
        \begin{itemize}
            \item $k^t_{i,j} \gets \{0,1\}^{\secp_\OBF}$  for $0 \leq i < D, 1 \leq j \leq B$
            \item $h^t_{i,j} \gets \QPRO_t(\Gen,k^t_{i,j})$  for $0 \leq i < D, 1 \leq j \leq B$
            \item $\overline{k}_t := \{k^t_{i,j}\}_{i,j}$
            \item $\overline{h}_t := \{h^t_{i,j}\}_{i,j}$
            \item $r_t \leftarrow \{0,1\}^*$
            \item $c_t \gets \com(\overline{k}_t; r_t)$
        \end{itemize}
        \item $\chal := \QPRO_0(\Eval, h^*, (c_1, \ldots, c_\lambda, \overline{h}_1, \ldots, \overline{h}_\lambda))$ where $\chal \in \{0,1\}^\secp$
        \item Let $\Opened := \{t : \chal_t = 1\}$. For $t \notin \Opened:$
        \begin{itemize}
            \item $\Obfr_t \leftarrow \{0,1\}^*$
            \item $\ObfC_t \gets \JLLWObf(1^\lambda, C, \overline{k}_t, \overline{h}_t; \Obfr_t)$. We note here that since we provide key and handle pairs as input, $\JLLWObf$ does not query the $\QPRO$ oracle.
        \end{itemize}
        \item For any $x$ of the form $(\pred, \chal, \{c_t,\overline{h}_t,\ObfC_t\}_{t\notin \Opened})$ and $w$ of the form $(C, \{r_t, \overline{k}_t, \Obfr_t\}_{t\notin \Opened})$ define the NP relation $\cR_\lambda$ as
        \[
            \cR_\lambda:= \left\{(x,w): \begin{array}{r}
                \pred(C) = 1\ \wedge\\
                \forall t\notin \Opened,\ c_t = \com(\overline{k_t}; r_t) \ \wedge \\
                \forall t\notin \Opened,\ \ObfC_t = \JLLWObf(1^\lambda, C, \overline{k}_t, \overline{h}_t; \Obfr_t)
            \end{array}\right\}
        \]
        \item Compute $\pi \leftarrow \NIZK.\mathsf{P}(1^{\secp}, \crs, x, w)$ for $x:=(\pred, \chal, \{c_t,\overline{h}_t,\ObfC_t\}_{t\notin \Opened})$ and $w:=(C, \{r_t, \overline{k}_t, \Obfr_t\}_{t\notin \Opened})$
         \item Return $\ObfC := (\{c_t, \overline{h}_t\}_{t\in[\lambda]}, \chal, \{\ObfC_t\}_{t\notin\Opened},\{\overline{k}_t, r_t\}_{t\in\Opened}, \pi)$
    \end{itemize}
   
    \item $\Ver^\QPRO(\pp, \pred, \ObfC)$:
    \begin{itemize}
        \item  Parse $\ObfC$ as $(\{c_t, \overline{h}_t\}_{t\in[\lambda]}, \chal, \{\ObfC_t\}_{t\notin\Opened},\{\overline{k}_t, r_t\}_{t\in\Opened}, \pi)$. Reject if any term is missing.
        \item Parse $\pp$ as $(\crs, h^*)$ 
        \item Check that $\chal = \QPRO_0(\Eval, h^*, (c_1, \ldots, c_\lambda, \overline{h}_1, \ldots, \overline{h}_\lambda))$. Reject otherwise.
        \item Check that for all $t\in\Opened$, $c_t = \com(\overline{k}_t; r_t)$. Reject otherwise.
        \item For each $t \in \Opened$ parse $\overline{h}_t$ as $\{h^t_{i,j}\}_{i,j}$ and $\overline{k}_t$ as $\{k^t_{i,j}\}_{i,j}$, where $0 \leq i < D$ and  $1 \leq j \leq B$.
        \item For each $t\in\Opened$, $0 \leq i < D$ and  $1 \leq j \leq B$:
        \begin{itemize}
            \item Check that $h^t_{i,j} = \QPRO_t(\Gen,k^t_{i,j})$. Reject otherwise.
        \end{itemize}
        \item Define $x:=(\pred, \chal, \{c_t,\overline{h}_t,\ObfC_t\}_{t\notin \Opened})$
        \item If $\NIZK.\mathsf{V}(1^\secp, \crs, x, \pi) = 1$ then accept, else reject. 
    \end{itemize}
    \item $\Eval^\QPRO(\ObfC, z):$
    \begin{itemize}
        \item Parse $\ObfC$ as $(\{c_t, \overline{h}_t\}_{t\in[\lambda]}, \chal, \{\ObfC_t\}_{t\notin\Opened},\{\overline{k}_t, r_t\}_{t\in\Opened})$
        \item For each $t \notin \Opened$:
        \begin{itemize}
            \item $y_t := \JLLWObf.\Eval^{\QPRO_t}(\ObfC_t, \overline{h}_t, z)$ 
        \end{itemize}
        \item Return the most frequent element in $\{y_t\}_{t\notin\Opened}$ breaking ties arbitrarily.
    \end{itemize}
\end{itemize}

\end{construction}

\begin{theorem}
    The obfuscator given in \cref{construction:provable-obfuscator-QPRO} satisfies \cref{def:provable-obfuscation} for all samplers $\cS$.
\end{theorem}
\begin{proof} We prove that the construction satisfies the following properties.

    \textbf{Knowledge Soundness.} We define $\Ext_0$ and $\Ext_1$ as follows.
    \begin{itemize}
        \item $\Ext_0(1^\lambda)$: Sample $\crs, \td \leftarrow \NIZK.\Ext_0(1^\lambda)$, $h^* \leftarrow \{0,1\}^\lambda$, and return $((\crs, h^*), \td)$.
        \item $\Ext_1(\pp, \td,\pred, \ObfC, \pi)$ does the following.
        \begin{itemize}
            \item  Parse $\ObfC$ as $(\{c_t, \overline{h}_t\}_{t\in[\lambda]}, \chal, \{\ObfC_t\}_{t\notin\Opened},\{\overline{k}_t, r_t\}_{t\in\Opened})$. Output $\bot$ if any term is missing.
            \item Parse $\pp$ as $(\crs, h^*)$ 
            \item  Define $x:=(\pred, \chal, \{c_t,\overline{h}_t,\ObfC_t\}_{t\notin \Opened})$
            \item Run $\NIZK.\Ext_1(\crs,\td,x,\pi)$ to obtain $w = (C, \{r_t, \overline{k}_t, \Obfr_t\}_{t\notin \Opened})$
            \item Return $C$
        \end{itemize}
    \end{itemize}
    To prove that extraction succeeds, we will rely on the security of cut-and-choose.
    \begin{claim}
        \label{clm:cut-and-choose}
        Define the following set, where $i \in [0,D-1]$ and $j\in [B]$:

        \[
        \mathsf{Good}:= \left\{(c, \overline{h}, \overline{k}, r, t) \text{ s.t. for $i \in [0,D-1]$ and $j\in [B]$ }: \begin{array}{r}
             h = \{h_{i,j}\}_{i,j} \text{ where } h_{i,j} \in \{0,1\}^\lambda\\
             k = \{k_{i,j}\}_{i,j} \text{ where } k_{i,j} = \QPRO_t(\Gen, h_{i,j})\\
             c = \com(\overline{k}; r)                
        \end{array}\right\}
        \]
        For all QPT adversaries $\cA$, for large enough $\lambda$
        \[
        \Pr\left[\begin{array}{l}
            \left( |\mathbb{B}|\; \geq (\lambda - |\Opened|)/2\right) ~ \wedge\\
            \forall t \in \Opened, \\ (c_t, \overline{h}_t, \overline{k}_t, r_t, t) \in \mathsf{Good} 
        \end{array}~\middle |\begin{array}{r}
             h^* \leftarrow \{0,1\}^\lambda\\
             \{c_t, \overline{h}_t, \overline{k}_t, r_t\}_{t\in[\lambda]} \leftarrow \cA^\QPRO(h^*)\\
             \chal := \QPRO_0(\Eval, h^*, (c_1, \ldots, c_\lambda, \overline{h}_1, \ldots, \overline{h}_\lambda))\\
             \mathbb{B}:= \{t : (c_t, \overline{h}_t, \overline{k}_t, r_t, t) \notin \mathsf{Good} \wedge c_t = \com(\overline{k}_t;r_t)\}
        \end{array}~\right] \leq \negl(\lambda)
        \]
        
    \end{claim}
    \begin{proof}
        We prove by reduction to the security of post-quantum Fiat Shamir for statistically sound protocols. Let $\Pi$ be the three round protocol that consists of $\cA$ sending $\{c_t, \overline{h}_t\}_{t\in[\lambda]}$, receiving a random string $\chal$, and sending $\{\overline{k}_t, r_t\}_{t\in[\lambda]}$. The adversary succeeds in breaking soundness if for $\mathbb{B}:= \{t : (c_t, \overline{h}_t, \overline{k}_t, r_t, t) \notin \mathsf{Good} \wedge c_t = \com(\overline{k}_t;r_t)\}$, $|\mathbb{B}|\; \geq (\lambda - |\Opened|)/2$ and $\forall t \in \Opened,(c_t, \overline{h}_t, \overline{k}_t, r_t, t) \in \mathsf{Good} $. For any first message by the adversary, let the set $\mathbb{G} := \{t : \exists \overline{k}, r \text{ s.t. } (c_t, \overline{h}_t, \overline{k}, r, t) \in \mathsf{Good}\}$. By the perfect binding of the commitment scheme, $\mathbb{G} \cap \mathbb{B} = \emptyset$. Therefore the adversary can only break soundness if $|\mathbb{G}| \leq \lambda - (\lambda - |\Opened|)/2 = (\lambda + |\Opened|)/2$ since otherwise $|\mathbb{B}|$ will be too small. Additionally if $\Opened \not\subseteq \mathbb{G}$ then $\cA$ cannot win, since if $t \notin \mathbb{G}$ then it must be the case that $(c_t, \overline{h}_t, \overline{k}_t, r_t, t) \notin \mathsf{Good}$. 
        Therefore, whenever $\cA$ succeeds at breaking soundness it must be the case that
        \begin{itemize}
            \item $\Opened \not\subseteq \mathbb{G}$ and
            \item $|\Opened| \geq 2|\mathbb{G}| -\lambda$.
        \end{itemize} 
        Note also that the first message fixes $\mathbb{G}$ and $\chal$ is chosen randomly independently of $\mathbb{G}$. Therefore 
        \[\Pr[\Opened \not\subseteq \mathbb{G}] \leq 1/2^{\lambda - |\mathbb{G}|}\]
        and by Hoeffding inequality, if $|\mathbb{G}| \geq 3\lambda/4$,
        \[\Pr[|\Opened| \geq 2|\mathbb{G}| -\lambda] \leq \exp(-(4|\mathbb{G}| - 3\lambda)^2/\lambda)\]
        If $|\mathbb{G}| \leq 3\lambda/4 + \lambda/4$ then $1/2^{\lambda - |\mathbb{G}|}\leq 1/2^{\lambda/8} \leq \negl(\lambda)$ while if $|\mathbb{G}| \geq 3\lambda/4 + \lambda/4$ then $\exp(-(4|\mathbb{G}| - 3\lambda)^2/\lambda)\leq \exp(-\lambda/4) \leq \negl(\lambda)$, so the adversary can break soundness with at most negligible probability.

        Let $\QPRO' := (\QPRO_1, \QPRO_2, \ldots, \QPRO_\lambda)$, i.e. all but the first instantiation of $\QPRO$. Applying the Fiat-Shamir transform \cite{AC:Unruh17} to the protocol above yields for all adversaries $\cA$ that make at most $\poly(q)$ queries to $\cO$, for large enough $\lambda$
        \[
        \Pr\left[\begin{array}{l}
            \left( |\mathbb{B}|\; \geq \lambda/4\right) ~ \wedge\\
            \forall t \in \Opened, \\ (c_t, \overline{h}_t, \overline{k}_t, r_t, t) \in \mathsf{Good} 
        \end{array}~\middle |\begin{array}{r}
             \{c_t, \overline{h}_t, \overline{k}_t, r_t\}_{t\in[\lambda]} \leftarrow \cA^{\cO, \QPRO'}\\
             \chal := \cO(c_1, \ldots, c_\lambda, \overline{h}_1, \ldots, \overline{h}_\lambda))\\
             \mathbb{B}:= \{t : (c_t, \overline{h}_t, \overline{k}_t, r_t, t) \notin \mathsf{Good} \wedge c_t = \com(\overline{k}_t;r_t)\}
        \end{array}~\right] \leq \negl(\lambda)
        \]
        where $\cO$ is a random oracle. 
        
        Now, by Lemma \ref{lem:qpro-keyswap} we know that for all QPT adversaries $\cA$,
    \[\Big|\Pr[\cA^{\QPRO_0}(h^*) = 1 : h \gets \{0,1\}^\lambda] - \Pr[\cA^{\QPRO_0[h^* \to k]}(h^*) = 1 : h^*,k \gets \{0,1\}^\lambda]\Big| = \negl(\lambda), \]
    where $\QPRO_0[h^* \to k] = \QPRO_0$ except that on input $(\Eval,h^*,x)$ it outputs $f_k(x)$ instead of $f_{\pi^{-1}(h^*)}(x)$. That is, it answers PRF queries on handle $h^*$ using an independently sampled PRF key. Additionally, by the post-quantum security of the PRF (and the fact that for any $h^*$ and $k$, the oracle ${\QPRO_0[h^* \to k]}$ can be simulated given $h^*$, oracle access to $f_k$, and a post-quantum secure PRP)
    \[\Big|\Pr[\cA^{{\QPRO_0}[h^* \to k]}(h^*) = 1 : h^*,k \gets \{0,1\}^\lambda] - \Pr[\cA^{{\QPRO_0}[h^*]}(h^*) = 1 : h \gets \{0,1\}^\lambda]\Big|  = \negl(\lambda), \] where
 $\QPRO_0[h^*] = \QPRO_0$ except that on input $(\Eval,h,x)$ it outputs $\cO(x)$ instead of $f_{\pi^{-1}(h)}(x)$ for a random oracle $\cO$. Let $\QPRO[h^*] = (\QPRO_0[h^*], \QPRO')$. Now, since for any $h^*$ and $k$, the oracle ${\QPRO_0[h^*]}$ can be simulated given $h^*$, oracle access to $\cO$, and a post-quantum secure PRP, we can replace access to $(\cO, \QPRO')$ with $\QPRO[h^*]$ for random $h^*$ in the post Fiat-Shamir protocol without any loss in soundness. That is, for all QPT adversaries $\cA$, for large enough $\lambda$
        \[
        \Pr\left[\begin{array}{l}
            \left( |\mathbb{B}|\; \geq \lambda/4\right) ~ \wedge\\
            \forall t \in \Opened, \\ (c_t, \overline{h}_t, \overline{k}_t, r_t) \in \mathsf{Good} 
        \end{array}~\middle |\begin{array}{r}
        h^* \leftarrow \{0,1\}^\lambda\\
             \{c_t, \overline{h}_t, \overline{k}_t, r_t\}_{t\in[\lambda]} \leftarrow \cA^{\QPRO[h^*]}\\
             \chal := \QPRO_0[h^*](\Eval, h^*, c_1, \ldots, c_\lambda, \overline{h}_1, \ldots, \overline{h}_\lambda))\\
             \mathbb{B}:= \{t : (c_t, \overline{h}_t, \overline{k}_t, r_t) \notin \mathsf{Good} \wedge c_t = \com(\overline{k}_t;r_t)\}
        \end{array}~\right] \leq \negl(\lambda)
        \]
        Since no efficient adversary can distinguish 
        $\QPRO_0[h^*]$ from $\QPRO_0$ and $\QPRO'$ is efficiently simulatable, no efficient adversary can distinguish 
        $\QPRO[h^*]$ from $\QPRO$. Using this fact as well as by noting that the condition on the LHS is efficiently checkable, we obtain that for all QPT adversaries $\cA$, for large enough $\lambda$
        \[
        \Pr\left[\begin{array}{l}
            \left( |\mathbb{B}|\; \geq \lambda/4\right) ~ \wedge\\
            \forall t \in \Opened, \\ (c_t, \overline{h}_t, \overline{k}_t, r_t) \in \mathsf{Good} 
        \end{array}~\middle |\begin{array}{r}
             h^* \leftarrow \{0,1\}^\lambda\\
             \{c_t, \overline{h}_t, \overline{k}_t, r_t\}_{t\in[\lambda]} \leftarrow \cA^\QPRO(h^*)\\
             \chal := \QPRO_0(\Eval, h^*, (c_1, \ldots, c_\lambda, \overline{h}_1, \ldots, \overline{h}_\lambda))\\
             \mathbb{B}:= \{t : (c_t, \overline{h}_t, \overline{k}_t, r_t) \notin \mathsf{Good} \wedge c_t = \com(\overline{k}_t;r_t)\}
        \end{array}~\right] \leq \negl(\lambda)
        \]
        which is the statement in the claim.
    \end{proof}
    Note that since
    \[
    \{\crs : \crs \leftarrow \NIZK.\Setup(1^\lambda)\} \approx \{\crs : (\crs,\td) \leftarrow \NIZK.\Ext_0(1^\lambda)\}
    \]
    and since $\Ext_0(1^\lambda)$ samples $h^*$ honestly,
    \[
    \{\pp : \pp \leftarrow \Setup(1^\lambda)\} \approx \{\pp : (\pp,\td) \leftarrow \Ext_0(1^\lambda)\}
    \]
    Suppose there exists a QPT adversary $\cA$ such that
    \[
        \Pr\left[
                \begin{array}{l}
                    \left(\Ver(\pp,\pred,\ObfC, \pi) = \top\right) \wedge \\
                    \left(\exists x \text{ s.t } C(x) \neq \Eval(\ObfC,x) \vee \pred(C) = 1\right)
                \end{array}
                : 
                \begin{array}{r}
                    (\pp,\td) \gets \Ext_0(1^\lambda) \\
                    (\ObfC, \pi) \gets \cA^\QPRO(\pp) \\ 
                    C \gets \Ext_1(\pp, \td,\pred,\ObfC, \pi)
                \end{array}
            \right] \geq \epsilon(\lambda)
    \]
    for some non-negligible function $\epsilon(\cdot)$, then since $\Ver$ includes verifying the underlying $\NIZK$, by the knowledge soundness of the underlying $\NIZK$ it must be the case that
    \[
        \Pr\left[
                \begin{array}{l}
                    \left(\Ver(\pp,\pred,\ObfC, \pi) = \top\right) \wedge \\
                    \left(\exists x \text{ s.t } C(x) \neq \Eval(\ObfC,x) \vee \pred(C) = 1\right) \wedge \\
                    (\cR(x,w) = 1)
                \end{array}
                : 
                \begin{array}{r}
                    (\pp,\td) \gets \Ext_0(1^\lambda) \\
                    (\ObfC, \pi) \gets \cA^\QPRO(\pp) \\ 
                    C \gets \Ext_1(\pp, \td,\pred,\ObfC, \pi)
                \end{array}
            \right] \geq \epsilon(\lambda) - \negl(\lambda)
    \]
    where $\cR$ is the relation defined by $\Prove$ while $x$ and $w$ are the instance and witness extracted by $\NIZK.\Ext_1$ run internally by $\Ext_1$. 

    Now if $\cR(x,w) = 1$ then it must be the case that
    \begin{itemize}
        \item $\pred(C)=1$ and
        \item for all $t \notin \Opened$
        \begin{itemize}
            \item $c_t = \com(\overline{k}_t; r_t)$ and
            \item $\ObfC_t = \JLLWObf(1^\lambda, C, \overline{k}_t, \overline{h}_t; \Obfr_t)$
        \end{itemize}
    \end{itemize}
    Additionally if $\Ver(\pp,\pred,\ObfC, \pi) = \top$ then it must be the case that $\chal = \QPRO(\Eval,h^*, (c_1, \ldots, c_\lambda, \overline{h}_1, \ldots, \overline{h}_\lambda))$ and for all $t \in \Opened$, $(c_t, \overline{h}_t, \overline{k}_t, r_t) \in \mathsf{Good}$. 
    This means that for all $t\notin \Opened$, $\Obf_t$ is an honestly computed obfuscation of $C$ using handles $\overline{h}_t$ and keys $\overline{k}_t$, and $c_t$ is an honestly computed commitments to $\overline{k}_t$. By the perfect correctness of $\JLLWObf$, if for all $k_{i,j} \in \overline{k}_t$ and $h_{i,j} \in \overline{h}_t$ it was the case that $\QPRO(\Gen, k_{i,j}) = h_{i,j}$ then for all $x$, $\JLLWObf.\Eval(\ObfC_t,x) = C(x)$. However, since $\Eval(\ObfC, x)$ computes the most frequent element of $\{\JLLWObf.\Eval(\ObfC_t,x)\}_{t\notin\Opened}$, then if there exists $x$ such that $\Eval(\ObfC, x) \neq C(x)$, it must mean that for the majority of $t\notin\Opened$, the keys in $\overline{k}_t$ do not all map to their corresponding handle in $\overline{h}_t$ under $\QPRO(\Gen,\cdot)$. This means that the set $\mathbb{B}:= \{t : (c_t, \overline{h}_t, \overline{k}_t, r_t) \notin \mathsf{Good} \wedge c_t = \com(\overline{k}_t;r_t)\}$ is of size atleast $\lambda - |\Opened|$. 
    
    Let $\widetilde{\cA}$ be the algorithm that receives $h^*$, samples $\crs, \td \leftarrow \NIZK.\Ext_0(1^\lambda)$, runs $\cA((\crs,h^*))$ to obtain $(\ObfC, \pi)$, and runs $\Ext_1(\pp,\td,\pred,\ObfC,\pi)$ to obtain the set $\{c_t, \overline{h}_t, \overline{k}_t, r_t\}_{t\in[\lambda]}$, which $\widetilde{\cA}$ then outputs. Then by the above 
    \[
        \Pr\left[\begin{array}{l}
            \left( |\mathbb{B}|\; \geq (\lambda - |\Opened|)/2\right) ~ \wedge\\
            \forall t \in \Opened, \\ (c_t, \overline{h}_t, \overline{k}_t, r_t) \in \mathsf{Good} 
        \end{array}~\middle |\begin{array}{r}
             h^* \leftarrow \{0,1\}^\lambda\\
             \{c_t, \overline{h}_t, \overline{k}_t, r_t\}_{t\in[\lambda]} \leftarrow \widetilde{\cA}^\QPRO(h^*)\\
             \chal := \QPRO_0(\Eval, h^*, (c_1, \ldots, c_\lambda, \overline{h}_1, \ldots, \overline{h}_\lambda))\\
             \mathbb{B}:= \{t : (c_t, \overline{h}_t, \overline{k}_t, r_t) \notin \mathsf{Good} \wedge c_t = \com(\overline{k}_t;r_t)\}
        \end{array}~\right] \geq \epsilon(\lambda)-\negl(\lambda)
        \]
        which contradicts Claim \ref{clm:cut-and-choose}.

        \noindent\paragraph{Simulation Security.}
        We define $\SimGen$ and $\SimObf$ as follows.
        \item $\SimGen^\QPRO(1^\lambda)$:
\begin{itemize}
    \item $(\crs,\td) \gets \NIZK.\SimGen(1^\lambda)$
    \item $h^* \gets \{0,1\}^\lambda$
    \item Return $((\crs,h^*),\td)$
\end{itemize}

\item $\SimObf^\QPRO(1^{\secp_\OBF},\pp,\td,\varphi,C)$:
\begin{itemize}
    \item Run $\Obf^\QPRO(1^{\secp_\OBF},\pp,\varphi,C)$ honestly, except generate $\pi$ by running $\NIZK.\Sim(\crs,\td,x)$.
\end{itemize}

Honest-vs-simulated indistinguishability follows from the zero-knowledge property of $\NIZK$. We will next show that $\SimObf$ satisfies subexponential ideal obfuscation security.

We show security by building a simulator which ``takes over'' all the oracles $(\QPRO_0,\QPRO_1,\ldots,\QPRO_\lambda)$. Let $S=(S_1,S_2,S_3)$, where $S_1$ will compute $(\pp,\td)$, $S_2$ will compute the obfuscated circuit $\widetilde{C}$, and $S_3$ will simulate the oracles after the obfuscated circuit is sent to the adversary. Formally, we show for every $C$
\begin{align*}
&\bigg|\Pr\!\left[
  A^{\QPRO}(\pp,\widetilde{C}) =1: 
  \begin{array}{l}
    (\pp,\td) \gets \SimGen^\QPRO(1^\lambda) \\
    \widetilde{C} \gets \SimObf^\QPRO(1^{\secp_\OBF},\pp,\td,\varphi,C)
  \end{array}
\right]
\\ &-
\Pr\!\left[
  A^{S_3^C}(\pp,\widetilde{C}) =1:
  \begin{array}{l}
    (\pp,\td) \gets S_1(1^\lambda) \\
    \widetilde{C} \gets S_2^C(1^{\secp_\OBF},\pp,\td,\varphi)
  \end{array}
\right]\bigg| \leq 1/2^{\lambda_\OBF^\epsilon}
\end{align*}
for some $\epsilon> 0$. For any function $f$ and handle $h$, left $\QPRO[h^* \mapsto f]$ refer to an oracle identical to $\QPRO$ except $\Eval$ queries with handle $h$ are answered using $f$ instead.
\noindent Define $S=(S_1,S_2,S_3)$, where all simulators share state, as follows.

\item $S_1(1^\lambda)$:
\begin{itemize}
    \item $(\crs,\td) \gets \NIZK.\SimGen(1^\lambda)$
    \item $h^* \gets \{0,1\}^\lambda$
    \item Return $((\crs,h^*),\td)$
\end{itemize}

\item $S_2^C(1^{\secp_\OBF},\pp,\td,\varphi)$:
\begin{itemize}
    \item Let $(S_{1,i},S_{2,i},S_{3,i})_{i\in[\lambda]}$ be $\lambda$ separate instantiations of the ideal obfuscation simulators for $\JLLWObf$, where $(S_{1,i},S_{2,i},S_{3,i})$ are allowed to control $\QPRO_i$.
    \item Compute $\widetilde{C}$ as in $\SimObf$ except:
    \begin{itemize}
        \item Sample $\chal \gets \{0,1\}^\lambda$
        \item $\forall t \in [\lambda]$, set $\widetilde{C}_t \gets S_{2,t}^C(1^{\secp_\OBF},D,S)$.
        \item If $t \in \Opened$, $\widetilde{C}_t$ is not sent to the adversary.
        \item For all $t \in \Opened$:
        \begin{itemize}
            \item $c_t \gets \com(0;r_t)$
            \item $\overline{h}_t$ is extracted from $\widetilde{C}_t$
        \end{itemize}
    \end{itemize}
\end{itemize}

\item $S_3^C$:
\begin{itemize}
    \item Let $\cO$ be an efficient simulation of a random oracle. Note that such simulation is possible via the compressed oracle technique \cite{C:Zhandry19}.
    \item Let $\cO'$ be identical to $\cO$ except $(c_1,\ldots,c_\lambda,\overline{h}_1,\ldots,\overline{h}_\lambda)$ maps to $\chal$.
    \item Let $\widetilde{\QPRO}_0$ be a simulated instance of $\QPRO_0$ (using a PRP).
    \item Give query access to $(\widetilde{\QPRO}_0[h^* \mapsto \cO'],S_{3,1}^C,\ldots,S_{3,\lambda}^C)$. 
\end{itemize}

\noindent We prove indistinguishability via a sequence of hybrids.

\begin{itemize}
    \item Hybrid$_0$: Identical to using $(\SimGen,\SimObf)$.
    \item Hybrid$_1$: Identical to Hybrid$_0$, except $A$ and $\SimObf$ are given query access to $(\QPRO_0[h^* \mapsto \cO],\QPRO')$ instead of $\QPRO$, where $\QPRO':= (\QPRO_1, \ldots, \QPRO_\lambda)$ . By Lemma \ref{lem:qpro-keyswap} and the post-quantum security of the PRF, as well as the observation that $\QPRO_0[h^*\mapsto f]$ can be simulated efficiently using query access to $f$, both hybrids are indistinguishable.
    \item Hybrid$_2$: Identical to Hybrid$_1$, except $\chal \gets \{0,1\}^\lambda$ and $A$ and $\SimObf$ are given query access to $(\QPRO_0[h^*\mapsto \cO'],\QPRO')$, where $\cO'$ is identical to $\cO$ except $(c_1,\ldots,c_\lambda,\overline{h}_1,\ldots,\overline{h}_\lambda)$ maps to $\chal$. Both hybrids are identical in the adversary’s view.
    \item Hybrid$_3$: Identical to Hybrid$_2$, except for all $t \in \Opened$, $c_t$ is computed as a commitment to $0$. Both hybrids are indistinguishable by the hiding property of the commitment scheme.
    \item Hybrid$_{4,i}$: For all $i \in [0,\lambda]$, Hybrid$_{4,i}$ is identical to Hybrid$_3$ except for all $j \leq i$:
    \begin{itemize}
        \item $\widetilde{C}_j \gets S_{2,j}^C(1^{\secp_\OBF},1^D,1^S)$
        \item If $j \notin \Opened$ then $\overline{h}_j$ is extracted from $\widetilde{C}_j$
    \end{itemize}
    and $A$ is given access to $({\QPRO}_0[h^*\mapsto \cO'],S_{3,1}^C,\ldots,S_{3,i}^C,\QPRO_{i+1},\ldots,\QPRO_\lambda)$. 
    Hybrid$_{4,0}$ is identical to Hybrid$_3$, and for all $i$, Hybrid$_{4,i-1}$ and Hybrid$_{4,i}$ are indistinguishable by the relativizing ideal obfuscation security of JLLW.
    \item Hybrid$_4$: Identical to Hybrid$_{4,\lambda}$ except $\widetilde{\QPRO}_0[h^*\mapsto \cO']$ is used instead of $\QPRO_0[h^*\mapsto \cO']$. Hybrids$_{u,\lambda}$ and Hybrid$_4$ are indistinguishable by PRP security.
\end{itemize}

The view of the adversary in Hybrid$_4$ is identical to the view when running with $S$. Since all primitives are subexponentially secure, the adversary has at most subexponential distinguishing advantage in $\secp_\OBF$.

Finally, by \Cref{lem:idealobfs-is-evasive}, evasive composability and indistinguishability obfuscation are implied by ideal obfuscation.
\end{proof}  
We obtain the following as a direct corollary of  the above theorem and \cref{thm:JLLW}.
\begin{theorem}
 \label{thm:provable-obfuscation} 
 Assuming functional encryption with subexponential security (\cref{def:FE}) and post-quantum NIZK arguments of knowledge for NP with a URS setup, there exists a provable obfuscation scheme (\cref{def:provable-obfuscation}) with a URS setup in the QPrO model (\cref{def:QPRO}).
\end{theorem}

\section{Security of the JLLW Obfuscator in the QPrO Model}
\label{sec:PROM}

This section is dedicated to proving that the JLLW construction of ideal obfuscation in the pseudorandom oracle model is post-quantum secure. 

\begin{theorem}\label{thm:JLLW}
    Let $\cR$ be an oracle. Assuming functional encryption and pseudorandom functions with subexponential security relative to $\cR$, the JLLW obfuscation given in $\cref{construction:JLLW}$ satisfies (subexponential) post-quantum ideal obfuscation (\cref{def:ideal-obf}) relative to $\cR$ in the quantum-accessible pseudorandom oracle model (\cref{def:QPRO}).
\end{theorem}

There are two main differences between the post-quantum setting and the classical setting that arise here. The first, and biggest, difference is the difficulty of adaptively programming random oracle queries. This nuance prevents the JLLW trick of switching behavior on an exponential number of inputs by adaptively reprogramming the oracle at only a polynomial number of queries. As such, our post-quantum result relies on the subexponential security of the underlying primitives, in contrast to JLLW's reliance on only polynomial hardness.

A second difference is the treatment of random permutations in the quantum setting. It is currently an open problem to perform efficient (statistical) simulation of a random permutation.\footnote{The analogous question for random functions is solved by the influential compressed oracle technique~\cite{C:Zhandry19}.}
In theory, this means that oracle access to a random permutation could break computational security assumptions. However, this cannot occur if post-quantum \emph{psuedo-}random permutations exist. A PRP would allow efficient simulation of the random permutation against any computationally efficient adversary.
Fortunately, both (post-quantum) psuedorandom functions and functional encryption imply the existence of (post-quantum) PRPs ~\cite{DBLP:journals/quantum/Zhandry25}, so we do not need any additional assumptions.

Building on the observation that a QPrO can be efficiently simulated using post-quantum PRFs and PRPs, one can see that the parallel composition of JLLW is an ideal obfuscator when multiple QPRO oracles are available.\footnote{Security can also be shown with a single QPrO by observing that each simulator really only needs to program a handful of key-handle pairs and the other handles can be answered relatively to the original QPrO. Proving this would further complicate an already technically involved proof for a small improvement in the model, so we omit the details.}

\begin{corollary}
    Assuming functional encryption and pseudorandom functions with subexponential security relative to $\cR$, the parallel composition 
    \[
        \Obf(C_1, \dots, C_n) = (\Obf(C_1), \dots, \Obf(C_n))
    \]
    of the JLLW obfuscation satisfies post-quantum ideal obfuscation (\cref{def:ideal-obf}) in the $n$-time quantum-accessible pseudorandom oracle model (\cref{def:QPRO}).
\end{corollary}

We also obtain the following as an immediate corollary, which allows us conclude that several objects from the literature that were previously only known in the classical oracle model can in fact be constructed in the QPrO model, assuming an appropriate flavor of functional encryption 
(e.g.\ witness encryption for QMA \cite{BM22}, copy-protection for all unlearable functionalities \cite{C:ALLZZ21}, obfuscation for various classes of quantum circuits \cite{BKNY23,BBV24,HT25}, and quantum fire \cite{DBLP:journals/iacr/CakanGS25}).

\begin{corollary}
    Assuming functional encryption as defined in \cref{def:FE}, there exists post-quantum ideal obfuscation in the quantum-accessible pseudorandom oracle model.
\end{corollary}

\subsection{QPrO Key Reprogramming}

Arguing that the $\QPRO$ behaves as a random function on a random handle $h$ requires first removing the key $k$ that is associated with the handle $h$ by the random permutation. To help with this step, we show the following lemma.

\begin{lemma}\label{lem:qpro-keyswap}
    Let $F = \{f_k\}_k$ be a pseudorandom function and $\QPRO$ a pseudorandom oracle for $F$. Then for any QPT adversary $\cA$,
    \[\Big|\Pr[\cA^\QPRO(h) = 1 : h \gets \{0,1\}^\lambda] - \Pr[\cA^{\QPRO[h \to k]}(h) = 1 : h,k \gets \{0,1\}^\lambda]\Big| = \negl(\lambda), \]
    where $\QPRO[h \to k] = \QPRO$ except that on input $(\Eval,h,x)$ it outputs $f_k(x)$ instead of $f_{\pi^{-1}(h)}(x)$. That is, it answers PRF queries on handle $h$ using an independently sampled PRF key.

    If $F$ is subexponentially secure, then the probability of distinguishing is subexponential.
\end{lemma}

\begin{proof}
    We consider a sequence of three hybrids, defined as follows. In each, begin by sampling a uniformly random permutation $\pi : \{0,1\}^\secp \to \{0,1\}^\secp$, and $k,k' \gets \{0,1\}^\secp$. Define $h = \pi(k)$ and $h' = \pi(k')$. For all $k'' \notin \{k,k'\}$, define $\QPRO(\Gen,k'') = \pi(k'')$. For all $h'' \notin \{h,h'\}$, define $\QPRO(\Eval,h'',\cdot) = f_{\pi^{-1}(h'')}(\cdot)$. Define $\QPRO(\Eval,h,\cdot) = f_k(\cdot)$. The adversary is initialized with $h$, given access to the $\QPRO$ oracle and outputs a bit, where the behavior of $\QPRO$ on the remaining inputs differs in each hybrid, as follows.
    \begin{itemize}
        \item $\cH_0$: $\QPRO(\Gen,k) = h$, $\QPRO(\Gen,k') = h'$, $\QPRO(\Eval,h',\cdot) = f_{k'}(\cdot)$.
        \item $\cH_1$: $\QPRO(\Gen,k) = h$, $\QPRO(\Gen,k') = h$, $\QPRO(\Eval,h',\cdot) = f_k(\cdot)$.
        \item $\cH_2$: $\QPRO(\Gen,k) = h'$, $\QPRO(\Gen,k') = h$, $\QPRO(\Eval,h',\cdot) = f_k(\cdot)$.
    \end{itemize}
    Observe that $\cH_0$ is distributed exactly as the LHS of the lemma statement, while $\cH_2$ is distributed exactly as the RHS of the lemma statement, and thus it suffices to show indistinguishability between the hybrids.
    \begin{itemize}
        \item $\cH_0 \approx \cH_1$: By \cref{lemma:oracle-ind}, it suffices to show that the adversary can output a string in $\{k',h'\}$ with only negligible probability in $\cH_0$. This is straightforward to see due to the fact that $k'$ is sampled independently of the rest of the experiment (including the string $h$ that the adversary is initialized with) and $h' = \pi(k')$.
        \item $\cH_1 \approx \cH_2$: By \cref{lemma:oracle-ind}, it suffices to show that the adversary can output $k$ with only negligible probability in $\cH_2$. Note that $\cH_2$ can be simulated by a reduction given just oracle access to $f_k$ using a post-quantum psuedorandom permutation (which is implied by post-quantum PRFs). Thus, this follows from a standard claim that, given oracle access to a PRF $f_k$, an adversary can output $k$ with only negligible probability. 
    \end{itemize}

    Assuming subexponential security of the underlying PRF and PRP, and observing that the key/handle spaces are exponentially sized, it is clear that the hybrids are subexponentially close.
\end{proof}

\subsection{Proof of Security}

\begin{proof}[Proof of \cref{thm:JLLW}]
    Correctness was shown in \cite{JLLW23}. We show security against quantum adversaries here. We remark that the security of the construction can be tuned to any subexponential function by analogously tuning the subexponential parameters used in the subclaims of the proof.

    We will consider a main series of hybrid experiments $\Hyb_{i,\rand}$ which are preceded by two hybrids $\Hyb_{}$ and $\Hyb_{\mathsf{PRP}}$. We show that
    \[
        \Hyb_{\mathsf{real}} 
        \approx \Hyb_{\mathsf{KH}}
        \approx \Hyb_{\mathsf{PRP}}
        \approx \Hyb^{x}_{0,\rand} 
        \approx  \dots 
        \approx \Hyb^{x}_{D,\rand} 
        \approx \Sim.
    \]
    
    The two initial hybrids make the following changes from the prior hybrid.
    \begin{itemize}
        \item $\Hyb_{KH}$: Instead of generating the key-handle pairs $(k_{d,b}, h_{d,b})$ used for the obfuscation by sampling $k_{d,b}\gets \{0,1\}^\lambda$ and setting $h_{d,b} = \pi(k_{d,b})$, independently sample $k_{d,b}, h_{d,b} \gets \{0,1\}^\lambda$.

        \item $\Hyb_{\mathsf{PRP}}$: The random permutation $\pi$ is replaced by a \emph{pseudorandom} permutation $\pi_{k_\PRP}$. Note that $\pi_{k_\PRP}$ is not used for relating $(k_{d,b}, h_{d,b})$.
    \end{itemize}

    The intermediate hybrids $\Hyb_{\delta,\rand}$ are specified in \cref{fig:PROM-hybrid-true-random-state} and \cref{fig:PROM-hybrid-true-random}. At a high level, each $\Hyb_{\delta,\rand}$ makes the following differences from $\Hyb_{\mathsf{PRP}}$:
    \begin{itemize}
        \item \textbf{Depth $d < \delta$:} The ciphertext $\ct_{\chi}$ for input prefix $\chi \in \{0,1\}^d$ is in \textit{simulation} mode and uses truly random $r_\chi$ in its encryption. Furthermore, every handle $h_{d,b}$ at this depth does not have a corresponding PRF key (i.e. $k_{d,b}$ is non-existent/unused).
        \item \textbf{Depth $d \geq \delta$:} The ciphertext $\ct_{\chi}$ is in \textit{normal} mode and uses pseudorandom $r_{\chi}$ expanded from $s_{\chi\leq \delta}$. The only exception is $d = \delta$, where $r_{\chi}$ is truly random. Additionally, every handle $h_{d,b}$ at this depth corresponds to some PRF key $k_{d,b}$.
    \end{itemize}
    $\Sim$ is identical to $\Hyb_{D,\rand}$ except that it changes the final ciphertext $\ct_{x}$ corresponding to the full input $x\in\{0,1\}^D$ to an encryption of
    \[
        (\flag_x, x, \info_x) \coloneqq 
        \begin{cases}
            (\mathsf{normal}, x, (C, s_x))  &\text{in } \Hyb_{D,\rand}
            \\
            (\mathsf{sim}, x, C(x)) &\text{in } \Sim
        \end{cases}
    \]
    Each hybrid is expressed in terms of three oracles: $\Sim_1$ corresponds to the $\QPRO$ oracle before obfuscation, $\Sim_2$ corresponds to the obfuscation step, and $\Sim_3$ corresponds to the $\QPRO$ oracle after obfuscation.

    \begin{figure}[!hpt]
        \centering
        \begin{mdframed}
            \noindent\underline{Shared State:}
            \\
            \begin{tabular}{c|p{3in}}
                    $\pi_{k_\PRP}$ 
                    & The pseudorandom permutation used to map keys onto handles.
                \\\hline
                    $\cC$ 
                    & The circuit being obfuscated. This is used in $\Sim_2$ and $\Sim_3$.
                \\\hline
                    $\Handles \coloneqq \{h_{d,b} \gets \{0,1\}^\lambda\}_{\substack{0\leq d \leq D,\\1\leq b \leq B}}$ 
                    & The handles of the $\QPRO$ in the obfuscation. $h_{d,b}$ corresponds to level $d$ and block $b$. Generated uniformly at random.
                \\\hline
                    \begin{tabular}{@{}c@{}}
                    $\Keys \coloneqq \{k_{d,b}\gets \{0,1\}^\lambda\}_{\substack{\diff{\delta} \leq i \leq D\\ 1\leq b\leq B}}$  
                    \end{tabular}
                    & Keys of $H$ used for the obfuscation. Generated uniformly at random.
                \\\hline
                         $\{(\pk_d, \sk_d)\}_{0\leq d\leq D}$
                    & Public and secret keys for each level $d$. Generated as in \cref{construction:JLLW}.
                \\\hline
                    \begin{tabular}{@{}c@{}}
                        $F_{b,d}: \{0,1\}^D\rightarrow\{0,1\}^L$ 
                        \\ For $0\leq d \leq \difference{\delta}$ and $1\leq b \leq B$ 
                        \end{tabular}
                    & Random functions for the non-programmed portion of evaluation queries.
                \\\hline
                    $F_\sigma: \{0,1\}^{<\difference{\delta}} \times \{1,\dots,B\} \rightarrow \{0,1\}^\lambda$
                    & A random function used to generate one-time pads for the ciphertexts at depth $d<D$ and block $b \in \{1,\dots, B\}$.
                \\\hline
                    $F_r: \{0,1\}^{\leq \difference{\delta}} \rightarrow \{0,1\}^\lambda$
                    & A random function used to generate randomness $r_\chi$ for ciphertext $\ct_\chi$ for input prefix $\chi\in \{0,1\}^{\leq D}$.
                \\\hline
                    $F_s: \{0,1\}^{\diff{\delta}} \rightarrow \{0,1\}^\lambda$
                    &
                    Used to define $r_\chi$ and $s_\chi$ for $|\chi| = \difference{\delta}$.
                \\\hline
                    $\ct_\chi = \Enc(\pk_{|\chi|}, \flag_{\chi}, \info_\chi; r_\chi)$
                    &
                    The ciphertext corresponding to input prefix $\chi\in\{0,1\}^{\leq D}$. This is implicitly defined by the other parameters and is used in $\Sim_2$ and $\Sim_3$.
                \\\hline
                    $\flag_\chi \coloneqq \begin{cases}
                        \diff{\mathsf{sim}} & \text{if } |\chi|< \diff{\delta}
                        \\
                        \mathsf{normal} & \text{if } |\chi| \geq \diff{\delta}
                    \end{cases}$ 
                    & Flag used to specify whether the ciphertext $\ct_\chi$ is in simulation mode or normal mode.
                \\\hline
                    \begin{tabular}{@{}l@{}}
                        For $|\chi| = \diff{\delta}$:
                        \\
                        \qquad \diff{$r_\chi \coloneqq F_r(\chi)$}
                        \\
                        \qquad \diff{$s_\chi \coloneqq F_s(\chi)$}
                        \\
                        For $|\chi| > \diff{\delta}$:
                        \\
                        \qquad $s_{\chi\concat 0}\concat r_{\chi\concat 0} \concat s_{\chi\concat 1}\concat r_{\chi\concat 1} \coloneqq G_{sr}(s_\chi)$
                    \end{tabular}
                    &
                    $r_\chi$ is the randomness used for the encryption of each $\chi$ at depth $\geq \delta$. $s_\chi$ is a PRG seed used to generate downstream $r$ and $s$.
                \\\hline
                    $\info \coloneqq \begin{cases}
                        \diff{\{F_\sigma(\chi, b)\}_{1\leq b\leq B}} & \text{if } |\chi| < \diff{\delta}
                        \\
                        \left\{C, \{k_{d,b}\}_{\substack{|\chi|\leq d\leq D,\\ 1\leq b\leq B}}, s_\chi\right\} & \text{if } |\chi| \geq \diff{\delta}
                    \end{cases}$
                    &
                    Information used to evaluate the ciphertext $\ct_\chi$ under the FE. 
                \\\hline
                    \begin{tabular}{@{}l@{}}
                        $G_{sr}:\{0,1\}^\lambda \rightarrow \{0,1\}^{4\lambda}$
                        $G_v:\{0,1\}^\secpar \rightarrow \{0,1\}^L$
                    \end{tabular}
                    & PRGs used to generate the encryption randomness for each $\ct_\chi$ and for succinctly hardcoding the programmed behavior into the FE ciphertexts, respectively.
            \end{tabular}
            \end{mdframed}
        \caption{$\Hyb_{\delta,\rand}$: Shared state of algorithms $\Sim^{(\delta,\rand)}_1, \Sim^{(\delta,\rand)}_2, \Sim_3^{(\delta,\rand)}$. Differences from $\Hyb_{\mathsf{real}}$ and dependence on $\delta$ highlighted.}
        \label{fig:PROM-hybrid-true-random-state}
    \end{figure}

    \begin{figure}[!hpt]
        \centering
        \begin{mdframed}
            \noindent\underline{Initialization.}
            Sample the following according to their distributions as specified in \cref{fig:PROM-hybrid-true-random-state}: 
            \[
                \pi,\ \Handles,\ \Keys,\ \{F_{b,d}\}_{\substack{0\leq d \leq D\\1\leq b \leq B}},\ F_\sigma,\ F_r
            \]
        
            \noindent\underline{$\Sim^{\rand,\delta}_1$:} 
            \\
            On input $(\Gen, k)$:
            \begin{enumerate}
                \item Output $\pi(k)$.
            \end{enumerate}
            
            On input $(\Eval, h, x)$:
            \begin{enumerate}
                \item \diff{If $h = h_{d,b}$ for some $h_{d,b}\in \Handles$}, output $f_{\diff{k_{d,b}}}(x)$.
                \item Otherwise, compute $k\gets \pi^{-1}(h)$ and output $f_k(x)$.
            \end{enumerate}

            \vspace{1em}
            \noindent\underline{$\Sim^{\rand,\delta}_2$:} 
            \begin{enumerate}
                \item Generate $\{\pk_d, \sk_d\}_{d = 0,\dots,D}$ as specified in \cref{construction:JLLW}.
                
                \item Output $\widehat{C}^\bullet\left[\ct_\epsilon,\{\sk_d\}_{0 \leq d \leq D},\{h_{d,b}\}_{\substack{0 \leq d < D,\\ 1 \leq b \leq B}}\right]$, as defined in \cref{construction:JLLW}.
            \end{enumerate}

            \vspace{1em}
            \noindent\underline{$\Sim_3$:}
            \\
            On input $(\Gen, k)$:
            \begin{enumerate}
                \item Output $\pi(k)$.
            \end{enumerate}
            
            On input $(\Eval, h, x)$:
            \begin{enumerate}
                \item \diff{If $h = h_{d,b}$ for some $h_{d,b} \in \Handles$}:
                \begin{enumerate}
                    \item \diff{If $d < \delta$}:
                    \begin{enumerate}
                        \item If $x = \chi\concat 0^{D-d}$ for $\chi \in \{0,1\}^d$:
                        \begin{enumerate}
                            \item Output \diff{$G_v(F_\sigma(\chi, b)) \oplus \left[\ct_{\chi\concat 0} \concat \ct_{\chi\concat 1}\right]_b$}.
                        \end{enumerate}
                        \item Otherwise output \diff{$F_{d,b}(x)$}.
                    \end{enumerate}
                    \item Otherwise output $f_{\diff{k_{d,b}}}(x)$.
                \end{enumerate}
                \item Otherwise compute $k\gets \pi^{-1}(h)$ and output $f_k(x)$.
            \end{enumerate}
        \end{mdframed}
        \caption{$\Hyb_{\diff{\delta},\rand}$: Specification of algorithms $\Sim^{(\delta,\rand)}_1, \Sim^{(\delta,\rand)}_2, \Sim_3^{(\delta,\rand)}$. Differences from $\Hyb_{\mathsf{real}}$ and dependence on $\delta$ highlighted.}
        \label{fig:PROM-hybrid-true-random}
    \end{figure}

    We begin by showing that the first two transitions and the last transition are indistinguishable, since these transitions are easier to see.

    \begin{claim}    
        \[
            \Hyb_{\mathsf{real}} \approx \Hyb_{\mathsf{KH}}
        \]
    \end{claim}
    \begin{proof}
        The only difference between these hybrids is when either $\Sim_1^{0,\rand}$ or $\Sim_3^{0,\rand}$, corresponding to the $\QPRO$ oracle, receive an $\Eval$ query on $(h_{d,b}, x)$ for some $h_{d,b}\in \Handles$. In this case, $\Hyb_{0,\rand}$ uses an independently random $k_{d,b}$ to answer the query instead of $\pi^{-1}(h)$.
        Since $d$ ranges over $0,\dots, D$ and $b$ ranges over $1, \dots, B$, the claim follows by $(D+1)\cdot B = \poly(\lambda)$ applications of \cref{lem:qpro-keyswap}.
    \end{proof}

    \begin{claim}
        Assuming post-quantum pseudorandom permutations relative to $\cR$,
        \[
            \Hyb_{\mathsf{KH}} \approx \Hyb_{\mathsf{PRP}}
        \]
    \end{claim}
    \begin{proof}
        This follows immediately from the post-quantum security of the pseudorandom permutation.
    \end{proof}

    \begin{claim}
        If $\FE$ is $2^{D}$-adaptively secure relative to $\cR$ then
        \[
            \Hyb_{D,\rand} \approx \Sim
        \]
    \end{claim}
    \begin{proof}
        The only difference between these hybrids is the plaintext encrypted under $\ct_{x}$ for \emph{each} full input $x\in \{0,1\}^D$. Specifically, it is an encryption of
        \[
            (\flag_x, x, \info_x) \coloneqq 
            \begin{cases}
                (\mathsf{normal}, x, (C, s_x))  &\text{in } \Hyb_{D,\rand}
                \\
                (\mathsf{sim}, x, C(x)) &\text{in } \Sim
            \end{cases}
        \]
        In both hybrids, these ciphertexts are encrypted using true randomness under the public key $\pk_D$. The secret key $\sk_D$ is for the function $\Eval$. Observe that the result of $\Eval$ is the same for both plaintexts:
        \begin{align*}
            \Eval(\mathsf{normal}, x, (C,s_x)) 
            &= \Eval_{\mathsf{normal}}(x, (C,x))
            \\
            &= C(x)
            \\
            &= \Eval_{\mathsf{sim}}(x, C(x)) = \Eval(\mathsf{sim}, x C(x))
        \end{align*}

        Since there are $2^{D}$ inputs to switch the ciphertexts for, the claim follows from $2^{D}$ hybrids each reducing to the $2^D$-query adaptive security of 1-key FE.
    \end{proof}

    Next, we move on to the main technical part of the proof -- showing that the intermediate hybrids for depths $\delta$ and $\delta+1$ are indistinguishable.

    \begin{claim}\label{claim:jllw-claim-depth-transition}
        If $f$, $G_v$, and FE are $2^{|\delta|}$-secure relative to $\cR$, then
        \[
            \Hyb_{\delta,\rand} \approx \Hyb_{\delta+1,\rand}
        \]
    \end{claim}
    \begin{proof}
        To transition between $\delta$ and $\delta+1$, we perform a hybrid argument over every input prefix $\chi\in \{0,1\}^{\delta}$ where we modify the corresponding intermediate ciphertext $\ct_{\chi}$ in a block-by-block manner. We will transition along the following sequence of hybrids in lexicographical order over $\chi \in \{0,1\}^\delta$.
        \begin{itemize}
            \item $\Hyb_{\delta, \chi,\rand}$ modifies how depth $\delta$ behaves on message prefixes $<\chi$. Specifically, $\chi' \in \{0,1\}^\delta$ at depth $\delta$ is treated as in $\Hyb_{\delta,\rand}$ if $\chi' < \chi$ and is treated as in $\Hyb_{\delta+1,\rand}$ if $\chi' \geq \chi$. This affects the following: the plaintext of $\ct_{\chi'}$, $\QPRO$ queries on $(\Eval, h_{\delta,\beta}, \chi'\concat 0^{D-|\chi'|})$, the encryption randomness $r_{\chi'}$, and the PRG seed $s_{\chi'}$.
            
            \item $\Hyb_{\delta, \chi,sr}$ is the same as $\Hyb_{\delta, \chi+1,\rand}$,\footnote{We emphasize the $\chi+1$ here.} except that the PRG seed $s_\chi$ and encryption randomness $r_\chi$ are generated as uniformly random.
        \end{itemize}

        The main step is to show that
        \begin{align}
            \Hyb_{\delta, \chi} 
            &\approx_{2^{-\ell}\negl(\lambda)}
            \Hyb_{\delta, \chi, sr} 
            \label{eq:jllw-depth-sr}
            \\
            &\approx_{2^{-\ell}\negl(\lambda)}
            \Hyb_{\delta, \chi+1}
            \label{eq:jllw-depth+1}
        \end{align}
        Invoking this $2^{|\delta|}$ times to cover each $\chi\in\{0,1\}^\delta$ gives the claim.

        \Cref{eq:jllw-depth-sr} follows immediately from the $2^{|\delta|}$-security of $G_{sr}$.
        To show \cref{eq:jllw-depth+1}, we iterate across the following hybrids over each block $\beta = 1$ to $\beta = B$.
        \begin{itemize}
            \item $\Hyb_{\delta, \chi, \beta, 1}$: Change $\ct_{\chi} = \Enc_{\pk_{\delta}}(\flag_\chi, \info_\chi; r_\chi)$ to hardcode the $\beta$'th block of $\ct_{\chi\concat 0}$ and $\ct_{\chi\concat 1}$, instead of computing them on the fly. 
            Specifically, modify the plaintext to
            \begin{align*}
                \flag_\chi &\coloneqq \mathsf{hyb}
                \\
                \info_\chi &\coloneqq \left(C,\  
                    \{k_{d,b}\}_{\substack{\delta<d<D\\ 1\leq j \leq B}},\  
                    s_\chi,\  
                    \diff{\beta},\  
                    \{F_{\sigma}(\chi, b)\}_{1\leq b \leq \diff{\beta}},\ 
                    \diff{w_{\chi,\beta}},\ 
                    \{k_{\delta, b}\}_{\diff{\beta} < j \leq B}\} \right)
            \end{align*}
            where
            \[
                w_{\chi,\beta} \coloneqq \left[\ct_{\chi\concat 0}\concat \ct_{\chi\concat 1}\right]_{\beta} \oplus f_{k_{\delta,\beta}}(\chi\concat 0^{D-\delta})
            \]
            Note that for $\beta > 1$, this recycles the space used for $w_{\chi,\beta-1} = G_v(F_\sigma(\chi, \beta-1))$ by directly storing the much-shorter seed $F_\sigma(\chi, \beta-1)$ in the other parts of the ciphertext.
            
            \item $\Hyb_{\delta, \chi, \beta, 2}$: Replace $f_{k_{\delta,\beta}}(\chi\concat 0^{D-\delta})$ by $F_{\delta, \beta}(\chi \concat 0^{D-\delta})$. Note that this modifies both the hardcoded $w_{\chi,\beta}$ and the reply to $\Sim_3^{\delta,\rand}(\Eval, h_{\delta,\beta}, \chi\concat 0^{D-\delta})$.
            
`            \item $\Hyb_{\delta, \chi, \beta, 3}$ Swap the role of $\QPRO$ queries on $(\Eval, $ with the role of $w_{\chi}$. Specifically, set
            \[
                w_{\chi} \coloneqq F_{\delta, \beta}(\chi \concat 0^{D-\delta})
            \]
            and reply to $\Sim_3^{\delta,\rand}(\Eval, h_{\delta,\beta}, \chi\concat 0^{D-\delta})$ with
            \[
                \left[\ct_{\chi\concat 0}\concat \ct_{\chi\concat 1}\right]_{\beta} \oplus F_{\delta,\beta}(\chi\concat 0^{D-\delta})
            \]
            
            \item $\Hyb_{\delta, \chi, \beta, 4}$: Replace $F_{\delta,\beta}(\chi\concat 0^{D-\delta})$ by $G_v(F_\sigma(\chi, \beta)$. Note that this modifies both the hardcoded $w_{\chi,\beta}$ and the reply to $\Sim_3^{\delta,\rand}(\Eval, h_{\delta,\beta}, \chi\concat 0^{D-\delta})$.
        \end{itemize}

        Using these hybrids, we show \cref{eq:jllw-depth+1} via the transitions
        \begin{align*}
            \Hyb_{\delta, \chi, sr} 
            &\approx_{2^{-\ell}\negl(\lambda)} \Hyb_{\delta, \chi, \diff{1, 1}}
            \approx_{2^{-\ell}\negl(\lambda)} \dots \approx_{2^{-\ell}\negl(\lambda)} \Hyb_{\delta, \chi, 1, \diff{5}}
            \\
            &\dots
            \\
            &\approx_{2^{-\ell}\negl(\lambda)} \Hyb_{\delta, \chi, \diff{B}, 1}
            \approx_{2^{-\ell}\negl(\lambda)} \dots \approx_{2^{-\ell}\negl(\lambda)} \Hyb_{\delta, \chi, B, 5}
            \\
            &\approx_{2^{-\ell}\negl(\lambda)} \Hyb_{\delta, \chi+1}
        \end{align*}

        \begin{claim}[Subclaim of \cref{claim:jllw-claim-depth-transition}]
            If $\FE$ is $2^{|\delta|}$-subexponentially adaptively secure, then
            \[
                \Hyb_{\delta, \chi, sr}
                \approx_{2^{-\ell}\negl(\secpar)}
                \Hyb_{\delta, \chi, 1, 1}
            \]
        \end{claim}
        \begin{proof}
            The only difference between these two hybrids is $\ct_\chi$. Specifically, it is an encryption of $(\flag_\chi, \chi, \info_\chi)$ where\footnote{Dependence on $\beta = 1$ highlighted.} 
            \begin{align*}
                \flag_\chi &\coloneqq \begin{cases}
                    \mathsf{normal} & \text{in } \Hyb_{\delta, \chi, sr}
                    \\
                    \mathsf{hyb} & \text{in } \Hyb_{\delta, \chi, 1, 1}
                \end{cases}
                \\
                \info_\chi &\coloneqq \begin{cases}
                    \left(C, \{k_{d,b}\}_{\substack{|\chi|\leq d\leq D,\\ 1\leq b\leq B}}, s_\chi\right) & \text{in } \Hyb_{\delta, \chi, sr}
                    \\
                    \left(C,\  
                    \{k_{d,b}\}_{\substack{\delta<d<D\\ 1\leq j \leq B}},\  
                    s_\chi,\  
                    \diff{1},\  
                    \{F_{\sigma}(\chi, b)\}_{1\leq b \leq \diff{1}},\ 
                    w_{\chi,\diff{1}},\ 
                    \{k_{\delta, b}\}_{\diff{1} < j \leq B}\} \right) & \text{in } \Hyb_{\delta,\chi, 1, 1}
                \end{cases}
            \end{align*}
            $\sk_\delta$ is a functional encryption key for $\Expand_\delta$. Observe that for both settings of $(\flag_\chi, \chi, \info_\chi)$ described above,
            \[
                \Expand_\delta(\flag_\chi, \chi, \info_\chi) 
                = 
                \bigg(\ct_{\chi\concat 0} \concat \ct_{\chi\concat 1}\bigg) \oplus \bigg(f_{k_{\delta,1}}(\chi\concat 0^{D-\delta}) \concat \dots \concat f_{k_{\delta,1}}(\chi\concat 0^{D-\delta}) \bigg)
            \]
            To see this, recall that in hybrid mode $\Expand_\delta$ computes all blocks the same, except for block $\beta=1$ in this case. In that position, it outputs the precomputed $w_{\chi,1}$, which matches the normal mode computation for that block.

            Therefore the claim follows from the $2^{-\ell}$-subexponential adaptive security of $\FE$.
        \end{proof}

        \begin{claim}[Subclaim of \cref{claim:jllw-claim-depth-transition}]
            If $f$ and $G_v$ are $2^{|\delta|}$-subexponentially secure then for all $\beta\in [1,B]$,
            \[
                \Hyb_{\delta, \chi, \beta, 1}
                \approx_{2^{-\ell}\negl(\secpar)} \Hyb_{\delta, \chi, \beta, 2}
                = \Hyb_{\delta, \chi, \beta, 3}
                \approx_{2^{-\ell}\negl(\secpar)} \Hyb_{\delta, \chi, \beta, 4}
            \]
        \end{claim}
        \begin{proof}
            The first transition follows from the security of $f$ as a pseudorandom function. The second follows from the perfect secrecy of the one-time pad. The third follows from the security of $G_v$ as a pseudorandom generator.
        \end{proof}

        \begin{claim}[Subclaim of \cref{claim:jllw-claim-depth-transition}]
            If $\FE$ is $2^{|\delta|}$-subexponentially adaptively secure, then for all $\beta\in [1,B]$,
            \[
                \Hyb_{\delta, \chi, \beta, 5}
                \approx_{2^{-\ell}\negl(\secpar)}
                \Hyb_{\delta, \chi, \beta+1, 1}
            \]
        \end{claim}
        \begin{proof}
            The only difference between these two hybrids is $\ct_\chi$. Specifically, it is an encryption of $(\flag_\chi, \chi, \info_\chi)$ where\footnote{Dependence on $\beta = 1$ highlighted.} 
            \begin{align*}
                \flag_\chi &\coloneqq \mathsf{hyb}
                \\
                \info_\chi &\coloneqq \begin{cases}
                    \left(C,\  
                    \{k_{d,b}\}_{\substack{\delta<d<D\\ 1\leq j \leq B}},\  
                    s_\chi,\  
                    \diff{\beta},\  
                    \{F_{\sigma}(\chi, b)\}_{1\leq b \leq \diff{\beta}},\ 
                    \diff{w_{\chi,\beta}},\ 
                    \{k_{\delta, b}\}_{\diff{\beta} < j \leq B}\} \right) & \text{in } \Hyb_{\delta, \chi, \beta, 5}
                    \\
                    \left(C,\  
                    \{k_{d,b}\}_{\substack{\delta<d<D\\ 1\leq j \leq B}},\  
                    s_\chi,\  
                    \diff{\beta+1},\  
                    \{F_{\sigma}(\chi, b)\}_{1\leq b \leq \diff{\beta+1}},\ 
                    \diff{w_{\chi,\beta+1}},\ 
                    \{k_{\delta, b}\}_{\diff{\beta+1} < j \leq B}\} \right) & \text{in } \Hyb_{\delta,\chi, \beta+1, 1}
                \end{cases}
            \end{align*}
            $\sk_\delta$ is a functional encryption key for $\Expand_\delta$. Observe that for both settings of $(\flag_\chi, \chi, \info_\chi)$ described above,
            \[
                \Expand_\delta(\flag_\chi, \chi, \info_\chi) 
                = 
                \bigg(\ct_{\chi\concat 0} \concat \ct_{\chi\concat 1}\bigg) \oplus \bigg(f_{k_{\delta,1}}(\chi\concat 0^{D-\delta}) \concat \dots \concat f_{k_{\delta,1}}(\chi\concat 0^{D-\delta}) \bigg)
            \] 
            To see this, observe that the only difference in evaluation comes from evaluating blocks $\beta$ and $\beta+1$. In $\Hyb_{\delta, \chi, \beta+1,1}$, block $\beta+1$ is output as the hard-coded $w_{\chi,\beta+1}$, which is pre-computed to be the same as the evaluation of block $\beta+1$ in $\Hyb_{\delta, \chi, \beta+1,1}$. Similarly, in $\Hyb_{\delta, \chi, \beta,5}$, block $\beta$ is output as the hard-coded $w_{\chi, \beta}$, which is pre-computed to be the same as the evaluation of block $\beta$ in $\Hyb_{\delta, \chi, \beta+1,1}$.

            Therefore the claim follows from the $2^{-\ell}$-subexponential adaptive security of $\FE$.
        \end{proof}

        \begin{claim}[Subclaim of \cref{claim:jllw-claim-depth-transition}]
            If $\FE$ is $2^{|\delta|}$-subexponentially adaptively secure, then
            \[
                \Hyb_{\delta, \chi, B, 5}
                \approx_{2^{-\ell}\negl(\secpar)}
                \Hyb_{\delta, \chi+1}
            \]
        \end{claim}
        \begin{proof}
            The only difference between these hybrids is the construction of $\ct_{\chi}$ using a hybrid flag or a sim flag. Specifically, it is an encryption of $(\flag_{\chi}, \chi, \info_\chi)$ where
            \begin{align*}
                \flag_\chi &\coloneqq \begin{cases}
                    \mathsf{hyb} & \text{in } \Hyb_{\delta, \chi, B, 5}
                    \\
                    \mathsf{sim} & \text{in } \Hyb_{\delta, \chi+1}
                \end{cases}
                \\
                \info_\chi &\coloneqq \begin{cases}
                    \left(C,\  
                    \{k_{d,b}\}_{\substack{\delta<d<D\\ 1\leq j \leq B}},\  
                    s_\chi,\  
                    B,\  
                    \{F_{\sigma}(\chi, b)\}_{1\leq b \leq B},\ 
                    w_{\chi,B}\} \right) & \text{in } \Hyb_{\delta, \chi, B, 5}
                    \\
                    \{F_\sigma(\chi, b)\}_{1\leq b\leq B} & \text{in } \Hyb_{\delta, \chi+1}
                \end{cases}
            \end{align*}
            $\sk_{\delta}$ is a functional encryption key for the function $\Expand_\delta$. So, to reduce to the security of $\FE$ we only need show that $\Eval$ behaves the same on both settings of $(\flag_{\chi}, \chi, \info_\chi)$. Observe that
            \begin{align*}
                &\Expand_\delta\left(\mathsf{hyb}, \chi, \left(C,\  
                    \{k_{d,b}\}_{\substack{\delta<d<D\\ 1\leq j \leq B}},\  
                    s_\chi,\  
                    B,\  
                    \{F_{\sigma}(\chi, b)\}_{1\leq b \leq B},\ 
                    w_{\chi,B}\} \right)\right)
                \\
                &= G_v(F_\sigma(\chi, 1)) \concat \dots \concat G_{v}(F_\sigma(\chi,B))
                \\
                &= \Expand_\delta(\mathsf{sim}, \chi, \{F_\sigma(\chi, b)\}_{1\leq b\leq B})
            \end{align*}
            Thus the claim follows from the $2^{-\delta}$-subexponential adaptive security of $\FE$.
        \end{proof}     
    \end{proof}
\end{proof}

\section{NIZK Arguments of Knowledge for QMA}
\label{sec:nizk-qma}

\subsection{Definition}

\begin{definition}[NIZKAoK (PoK) for $\QMA$ in the CRS Model]
\label{def:nizkpok-qma}
Let $\cL = (\cL_{yes}, \cL_{no})$ be a $\QMA$ promise problem with relation $(Q,\alpha(\cdot),\beta(\cdot))$, where $\alpha(n) = 1-\negl(n)$ and $\beta(n) = \negl(n)$. Let $\gamma = \gamma(n)$ be a ``witness quality'' parameter, where $\gamma(n) \in [1/\poly(n),1-1/\poly(n)]$. $\Pi = (\Setup, \sP, \sV)$ is a non-interactive, zero-knowledge argument (resp. proof) of knowledge for $\cL$ in the CRS model if it has the following syntax and properties.

\noindent \textbf{Syntax.}
\begin{itemize}
    \item $\crs \gets \Setup(1^\lambda)$: The classical polynomial-size circuit $\Setup$ on input $1^\lambda$ outputs a common reference string $\crs$.
    \item $\pi \gets \sP(1^\lambda,\crs, x, \ket{\psi})$: The quantum polynomial-size circuit $\sP$ on input the security parameter $1^\secp$, a common random string $\crs$, an instance $x \in \{0,1\}^n$ for $n \in [\secp,\poly(\secp)]$, and a witness $\ket{\psi}$, outputs a (potentially quantum) proof $\pi$.
    \item $\sV(1^\lambda, \gamma, \crs, x, \pi) \in \zo$: The quantum polynomial-size circuit $\sV$ on input the security parameter $1^\secp$, the quality parameter $\gamma$, a common reference string $\crs$, an instance $x$, and a proof $\pi$, outputs a bit.
\end{itemize}

\noindent \textbf{Properties.}
\begin{itemize}
    \item {\bf Uniform Random String.}
    $\Pi$ satisfies the uniform random string property of \cref{def:nizk-np}.

    \item {\bf Completeness.}
    There exists a negligible function $\negl(\cdot)$ such that for every $\gamma(\cdot)$, every $\lambda \in \bbN$, every $n \in [\secp,\poly(\secp)]$, and every $(x, \ket{\psi}) \in R_{Q,\alpha(n)}$,
    \begin{equation*}
        \Pr_{\substack{\crs \gets \Setup(1^\lambda) \\ \pi \gets \sP(1^\secp,\crs, x, \ket{\psi})}}[\sV(1^\secp, \gamma(n), \crs, x, \pi) = 1] = 1 - \negl(\lambda).
    \end{equation*}

    \item {\bf Computational (resp. Statistical) Soundness.}
    There exists a negligible function $\negl(\cdot)$ such that for every $\gamma(\cdot)$, every quantum polynomial-size (resp. unbounded) adversary $\cA$, every sufficiently large $\lambda \in \bbN$, every $n \in [\secp,\poly(\secp)]$, and every $x \in N_{Q,\beta(n)}$, 
    \begin{equation*}
        \Pr_{\substack{\crs \gets \Setup(1^\lambda) \\ \pi \gets \cA(\crs)}}[\sV(1^\secp,\gamma(n),\crs, x, \pi) = 1] =  \negl(\lambda).
    \end{equation*}

    \item {\bf Adaptive Argument (resp. Proof) of Quantum Knowledge.}
    There exists a quantum polynomial-size extractor $\Ext = (\Ext_0, \Ext_1)$, and negligible functions $\negl_0(\cdot), \negl_1(\cdot),\negl_2(\secp)$ such that:
    \begin{enumerate}
        \item For every quantum polynomial-size (resp. unbounded) distinguisher $\cD$, every sufficiently large $\lambda \in \bbN$, and every $n \in [\secp,\poly(\secp)]$,
        \begin{equation*}
            \left|\Pr_{\crs \gets \Setup(1^\lambda)}[\cD(\crs) = 1] - \Pr_{(\crs,\td) \gets \Ext_0(1^\lambda)}[\cD(\crs) = 1] \right| \le \negl_0(\lambda).
        \end{equation*}
        \item Given $\secp, x, \gamma, \crs$, $\td$, and $\pi$, define \[p_V \coloneqq \Pr[\sV(1^\secp,\gamma(n),\crs,x,\pi) = 1], \ \ \ p_E \coloneqq \Pr_{\substack{ \rho \gets \Ext_1(1^\secp,\gamma(n),\crs,\td,x,\pi)}}[(x,\rho) \in R_{Q,\gamma(n)}].\] Then, for every $\gamma(\cdot)$, every quantum polynomial-size (resp. unbounded) adversary $\cA$, and every sufficiently large $\secp \in \bbN$,
        \[\Pr_{\substack{(\crs,\td) \gets \Ext_0(1^\secp) \\ (x,\pi) \gets \cA(\crs)}}[p_E \geq p_V - \negl_1(\secp)] \geq 1 - \negl_2(\secp).\]
    \end{enumerate}

    \item {\bf Non-Adaptive Computational Zero-Knowledge.}
    There exists a quantum polynomial-size $\Sim$ and a negligible function $\negl(\cdot)$ such that for every quantum polynomial-size distinguisher $\cD$, every sufficiently large $\lambda \in \bbN$, every $n \in [\secp,\poly(\secp)]$, and every $(x, \ket{\psi}) \in R_{Q,\alpha(n)}$,
    \begin{align*}
        \left\vert \Pr_{\substack{\crs \gets \Setup(1^\lambda) \\ \pi \gets \sP(1^\secp,\crs, x, \ket{\psi})}}[\cD(\crs, x, \pi) = 1] - \Pr_{\substack{(\crs, \pi) \gets \Sim(1^\lambda)}}[\cD(\crs, x, \pi) = 1]\right\vert \le \negl(\lambda).
    \end{align*}

\end{itemize}

    Finally, we formulate a novel definition for the notion of argument / proof of quantum knowledge where the extractor can extract from \emph{already-verified} proofs.

    \begin{itemize}
        \item {\bf Adaptive Post-Verified Argument (resp. Proof) of Quantum Knowledge.}
    There exists a quantum polynomial-size extractor $\Ext = (\Ext_0, \Ext_1)$, and negligible functions $\negl_0(\cdot), \negl_1(\cdot)$ such that:
    \begin{enumerate}
        \item For every quantum polynomial-size (resp. unbounded) distinguisher $\cD$, every sufficiently large $\lambda \in \bbN$, and every $n \in [\secp,\poly(\secp)]$,
        \begin{equation*}
            \left|\Pr_{\crs \gets \Setup(1^\lambda)}[\cD(\crs) = 1] - \Pr_{(\crs,\td) \gets \Ext_0(1^\lambda)}[\cD(\crs) = 1] \right| \le \negl_0(\lambda).
        \end{equation*}
        \item For every quantum polynomial-size (resp. unbounded) adversary $\cA$, every sufficiently large $\lambda \in \bbN$, and every $n \in [\secp,\poly(\secp)]$,
        \begin{equation*}
            \Pr_{\substack{(\crs, \td) \gets \Ext_0(1^\lambda) \\ (x, \pi) \gets \cA(\crs) \\ (b, \pi') \gets \sV(1^\secp,\gamma(n),\crs, x, \pi) \\ \rho \gets \Ext_1(1^\secp,\gamma(n),\crs, \td, x, \pi')}}\left[b = 1 \wedge (x, \rho) \not\in R_{Q,\gamma(n)}\right] \le \negl_1(\lambda).
        \end{equation*}
    \end{enumerate}
    \end{itemize}
\end{definition}

\begin{remark}
    It is easy to see that the notion of \emph{post-verified} argument (resp. proof) of quantum knowledge implies the standard notion. Indeed, we can define the extractor for the standard notion to apply the honest verifier followed by the post-verified extractor. It is also straightforward to see that, for any choice of $\gamma(n) \in [1/\poly(n),1-1/\poly(n)]$, the standard notion of argument (resp. proof) of knowledge implies computational (resp. statistical) soundness.
\end{remark}

\subsection{Protocol}
\label{sec:protocol}

Let $\{\theta_{H,i},f_{H,i}\}_{i \in [N]}$ denote a ZX verifier with strong completeness (\cref{def:ZX-strong}) for some QMA language $\cL = (\cL_{yes},\cL_{no})$, where each instance $H$ is an $n$-qubit Hamiltonian.

Let $\CSA = (\Gen, \Enc, \Dec, \Ver)$ be a publicly-verifiable CSA (\cref{def:CSA}). 

Let $\cO = (\Setup, \Obf, \Eval, \Ver)$ be a provably-correct obfuscation scheme (\cref{def:provable-obfuscation}).

Let $G: \{0,1\}^\secp \to \{0,1\}^{\poly(\secp)}$ be a pseudorandom generator with arbitrary polynomial stretch.

\vspace{0.5cm}

$\Setup(1^{\secpar})$:
\begin{enumerate}
    \item Generate the public parameters for the provably-correct obfuscation scheme $\Obf$. Formally, compute $\pp \gets \cO.\Setup(1^{\secp_\ZK})$, where $\secp_\ZK = \secp$, and output $\crs = \pp$.
\end{enumerate}

\vspace{0.3cm}

$\sP(1^\secp,\crs, H, \ket{\psi})$:
\begin{enumerate}
    \item Encode the witness using the CSA scheme. Formally,
  
    \begin{enumerate}
        \item Sample a key $k \gets \CSA.\KeyGen(1^{d(\secp^2,m)},1^m)$, where $m$ is the number of qubits in $\ket{\psi}$, and $d(\cdot,\cdot)$ is maximum of the functions specified in the statements of \cref{lemma:csa-privacy} and \cref{thm:auth-security}.
        \item Encode the witness as $\ket{\phi} = \CSA.\Enc_k(\ket{\psi})$.
    \end{enumerate}

    \item \label{step:M} Define a classical circuit $\cM_k$ which outputs a decoding of the witness. Formally,
    \begin{enumerate}
        \item Define $\cM_k$ hardwired with $k$ which, on input ($r$, $s$):
        \begin{enumerate}
            \item Compute $i = G(r)$, and set $(\theta,f) = (\theta_{H,i},f_{H,i})$. 
            \item Output $\CSA.\Dec_{k, \theta, \overline{f}}(s)$, where $\overline{f} = 1 - f$.
        \end{enumerate}
    \end{enumerate}

    \item Obfuscate the CSA verifier and classical circuit $\cM[k]$. Formally,
    \begin{enumerate}
        
        \item Define predicate $\pred$ as $\pred(C) = 1$ iff there exists $k'$ such that $C = \CSA.\Ver_{k', \bullet} \Vert \cM_{k'}$.

        \item Compute $\ObfC \gets \cO.\Obf(1^{q(\secp^2,m)},\pp, \pred, \CSA.\Ver_{k, \bullet} \Vert \cM_k)$, where $q(\cdot,\cdot)$ is the maximum of $\secp^{2/\epsilon}$ and the functions specified in the statements of \cref{lemma:csa-privacy} and \cref{thm:auth-security}. Here, $\epsilon$ is the constant from the Simulated-Circuit Indistinguishability property of the provable-obfuscation (\cref{def:provable-obfuscation}).
    \end{enumerate}
    
    \item Output $\pi = (\ket{\phi}, \ObfC)$.
\end{enumerate}

\vspace{0.3cm}

$\sV(1^\secp,\gamma,\crs, H, \pi)$:
\begin{enumerate}
    \item Parse all inputs. Formally,
    \begin{enumerate}
        \item Parse $\crs = \pp$.
        \item Parse $\pi = (\rho, \ObfC)$, define $\ObfV = \ObfC(0, \bullet, \bullet)$, and define $\ObfM = \ObfC(1, \bullet, \bullet)$.
    \end{enumerate}
    \item Verify that the provable obfuscation's verifier accepts. Formally,
    \begin{enumerate}
        \item Define predicate $\pred$ as $\pred(C) = 1$ iff there exists $k'$ such that $C = \CSA.\Ver_{k', \bullet} \Vert \cM_{k'}$.
        \item Verify that $\cO.\Ver(\pp, \pred, \ObfC) = 1$.
    \end{enumerate}
    \item Verify that the witness was encoded correctly using the obfuscated CSA verifier and check that the obfuscated $\cM$ accepts. Formally,
    \begin{enumerate}
        \item Define a POVM $(\cP_1, \cP_0)$ where $\cP_1 = \frac{1}{2^\secp}\sum_{r \in \zo^\lambda} P_r$, $\cP_0 = \frac{1}{2^\secp}\sum_{r \in \zo^\lambda} (\identity - P_r)$, and $P_r$ applied to state $\rho$ performs the following checks:
        \begin{enumerate}
            \item $(1, \rho') = \cO.\Eval(\ObfV, (0^m, \rho))$,
            \item $(1, \Had^{1^m}(\rho'')) = \cO.\Eval(\ObfV, (1^m, \Had^{1^m}(\rho')))$, and
            \item $(1, \rho''') = \identity - \ObfM(r, \Had^{\theta_{H,G(r)}}(\rho''))$.
        \end{enumerate}
        \item Let $\ATI_{1 - \gamma'}$ (\cref{lem:ati}) be defined according to the POVM $(\cP_1, \cP_0)$ for $\gamma' \in [1/\poly(n), 1-1/\poly(n)]$ such that there is some polynomial $p(\cdot)$ where $\gamma' = \gamma + \frac{1}{p(n)}$. 
        \item Compute $(b, \rho^*_b) \gets \ATI_{1 - \gamma'}(\rho)$.
    \end{enumerate}

    \item Reconstruct the proof. Formally,
    \begin{enumerate}
        \item Define $\pi' = (\rho^*_b, \ObfC)$.
    \end{enumerate}

    \item Output $(b, \pi')$.
\end{enumerate}

\subsection{Analysis}

\begin{theorem}\label{thm:AoK}
    Given that
    \begin{itemize}
        \item $\{\theta_{H,i},f_{H,i}\}_{i \in [N]}$ is a ZX verifier with strong completeness (\cref{def:ZX-strong}) for QMA language $\cL = (\cL_{yes},\cL_{no})$,
        \item $\CSA$ is a publicly-verifiable CSA (\cref{def:CSA}),
        \item $\cO = (\Setup, \Obf, \Eval, \Ver)$ is a sub-exponentially secure provably-correct obfuscator (\cref{def:provable-obfuscation}) with computational (statistical) knowledge soundness, and
        \item $G$ is a pseudorandom generator,
    \end{itemize} 
    then, for any $\gamma \in [1/\poly(n), 1-1/\poly(n)]$, the construction in \cref{sec:protocol} is a non-interactive zero-knowledge argument (proof) of knowledge for language $\cL$ (\cref{def:nizkpok-qma}) that satisfies
    \begin{itemize}
        \item Completeness,
        \item Adaptive post-verified argument (proof) of quantum knowledge,
        \item Non-adaptive computational zero-knowledge.
    \end{itemize}
    If the obfuscator is in the URS model, then the NIZK argument of knowledge for QMA is in the URS model.
\end{theorem}

\begin{remark}
    Since every promise problem in QMA has a ZX verifier with strong completeness (\cref{thm:ZX-strong}), the above theorem gives an argument of quantum knowledge for all of QMA.
\end{remark}

\begin{proof}
    \textbf{Correctness.}
    This follows from the correctness of $\CSA$ (\cref{def:coset-auth-correctness}), completeness of ZX verifier with strong completeness (\cref{thm:ZX-strong}), functionality-preservation and completeness of provably-correct obfuscation (\cref{def:provable-obfuscation}), and correctness of $\ATI_{1 - \gamma'}$ (\cref{lem:ati}).

    \textbf{Adaptive Post-Verified Argument (Proof) of Quantum Knowledge.}

    Let $(\cO.\Ext_0, \cO.\Ext_1)$ be the proof of knowledge extractor of $\cO$.
    We define $\Ext_0$ with oracle access to $\cO.\Ext_0$ as follows:
    \begin{enumerate}
        \item[] \hspace{-1cm} \emph{Input}: $1^\secpar$

        \item Compute $(\pp, \td) \gets \cO.\Ext_0(1^\lambda)$.

        \item Output $\crs = \pp$ and $\td$.
    \end{enumerate}
    We define $\Ext_1$ with oracle access to $\cO.\Ext_1$ as follows:
    \begin{enumerate}
        \item[] \hspace{-1cm} \emph{Input}: $\gamma$, $\crs = \pp$ and $\td$, $H$, $\pi^* = (\rho^*, \ObfC)$.

        \item Define predicate $\pred$ as $\pred(C) = 1$ iff there exists $k'$ such that $C = \CSA.\Ver_{k', \bullet} \Vert \cM_{k'}$ for $\cM$ defined in \cref{step:M}.
    
        \item Compute $C \gets \cO.\Ext_1(\pp, \td, \pred, \ObfC)$.
        
        \item Parse $C = \CSA.\Ver_{k', \bullet} \Vert \cM_{k'}$.

        \item \label{step:ExtVerComp} Compute $(\_, \rho') \gets \CSA.\Ver_{k', 0^n}(\rho^*)$.
        
        \item \label{step:ExtVerHad} Compute $(\_, \Had^{1^n}(\rho'')) \gets \CSA.\Ver_{k', 1^n}(\Had^{1^n}(\rho'))$.

        \item Compute $\rho_\psi = \CSA.\Enc_{k'}^\dagger(\rho'')$. 
    
        \item Output $\rho_\psi$.
    \end{enumerate}

    The output of $\Setup$ and $\Ext_0$ is computationally indistinguishable by the computational indistinguishability of $\cO.\Setup$ and $\cO.\Ext_0$ from the knowledge soundness of provably-correct obfuscation (\cref{def:provable-obfuscation}).

    Let a polynomial $p(\cdot)$ and a polynomial-size quantum circuit $\cA$ be given such that for every sufficiently large $\secpar \in \bbN$,
    \begin{equation}
        \label{eq:aok-qma}
        \Pr_{\substack{(\crs, \td) \gets \Ext_0(1^\lambda) \\ (H, \pi) \gets \cA(\crs) \\ (b, \pi') \gets \sV(\gamma, \crs, H, \pi) \\ \rho_\psi \gets \Ext_1(\gamma, \crs, \td, H, \pi')}}\left[b = 1 \wedge (H, \rho_\psi) \not\in \cR_{Q,\gamma}\right] \ge \frac{1}{p(\lambda)}.
    \end{equation}

    By the knowledge soundness of provably-correct obfuscation (\cref{def:provable-obfuscation}), we have that there exists a negligible function $\negl(\cdot)$
    \begin{equation}
        \label{eq:ks-qma}
        \Pr_{\substack{(\crs,\td) \gets \Ext_0(1^\lambda) \\ (H, \pi) \gets \cA(\crs) \\ (b, \pi') \gets \sV(\gamma, \crs, H, \pi) \\ C \gets \cO.\Ext_1(\pp, \td,\pred,\ObfC)}}\left[\begin{array}{l}\left(b = 0\right) \vee \\\left(\forall x, C(x) = \Eval(\ObfC,x) \wedge \pred(C) = 1\right)\end{array}\right] \ge 1-\negl(\lambda).
    \end{equation}

    Hence, by \cref{eq:aok-qma} and \cref{eq:ks-qma}, we have that there exists a polynomial $p'(\cdot)$ such that
    \begin{equation}
        \label{eq:3-qma}
        \Pr_{\substack{(\crs,\td) \gets \Ext_0(1^\lambda) \\ (H, \pi) \gets \cA(\crs) \\ (b, \pi') \gets \sV(\gamma, \crs, H, \pi) \\ \rho_\psi \gets \Ext_1(\gamma, \crs, \td, H, \pi')}}\left[\begin{array}{l}b = 1 \wedge \left(\forall x, C(x) = \Eval(\ObfC,x)\right) \wedge \\ \pred(C) = 1 \wedge (H, \rho_\psi) \not\in R_{Q,\gamma} \end{array}\right] \ge \frac{1}{p'(\lambda)}.
    \end{equation}

    Let the variables sampled according to \cref{eq:3-qma} be given. When $\pred(C) = 1$, this implies that there exists $k'$ such that $C = \CSA.\Ver_{k', \bullet} \Vert \cM_{k'}$ (by definition of $\pred$). Additionally, when $\forall x, C(x) = \Eval(\ObfC,x)$ and $b = 1$, this means that $(\rho^*, 1) = \ATI_{1 - \gamma'}(\rho)$ for POVM $(\cP_1',  \cP_0')$ where $\cP_1' = \frac{1}{N}\sum_r P_r'$, $\cP_0' = \frac{1}{N}\sum_r (\identity - P_r')$, and $P_r'$ applied to state $\rho$ performs the following checks:
    \begin{enumerate}
        \item $(1, \rho') = \CSA.\Ver_{k', 0^n}(\rho)$,
        \item $(1, \Had^{1^n}(\rho'')) = \CSA.\Ver_{k', 1^n}(\Had^{1^n}(\rho'))$, and
        \item $(1, \rho''') = (\identity - \CSA.\Dec_{k',\theta_r,\overline{f_r}}(\Had^{\theta_r}(\cdot)))$ where $(\theta_r, f_r) \defeq \Samp(H; r)$ for all $r \in [N]$.
    \end{enumerate}
    Therefore,
    \begin{equation}
        \label{eq:4-qma}
        \Pr_{\substack{(\crs,\td) \gets \Ext_0(1^\lambda) \\ (H, \pi) \gets \cA(\crs) \\ (b, \pi'=\rho^*_b) \gets \sV(\gamma, \crs, H, \pi) \\ \rho_\psi \gets \Ext_1(\gamma, \crs, \td, H, \pi')}}\left[\begin{array}{l} b = 1 \wedge (H, \rho_\psi) \not\in R_{Q,\gamma} \end{array}\right] \ge \frac{1}{p'(\lambda)}.
    \end{equation}
    
    Let the variables sampled according to \cref{eq:4-qma} be given.
    By the soundness of $\ATI_{1 - \gamma'}$ (\cref{lem:ati}), if $b=1$ when running $\ATI_{1 - \gamma'}$ with POVM $(\cP_1',  \cP_0')$ defined previously, we have that $\Tr[\cP_1' \rho^*_b] \ge \gamma'(n)$ with overwhelming probability. That is, there exists a polynomial $p''(\cdot)$ such that 
    
    \begin{equation}
        \label{eq:5-qma}
        \Pr_{\substack{(\crs,\td) \gets \Ext_0(1^\lambda) \\ (H, \pi) \gets \cA(\crs) \\ (b, \pi') \gets \sV(\gamma, \crs, H, \pi) \\ \rho_\psi \gets \Ext_1(\gamma, \crs, \td, H, \pi')}}\left[\begin{array}{l} \left(\Tr[\cP_1' \rho^*_b] \ge \gamma'(n)\right) \wedge (H, \rho_\psi) \not\in R_{Q,\gamma} \end{array}\right] \ge \frac{1}{p''(\lambda)}.
    \end{equation}

    Let the variables sampled according to \cref{eq:5-qma} be given. We note that $\Ext_1$ performs \cref{step:ExtVerComp} followed by \cref{step:ExtVerHad} as in the first two steps of the POVM $(\cP_1',  \cP_0')$. Hence, we use the same variables for comparison in the following analysis: 
    \begin{align*}
        \Tr[\cP_1'\rho^*_b] &= \frac{1}{N}\sum_r \Pr_{\substack{(b', \rho') \gets \CSA.\Ver_{k', 0^n}(\rho^*_b) \\ (b'', \Had^{1^n}(\rho'')) \gets \CSA.\Ver_{k', 1^n}(\Had^{1^n}(\rho')) \\ (b''', \rho''') \gets (\identity - \CSA.\Dec_{k',\theta_r,\overline{f_r}}(\Had^{\theta_r}(\rho''))))}}[b' = b'' = 1 \wedge b''' = 1].
    \end{align*}
    When the above event occurs, $(1, \rho') = \Ver_{k', 0^n}(\rho^*_b)$ and $(1, \Had^{1^n}(\rho'')) = \Ver_{k', 1^n}(\Had^{1^n}(\rho'))$, then by $\CSA$ soundness (\cref{lem:csa-soundness}) we have that $\rho'' \in \Enc_{k'}$. Hence, using the above argument we have 
    \begin{align*}
        \Tr[\cP_1'\rho^*_b] &\le \frac{1}{N}\sum_r \Pr_{\substack{(b', \rho') \gets \CSA.\Ver_{k', 0^n}(\rho^*_b) \\ (b'', \Had^{1^n}(\rho'')) \gets \CSA.\Ver_{k', 1^n}(\Had^{1^n}(\rho')) \\ (b''', \rho''') \gets (\identity - \CSA.\Dec_{k',\theta_r,\overline{f_r}}(\Had^{\theta_r}(\rho''))))}}[\rho'' \in \Enc_{k'} \wedge b''' = 1].
    \end{align*}
    Now, since $\rho'' \in \Enc_{k'}$ and $\rho_\psi = \Enc_{k'}^\dagger(\rho'')$ (definition of $\Ext_0$), by the correctness of $\CSA$ (\cref{def:coset-auth-correctness}), 
    \begin{align*}
        &\Tr[\cP_1'\rho^*_b] \\
        &\le \frac{1}{N}\sum_r \Pr_{\substack{(b', \rho') \gets \CSA.\Ver_{k', 0^n}(\rho^*_b) \\ (b'', \Had^{1^n}(\rho'')) \gets \CSA.\Ver_{k', 1^n}(\Had^{1^n}(\rho')) \\ (b''', \rho''') \gets (\identity - \CSA.\Dec_{k',\theta_r,\overline{f_r}}(\Had^{\theta_r}(\rho''))))}}[\rho'' \in \Enc_{k'} \wedge b''' = 1]\\
        &= \frac{1}{N}\sum_r \Pr_{\substack{(b', \rho') \gets \CSA.\Ver_{k', 0^n}(\rho^*_b) \\ (b'', \Had^{1^n}(\rho'')) \gets \CSA.\Ver_{k', 1^n}(\Had^{1^n}(\rho'))}}[\rho'' \in \Enc_{k'} \wedge \CSA.\Dec_{k',\theta_r,\overline{f_r}}(\Had^{\theta_r}(\rho'')) = 0]\\
        &= \frac{1}{N}\sum_r \Pr_{\substack{(b', \rho') \gets \CSA.\Ver_{k', 0^n}(\rho^*_b) \\ (b'', \Had^{1^n}(\rho'')) \gets \CSA.\Ver_{k', 1^n}(\Had^{1^n}(\rho'))}}[\rho'' \in \Enc_{k'} \wedge \CSA.\Dec_{k',\theta_r,f_r}(\Had^{\theta_r}(\rho'')) = 1]\\
        &\le \frac{1}{N}\sum_r \Pr_{\substack{(b', \rho') \gets \CSA.\Ver_{k', 0^n}(\rho^*_b) \\ (b'', \Had^{1^n}(\rho'')) \gets \CSA.\Ver_{k', 1^n}(\Had^{1^n}(\rho')) \\ \rho_\psi = \Enc_{k'}^\dagger(\rho'')}}[\CSA.\Dec_{k',\theta_r,f_r}(\Had^{\theta_r}(\Enc_{k'}(\rho_\psi)) = 1]\\
        &= \frac{1}{N}\sum_r \Pr_{\substack{(b', \rho') \gets \CSA.\Ver_{k', 0^n}(\rho^*_b) \\ (b'', \Had^{1^n}(\rho'')) \gets \CSA.\Ver_{k', 1^n}(\Had^{1^n}(\rho')) \\ \rho_\psi = \Enc_{k'}^\dagger(\rho'')}}[M[\theta_r, f_r](\rho_\psi) \text{ accepts}]\\
        &= \bbE_{\substack{r \gets \zo^\lambda \\ i \gets G(r) \\ (b', \rho') \gets \CSA.\Ver_{k', 0^n}(\rho^*_b) \\ (b'', \Had^{1^n}(\rho'')) \gets \CSA.\Ver_{k', 1^n}(\Had^{1^n}(\rho')) \\ \rho_\psi = \Enc_{k'}^\dagger(\rho'')}} \Big[ \Tr(M[\theta_{H,\lambda,i},f_{H,\lambda,i}]\rho_\psi) \Big].
    \end{align*}
    Hence, using the above argument with \cref{eq:5-qma}, we have that
    \begin{equation}
        \label{eq:almost-penultimate-qma}
        \Pr_{\substack{(\crs,\td) \gets \Ext_0(1^\lambda) \\ (H, \pi) \gets \cA(\crs) \\ (b, \pi') \gets \sV(\gamma, \crs, H, \pi) \\ \rho_\psi \gets \Ext_1(\gamma, \crs, \td, H, \pi')}}\left[\begin{array}{l} \left(\bbE_{\substack{r \gets \zo^\lambda \\ i \gets G(r)}} \Big[ \Tr(M[\theta_{H,\lambda,i},f_{H,\lambda,i}]\rho_\psi) \Big] \ge \gamma'(n)\right) \wedge \\ (H, \rho_\psi) \not\in R_{Q,\gamma} \end{array}\right] \ge \frac{1}{p''(\lambda)}.
    \end{equation}
    We will now argue that we can switch a pseudorandom choice of measurement to a uniformly random choice of measurement.

    \begin{claim}
        \label{claim:meas_indisting}
        There exists a negligible function $\negl(\cdot)$ such that for all $\rho \in D(\cB^{\otimes p(|x|)})$, 
        \begin{align*}
            \left| \bbE_{\substack{r \gets \zo^\lambda \\ i \gets G(r)}} \Big[ \Tr(M[\theta_{H,\lambda,i},f_{H,\lambda,i}]\rho) \Big] - \bbE_{i \gets [\poly(\lambda)]} \Big[ \Tr(M[\theta_{H,\lambda,i},f_{H,\lambda,i}]\rho) \Big] \right| \le \negl(\lambda).
        \end{align*}
    \end{claim}
    
    \begin{proof}
        This follows from the security of the pseudorandom generator $G$ against non-uniform adversaries and the observation that the measurements $\{M[\theta_{H,\lambda,i},f_{H,\lambda,i}]\}_{i \in [\poly(\lambda)]}$ are efficiently computable.
    \end{proof}

    Therefore, by \cref{eq:almost-penultimate-qma} and \cref{claim:meas_indisting}, there exists a negligible function $\negl(\cdot)$ such that
    \begin{equation}
        \Pr_{\substack{(\crs,\td) \gets \Ext_0(1^\lambda) \\ (H, \pi) \gets \cA(\crs) \\ (b, \pi') \gets \sV(\gamma, \crs, H, \pi) \\ \rho_\psi \gets \Ext_1(\gamma, \crs, \td, H, \pi')}}\left[\begin{array}{l} \left(\bbE_{i \gets [\poly(\lambda)]} \Big[ \Tr(M[\theta_{H,\lambda,i},f_{H,\lambda,i}]\rho_\psi) \Big] \ge \gamma'(n) - \negl(\lambda)\right) \wedge \\ (H, \rho_\psi) \not\in R_{Q,\gamma} \end{array}\right] \ge \frac{1}{p''(\lambda)},
    \end{equation}
    and, by the definition of $\gamma'$, this implies that
    \begin{equation}
        \label{eq:penultimate-qma}
        \Pr_{\substack{(\crs,\td) \gets \Ext_0(1^\lambda) \\ (H, \pi) \gets \cA(\crs) \\ (b, \pi') \gets \sV(\gamma, \crs, H, \pi) \\ \rho_\psi \gets \Ext_1(\gamma, \crs, \td, H, \pi')}}\left[\begin{array}{l} \left(\bbE_{i \gets [\poly(\lambda)]} \Big[ \Tr(M[\theta_{H,\lambda,i},f_{H,\lambda,i}]\rho_\psi) \Big] \ge \gamma(n)\right) \wedge \\ (H, \rho_\psi) \not\in R_{Q,\gamma} \end{array}\right] \ge \frac{1}{p''(\lambda)}.
    \end{equation}

    By the definition of $\cR_{Q,\gamma}$ as the $\gamma(n)$-relation (\cref{def:qma-rel}) with respect to the ZX verifier with strong completeness (\cref{def:ZX-strong}),
    \begin{equation}
        \label{eq:end-qma}
        \Pr_{\substack{(\crs,\td) \gets \Ext_0(1^\lambda) \\ (H, \pi) \gets \cA(\crs) \\ (b, \pi') \gets \sV(\gamma, \crs, H, \pi) \\ \rho_\psi \gets \Ext_1(\gamma, \crs, \td, H, \pi')}}\left[\begin{array}{l}(H, \rho_\psi) \in R_{Q,\gamma}  \wedge (H, \rho_\psi) \not\in R_{Q,\gamma} \end{array}\right] \ge \frac{1}{p''(\lambda)}.
    \end{equation}
    
    Since this is a contradiction, we have proven knowledge soundness.

    \textbf{Computational Zero-Knowledge.}

    Let $\cO.\Sim = (\cO.\SimGen, \cO.\SimObf)$ be the simulator from the simulation security of the provably-correct obfuscation $\cO$.
    We define $\Sim_0$ with oracle access to $\cO.\SimGen$ as follows:
    \begin{enumerate}
        \item[] \hspace{-1cm} \emph{Input}: $1^\secpar$

        \item Compute $(\pp, \td) \gets \cO.\SimGen(1^\lambda)$.

        \item Output $\crs = \pp$ and $\td$.
    \end{enumerate}
    We define $\Sim_1$ with oracle access to $\cO.\SimObf$ as follows:
    \begin{enumerate}
        \item[] \hspace{-1cm} \emph{Input}: $\crs$, $\td$, $H$
        
        \item Sample a key $k \gets \CSA.\KeyGen(1^{d(\secp^2,m)}, 1^m)$.

        \item Encode dummy witness as $\ket{\phi} = \CSA.\Enc_k(\ket{0^m})$.

        \item Define predicate $\pred$ as $\pred(C) = 1$ iff there exists $k'$ such that $C = \CSA.\Ver_{k', \bullet} \Vert \cM_{k'}$ for $\cM$ defined in \cref{step:M}.

        \item Compute $\ObfC \gets \cO.\SimObf(1^{q(\secp^2,m)},\pp, \td, \pred, \CSA.\Ver_{k, \bullet} \Vert C_{\mathsf{null}})$.
        
        \item Output $\pi = (\ket{\phi}, \ObfC)$.
    \end{enumerate}

    Let a polynomial-size quantum circuit $\cD$, sufficiently large $\secpar \in \bbN$, and $(H, \ket{\psi}) \in R_{Q,1-\negl(n)}$ be given. We construct the following series of hybrids to argue $\negl(\secp)$ indistinguishability of $b$ from the honest distribution $\cH_0$ and simulated distribution $\cH_4$:

    \paragraph{$\cH_0$ :}
    Honest protocol:
    $\crs \gets \Setup(1^\secpar)$. $\pi \gets \sP(1^\secp,\crs, H, \ket{\psi})$. $b \gets \cD(\crs, H, \pi)$.

    \paragraph{$\cH_1$ :}
    Same as $\cH_0$ except that:
    \begin{itemize}
        \item Compute $(\pp, \td) \gets \cO.\SimGen(1^\secpar)$ and set $\crs = \pp$.
        \item Compute $\ObfC \gets \cO.\SimObf(1^{q(\secp^2,m)},\pp, \td, \pred, \CSA.\Ver_{k, \bullet} \Vert \cM_k)$.
    \end{itemize}

    \paragraph{$\cH_2$ :}
    Same as $\cH_1$ except that:
    \begin{itemize}
        \item Compute $\ObfC \gets \cO.\SimObf(1^{q(\secp^2,m)},\pp, \td, \pred, \CSA.\Ver_{k, \bullet} \Vert C_{\mathsf{null}})$.
    \end{itemize}

    \paragraph{$\cH_3$ :}
    Same as $\cH_2$ except that:
    \begin{itemize}
        \item Compute $\ket{\phi} = \CSA.\Enc_k(\ket{0})$.
    \end{itemize}

    \paragraph{$\cH_4$ :}
    Simulated protocol:
    $(\crs, \td) \gets \Sim_0(1^\secpar)$. $\pi \gets \Sim_1(\crs, \td, H)$. $b \gets \cD(\crs, H, \pi)$.

    $\cH_0$ and $\cH_1$ are $\negl(\secp)$ indistinguishable by the honest-to-simulated indistinguishability property of provably-correct obfuscation (\cref{def:provable-obfuscation}). $\cH_2$ and $\cH_3$ are $\negl(\secp)$ indistinguishable by the Hiding property of CSA scheme (\cref{lemma:csa-privacy}). $\cH_3$ and $\cH_4$ are identical by definition of $(\Sim_0, \Sim_1)$.

    All that remains to prove is that $\cH_1$ and $\cH_2$ are $\negl(\secp)$ indistinguishable. We will show this for fixed randomness via a series of hybrids, then combine them using the evasive composability property of provably-correct obfuscation (\cref{def:provable-obfuscation}).

    \begin{claim}
        \label{claim:nizk-if-evasive-comp}
        Let $\cS$ prepare $(H, \ket{\psi})$, sample $k \gets \CSA.\KeyGen(1^{d(\secp^2,m)}, 1^m)$, compute $\ket{\phi} \gets \CSA.\Enc_k(\ket{\psi})$, and output $(\ket{\phi}, \{\cM_k(r, \bullet)\}_{r\in \{0,1\}^\secp})$ for $\cM$ defined in \cref{step:M}.
        For any $r^* \in \{0,1\}^\secp$, any predicate $\pred$, and any QPT adversary $\cA$, 
        \begin{align*}
            &\bigg|\Pr_{\substack{(\pp, \td) \gets \Sim_0(1^\lambda) \\ \left(\ket{\phi},\{\CSA.\Ver_{k, \bullet} \Vert\cM_k(r, \bullet)\}_r\right) \gets \cS}}\left[\cA\left(\ket{\phi},\cO.\SimObf\left(1^{q(\secp^2,m)}, \pp, \td, \pred,\CSA.\Ver_{k, \bullet} \Vert \cM_k(r^*, \bullet)\right)\right) = 1\right]\\ 
            &- \Pr_{\substack{(\pp, \td) \gets \Sim_0(1^\lambda) \\ \left(\ket{\phi},\{\CSA.\Ver_{k, \bullet} \Vert\cM_k(r, \bullet)\}_r\right) \gets \cS}}\left[\cA\left(\ket{\phi},\cO.\SimObf\left(1^{q(\secp^2,m)}, \pp, \td,\pred,\CSA.\Ver_{k, \bullet} \Vert C_{\mathsf{null}}\right)\right) = 1\right]\bigg| \le \frac{\negl(\secp)}{2^{\lambda}}.
        \end{align*}
    \end{claim}

    \begin{proof}
        Let $r^* \in \{0,1\}^\secp$, $\pred$, and $\cA$ be given.
        We construct the following series of hybrids:
        
        \paragraph{$\cH_{1,0}$ :}
        Same as $\cH_1$ above for fixed $r^* \in \{0,1\}^\secp$:
        \begin{itemize}
            \item $\left(\ket{\phi},\{\CSA.\Ver_{k, \bullet} \Vert\cM_k(r, \bullet)\}_r\right) \gets \cS$.
            \item $\ObfC \gets \cO.\SimObf\left(1^{q(\secp^2,m)}, \pp, \td, \pred,\CSA.\Ver_{k, \bullet} \Vert \cM_k(r^*, \bullet)\right)$.
            \item $b \gets \cA(\ket{\phi}, \ObfC)$.
        \end{itemize}

        \paragraph{$\cH_{1,1}$ :}
        Same as $\cH_{1,0}$ except:
        \begin{itemize}
            \item Replace 
            \[\ket{\psi} \to \ket*{\psi'} \coloneqq \frac{M[\theta_{H,i},f_{H,i}]\ket{\psi}}{ \|M[\theta_{H,i},f_{H,i}]\ket{\psi}\|},\]
            where $i = G(r^*)$.
        \end{itemize}

        \paragraph{$\cH_{1,2}$ :}
        Same as $\cH_{1,1}$ except:
        \begin{itemize}
            \item Replace $\overline{f}$ in $\cM_k(r^*, \bullet)$ with the zero function $f^* = 0$.
        \end{itemize}

        \paragraph{$\cH_{1,3}$ :}
        Same as $\cH_{1,2}$ except:
        \begin{itemize}
            \item Replace $\cM_k(r^*, \bullet)$ from $\cH_{1,2}$ with $C_{\mathsf{null}}$.
        \end{itemize}

        \paragraph{$\cH_{1,4}$ :}
        Same as $\cH_{1,3}$ except:
        \begin{itemize}
            \item Replace $\ket*{\psi'}$ defined in $\cH_{1,1}$ with $\ket{0^m}$.
        \end{itemize}
    
        \paragraph{$\cH_{1,5}$ :}
        Same as $\cH_2$ above for fixed $r^* \in \{0,1\}^\secp$:
        \begin{itemize}
            \item $\left(\ket{\phi},\{\CSA.\Ver_{k, \bullet} \Vert\cM_k(r, \bullet)\}_r\right) \gets \cS$.
            \item $\ObfC \gets \cO.\SimObf\left(1^{q(\secp^2,m)}, \pp, \td, \pred,\CSA.\Ver_{k, \bullet} \Vert C_{\mathsf{null}}\right)$.
            \item $b \gets \cA(\ket{\phi}, \ObfC)$.
        \end{itemize}
        
        By the strong completeness of the ZX verifier (\cref{thm:ZX-strong}), Gentle Measurement (\cref{lemma:gentle-measurement}), and the fact that $n \geq \secp$, we have that $\cH_{1,0}$ and $\cH_{1,1}$ are $2^{-3\secp/2}$-indistinguishable = $(\negl(\secp)/2^\secp)$-indistinguishable. 

        By the Measurement Indistinguishability property of CSA (\cref{thm:auth-security}), we have that $\cH_{1,1}$ and $\cH_{1,2}$ are $2^{-\Omega(\lambda^2)}$-indistinguishable = $(\negl(\secp)/2^\secp)$-indistinguishable.

        By the Simulated-Circuit Indistinguishability of the provable-obfuscation (\cref{def:provable-obfuscation}), we have that $\cH_{1,2}$ and $\cH_{1,3}$ are $2^{-\lambda^2}$-indistinguishable = $(\negl(\secp)/2^\secp)$-indistinguishable.

        By the Hiding property of CSA (\cref{lemma:csa-privacy}), we have that $\cH_{1,3}$ and $\cH_{1,4}$ are $2^{-\Omega(\lambda^2)}$-indistinguishable = $(\negl(\secp)/2^\secp)$-indistinguishable.
        
        $\cH_{1,4}$ and $\cH_{1,5}$ are identical by definition of $\cS$.
    \end{proof}

    By \cref{claim:nizk-if-evasive-comp} and the evasive composability property of the provable-obfuscation (\cref{def:provable-obfuscation}), we have that for $\cS$ (defined in \cref{claim:nizk-if-evasive-comp}), for any predicate $\pred$, and for any QPT adversary $\cA$, there exists a negligible function $\negl(\cdot)$ such that
    \begin{align*}
        &\bigg|\Pr_{\substack{(\pp, \td) \gets \Sim_0(1^\lambda) \\ \left(\ket{\phi},\{\CSA.\Ver_{k, \bullet} \Vert\cM_k(r, \bullet)\}_r\right) \gets \cS}}\left[\cA\left(\ket{\phi},\cO.\SimObf\left(1^{q(\secp^2,m)}, \pp, \td, \pred,\CSA.\Ver_{k, \bullet} \Vert \cM_k\right)\right) = 1\right]\\ 
        &- \Pr_{\substack{(\pp, \td) \gets \Sim_0(1^\lambda) \\ \left(\ket{\phi},\{\CSA.\Ver_{k, \bullet} \Vert\cM_k(r, \bullet)\}_r\right) \gets \cS}}\left[\cA\left(\ket{\phi},\cO.\SimObf\left(1^{q(\secp^2,m)}, \pp, \td,\pred,\CSA.\Ver_{k, \bullet} \Vert C_{\mathsf{null}}\right)\right) = 1\right]\bigg| \le \negl(\lambda).
    \end{align*}
    This implies that $\cH_1$ and $\cH_2$ are $\negl(\secp)$ indistinguishable, thus concluding the proof.
\end{proof}

Due to \cref{thm:provable-obfuscation}, we have the following corollaries.

\begin{corollary}
    Assuming functional encryption with subexponential security (\cref{def:FE}) and post-quantum NIZK arguments of knowledge for NP with a URS setup, there exists a NIZK argument of knowledge for QMA with a URS setup in the QPrO model (\cref{def:QPRO}).
\end{corollary}

\begin{corollary}
    Assuming indistinguishability obfuscation satisfies $\cS$-evasive composability (\cref{def:provable-obfuscation}) for $\cS$ defined in \cref{claim:nizk-if-evasive-comp} and post-quantum NIZK proof of knowledge for NP with a CRS setup, there exists a NIZK proof of knowledge for QMA with a CRS setup.
\end{corollary}

\textbf{Acknowledgements.} RJ and KT were supported in part by AFOSR, NSF 2112890, NSF CNS-2247727 and a Google
Research Scholar award. 
This material is based upon work supported by the Air Force Office of
Scientific Research under award number FA9550-23-1-0543.
We thank Dakshita Khurana for her assistance in the early stages of the project.

\bibliographystyle{alpha}
\bibliography{bib/refs,bib/abbrev3,bib/crypto}

\addcontentsline{toc}{section}{References}
\appendix

\end{document}